\documentclass[aps,
prx,
twocolumn,
superscriptaddress,
longbibliography,
nofootinbib]{revtex4-2}

\usepackage{graphicx}
\usepackage{bm}
\usepackage{physics}
\usepackage{xcolor}
\usepackage{enumitem}
\usepackage{amsmath,amsthm, amssymb}
\usepackage{mathrsfs,dsfont}
\usepackage[normalem]{ulem}
\usepackage{natbib}
\usepackage{bbm}
\usepackage{comment}
\usepackage{algorithm}
\usepackage{algpseudocode}
\usepackage[colorlinks=true]{hyperref}
\hypersetup{citecolor = blue}
\usepackage{tikz}
\usepackage{framed}

\def\equationautorefname~#1\null{Eq.~(#1)\null}

\newcommand{\appref}[1]{\hyperref[#1]{App.~\ref*{#1}}}

\newtheorem{prototheorem}{Theorem}[section]
\newtheorem{protodef}[prototheorem]{Definition}

\newtheorem{protolemma}[prototheorem]{Lemma}

\newtheorem{proposition}[prototheorem]{Proposition}

\newtheorem{fact}[prototheorem]{Fact}
\newtheorem{claim}[prototheorem]{Claim}

\newtheorem{protoconjecture}[prototheorem]{Conjecture}

\newtheorem{property}[prototheorem]{Property}

\newenvironment{theorem}{\colorlet{shadecolor}{gray!15}\begin{shaded}\begin{prototheorem}}
{\end{prototheorem}\end{shaded}}

\newenvironment{lemma}{\colorlet{shadecolor}{gray!15}\begin{shaded}\begin{protolemma}}
{\end{protolemma}\end{shaded}}

\newcommand{\N}{\mathbb{N}}

\newcommand{\id}{\mathds{1}}
\newcommand{\tmix}{t_{\rm mix}}
\newcommand{\trel}{t_{\rm rel}}
\newcommand{\transset}{\mathscr{T}}
\newcommand{\transgraph}{\mathscr{G}_{\rm T}}

\newcommand{\susk}

\begin{document}

\title{Optimal Decoding with the Worm}

\author{Zac Tobias}
\email{ztobias@stanford.edu}
\affiliation{Department of Physics, Stanford University, Stanford, California 94305, USA}
\author{Nikolas P.\ Breuckmann}
\email{niko.breuckmann@bristol.ac.uk}
\affiliation{School of Mathematics, University of Bristol, Bristol BS8 1UG, United Kingdom}
\author{Benedikt Placke}
\email{benedikt.placke@physics.ox.ac.uk}
\affiliation{Rudolf Peierls Centre for Theoretical Physics, University of Oxford, Oxford OX1 3PU, United Kingdom}

\begin{abstract}
    We propose a new decoder for ``matchable'' qLDPC codes that uses a Markov Chain Monte Carlo algorithm---called the \emph{worm algorithm}---to approximately compute the probabilities of logical error classes given a syndrome. The algorithm hence performs (approximate) \emph{optimal} decoding, and we expect it to be computationally efficient in certain settings.
    The algorithm is applicable to decoding random errors for the surface code, the honeycomb Floquet code, and hyperbolic surface codes with constant rate, in all cases with and without measurement errors.
    The efficiency of the decoder hinges on the mixing time of the underlying Markov chain. We give a rigorous mixing time guarantee in terms of a quantity that we call the \emph{defect susceptibility}. We connect this quantity to the notion of disorder operators in statistical mechanics and use this to argue (non-rigorously) that the algorithm is efficient for \emph{typical} errors in the entire decodable phase.
    We also demonstrate the effectiveness of the worm decoder numerically by applying it to the surface code with measurement errors as well as a family of hyperbolic surface codes.
    For most codes, the matchability condition restricts direct application of our decoder to noise models with independent bit-flip, phase-flip, and measurement errors. However, our decoder returns \emph{soft information} which makes it useful also in heuristic ``correlated decoding'' schemes which work beyond this simple setting. We demonstrate this by simulating decoding of the surface code under depolarizing noise, and we find that the threshold for ``correlated worm decoding'' is substantially higher than for both minimum-weight perfect matching and for correlated matching.
\end{abstract}

\maketitle

\section{Introduction}

Quantum error correction is widely believed to be indispensable for achieving practical, large-scale quantum computation~\cite{terhal2015quantum}. 
By encoding logical qubits into many physical qubits, quantum error correction enables us to suppress logical error rates exponentially as more physical qubits are added, provided that the physical error rate lies below a critical threshold~\cite{aharonov1997fault,kitaev1997quantum,kitaev1997quantum2}. 
Among the various quantum error-correcting codes proposed, the surface code stands out as one of the most promising candidates for near-term fault-tolerant architectures, owing to its high threshold, local stabilizer measurements on a two-dimensional lattice, and experimental feasibility~\cite{bravyi1998quantum,freedman2001projective,Dennis_2002,google2025below_threshold}.

A crucial component of any quantum error correction scheme is the \emph{decoder}, which is a classical algorithm that processes syndrome information extracted from stabilizer measurements and determines an appropriate recovery operation to restore the encoded quantum state.
In general, the decoding problem is computationally intractable~\cite{iyer2015hardness}.

However, for the surface code under certain noise models, efficient decoding algorithms exist.
The most widely used approach is \emph{minimum-weight perfect matching} (MWPM), which finds the most likely error consistent with the observed syndrome by mapping the problem to a weighted graph matching problem~\cite{Dennis_2002}.
While computationally efficient and practically successful \cite{Wang_2003}, MWPM is sub-optimal: it does not account for the degeneracy of the code, whereby multiple distinct error configurations may produce the same syndrome and lead to equivalent logical outcomes.

Recent work has pursued improvements to MWPM-based decoding by incorporating more realistic error models; notable examples include the Tesseract decoder~\cite{beni2025tesseract}, which uses $A^*$ search to find the most-likely error for general quantum LDPC codes, and the Hyperion decoder~\cite{wu2025minimum}, which generalizes the blossom algorithm to hypergraphs. For a comprehensive survey of decoding algorithms for surface codes, we refer the reader to~\cite{deMarti_iOlius_2024}.

What is left towards optimality in these cases is to incorporate the degeneracy. An \emph{optimal decoder}, also known as a \emph{maximum-likelihood decoder}, returns the recovery operation that maximizes the probability of successfully restoring the encoded state under the assumed noise model.
Such a decoder accounts for all errors consistent with the syndrome and selects the most probable logical equivalence class.

Full maximum-likelihood decoding of arbitrary error models can, in principle, be achieved using tensor network methods, which however are efficient only in certain restricted settings.
For the two-dimensional surface code with noiseless syndrome extraction, Bravyi, Suchara, and Vargo demonstrated efficient implementations using both matchgate simulation and matrix product states~\cite{bravyi2014efficient}, and the special case of decoding over the erasure channel in this case is possible even in linear time \cite{delfosse2020erasure_ml}.
Subsequent work by Chubb extended tensor network decoding to general 2D Pauli codes~\cite{chubb2021generaltensornetworkdecoding}, and recently showed how to generalize these techniques beyond 2D, enabling application to 3D codes and 2D codes with circuit-level noise, though the tensor contraction becomes significantly more challenging in higher dimensions~\cite{PRXQuantum.5.040303}. 
Despite these advances, maximum-likelihood decoding (MLD) of the surface code in the more realistic setting with imperfect measurements, remains an open problem.

\begin{figure}
  \centering
  \includegraphics[width=1.0\linewidth]{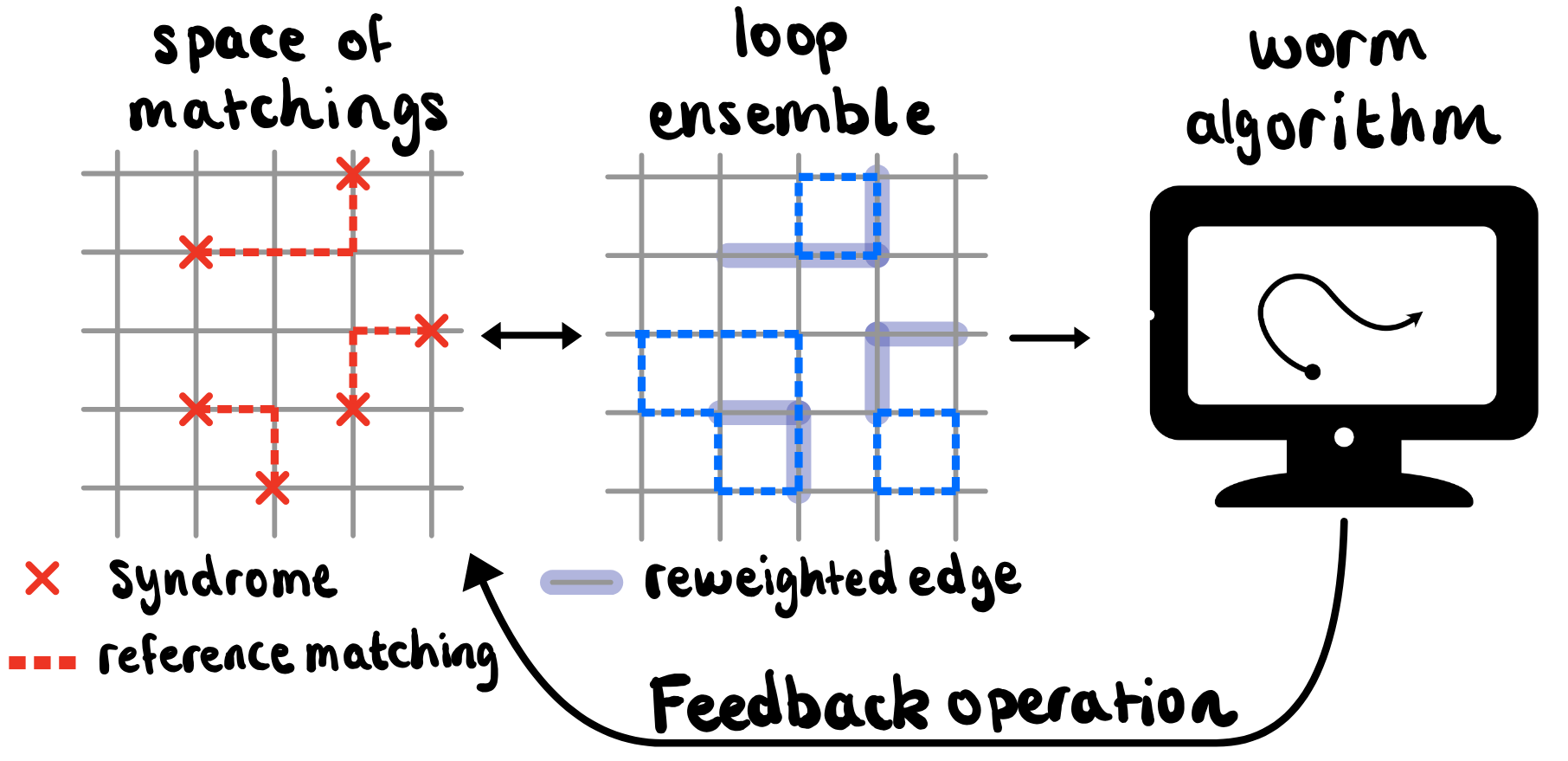}
  \caption{Schematic illustration of optimal decoding using the worm algorithm. A reference matching consistent with the observed syndrome is first identified and used to reweight the decoding graph. The worm algorithm then samples closed-loop configurations on the reweighted graph, thereby generating error chains consistent with the syndrome. These samples can be used to determine the most likely logical class of errors, enabling optimal decoding.}
  \label{fig:setup_schematic}
\end{figure}

In this work, we introduce the \emph{worm decoder} which performs (approximate) optimal decoding for any \emph{matchable} decoding problem.
By matchable, we mean a decoding problem whose detector error model~\cite{gidney2021stim} can be represented as a weighted graph in which each error mechanism triggers at most two detectors.
In such problems, detection events correspond to the boundaries of error chains, and standard graph-based decoders such as MWPM can be directly applied.
The surface code under many standard noise models, including imperfect measurements, falls within this class, but also the honeycomb Floquet code \cite{Hastings2021dynamically}, and hyperbolic surface codes with constant rate \cite{breuckmann2016constructions}.
The improvement over MWPM facilitated by the worm decoder is therefore distinct to that offered by e.g. Tesseract and Hyperion. While the codes and noise models to which it is applicable remain restricted (but include examples of interest), it effectively incorporates the degeneracy of the code.

The basic setup of the worm decoder is illustrated in \autoref{fig:setup_schematic}. At its core, it exploits a well-known Monte Carlo technique: the \emph{worm algorithm}~\cite{Broder1986worm,PROKOFEV1998253}.
The worm algorithm defines a Markov process whose stationary distribution is the Gibbs distribution over loop configurations\footnote{More precisely, it is a distribution over loops and almost-loops, as we discuss in detail below.} on a weighted graph.
We show that this property enables direct sampling from the distribution of error configurations conditioned on a given syndrome.
Sampling from this conditional distribution automatically yields a decoder by selecting a recovery operation consistent with the logical sector that arises most frequently among the samples. 
As a general result, we discuss exactly how such sampling based decoders approximate maximum likelihood decoding, showing in particular that under mild assumptions they have asymptotically the same error suppression even at finite bias and sampling error.

The key insight underlying the worm decoder is that the distribution of errors consistent with a given syndrome can be naturally reformulated as a distribution over loop configurations on an appropriately weighted graph.
Starting from the noise model encoded in the detector error model and an (arbitrary) initial reference matching consistent with the syndrome, one constructs a loop graph whose configurations correspond precisely to possible error chains.
Error chains differing by a closed loop are equivalent from the syndrome's perspective, but may belong to different logical equivalence classes if the loop is homologically non-trivial.

A critical feature of the worm algorithm is that it operates in an enlarged configuration space.
To efficiently sample closed loops, the state space is extended to include configurations with two open boundaries, i.e., the head and tail of the ``worm'' (see \autoref{fig:worm_process_schematic} for a schematic).\footnote{Although it is hard to tell which end is the head and which is the tail.}
These temporary open configurations enable non-local updates which can avoid bottlenecks and change the homology class of the loop configuration, allowing the Markov chain to freely explore all homological sectors.

The efficiency of Markov Chain Monte Carlo (MCMC) algorithms is generally limited by their mixing time, which roughly speaking sets the time the chain must be run to converge to the stationary distribution, as well as to obtain independent samples. 
It is well-appreciated that in many settings, the effectively non-local nature of the worm updates avoids ``bottlenecks'' leading to slow mixing of other MCMC algorithms such as standard Metropolis sampling with local updates. 
It underlies several provably efficient sampling schemes including the celebrated scheme by Jerrum, Sinclair, and Vigoda for approximating the permanent \cite{jerrum_sinclair_1989_permanent,jerrum2003book, jerrum_sinclair_vigoda_2004_permanent} as well as schemes for sampling the partition function of the Ising model \cite{Wang_2005,Collevecchio_2016}.
We generalize the latter results to our setting to obtain a rigorous mixing time guarantee in terms of a syndrome-dependent quantity which we call the `defect susceptibility'.
We connect this quantity to the physics of disorder operators from statistical mechanics \cite{fradkin2017disorder_ops} and use this to (non-rigorously) argue that the mixing time is indeed bounded by a polynomial in system size, for all but a vanishing fraction of errors, throughout the entire decodable phase. This means that the worm decoder has a vanishing error rate even with polynomial runtime, at any error rate below the threshold of the optimal decoder.

Beyond `just' returning a recovery operation, ability to sample from the conditional error distribution provides \emph{soft information} that is useful for further processing, such as effectively incorporating correlations \cite{fowler2013optimalcomplexitycorrectioncorrelated,pattison2021improved,bravyi2014efficient,wootton2021high,old2023generalized} or enabling post-selection-based schemes \cite{chen2025scalableaccuracygainspostselection,lee2025efficient}.

The remainder of this paper is organized as follows. 
In \autoref{sec:decoding}, we give a general introduction to decoding quantum codes, review the detector error model formalism, as well as the formal definition of the optimal decoding problem. We also argue, using a large deviation principle, that the performance of an approximate optimal decoder asymptotically approaches that of the exact optimal decoder even at finite resolution.
In \autoref{sec:worm_algorithm}, we introduce the worm algorithm and decoder, show that it performs optimal decoding for matchable codes, and comment on several ways to use the soft information provided by the worm. 
In \autoref{sec:fast_mixing}, we present our rigorous mixing-time guarantee and discuss its connection to disorder operators and their statistics.
In \autoref{sec:numerical_results}, we present numerical simulations of using the worm decoder to decode the surface code with measurement errors, and a family of hyperbolic surface codes with constant rate.
Finally, in \autoref{sec:correl_decoding} we show how to use the worm to improve decoding performance beyond just matchable codes. To this end, we introduce an iterative decoding scheme inspired by correlated matching, and benchmark it on a surface code under depolarizing noise.
We conclude in \autoref{sec:conclusion}.

\section{Decoding Quantum Codes\label{sec:decoding}}

In this section, we provide a brief introduction to the problem of decoding quantum codes applicable to a broad range of settings. This serves as an overview of our general setup and notation, and, for readers unfamiliar with this problem, also a general introduction to the topic.

In \autoref{sec:DEM}, we introduce the \emph{detector error model} (DEM) formalism, which provides a general framework for describing a wide class of error models for stabilizer codes, and is also the natural language used by simulation tools such as STIM~\cite{gidney2021stim}. 
Using this formalism, we discuss, in \autoref{sec:optimal_decoding} the process of \emph{optimal decoding}, which, for a given error model, corresponds to identifying the most likely logical equivalence class of errors consistent with an observed syndrome (or detection event).
In \autoref{sec:approximate_optimal}, we show that optimal decoding can be approximated by (approximately) sampling from the distribution of errors conditioned on the observed syndrome. In particular, we argue that even at finite bias and sample size, a sampling-based approximate optimal decoder has the same leading asymptotic error suppression as the perfect optimal decoder, when assuming the validity of a certain large deviation principle.

We then specialize, in \autoref{sec:matchable_decoding}, to \emph{matchable} decoding problems. Finally, in \autoref{sec:optimal_matchable}, we show that optimal decoding for matchable DEMs can be reformulated in terms of properties of an ensemble of \emph{cycles} on a weighted graph. This reformulation naturally suggests the use of the worm algorithm for decoding, which we introduce in \autoref{sec:worm_algorithm}.

\subsection{Detector error models}\label{sec:DEM}

The \emph{Detector Error Model} (DEM) \cite{gidney2021stim,Piveteau_2024} formalism provides an effective way to describe general noise models consisting of \emph{independent} error mechanisms, for codes that have measurement outcomes which are parity checks. Examples of such codes include stabilizer codes~\cite{gottesman1997stabilizercodesquantumerror}, Floquet codes~\cite{Hastings2021dynamically} and subsystem codes~\cite{Poulin_2005,kribs2006operatorquantumerrorcorrection}, however much of our discussion is focused around stabilizer codes.

The effectiveness of this formalism comes from focusing on the relationship between \emph{error mechanisms} and \emph{detectors}, rather than tracking faults on individual qubits or gates. Specifically, the DEM tracks how those faults affect the binary outcomes of stabilizer measurements, which can then be formulated in terms of \emph{detectors}. 

For the case of stabilizer codes without measurement errors, typically only a single round of measurements is required prior to decoding. In this case, the detector outcomes simply correspond to the stabilizer measurements. Further, if we consider a CSS code with i.i.d.\ bit-flip errors and no measurement errors, then the detectors in the DEM are simply the $Z$-stabilizer checks, and the error mechanisms correspond to $X$ errors on different qubits.

More generally, detectors describe products of measurement outcomes in a circuit that are deterministic in the absence of errors. A simple example is the same setup as before but in the presence of stabilizer measurement errors. In this case, one usually performs repeated rounds of stabilizer measurements~\cite{Dennis_2002}. This yields a space--time history of stabilizer outcomes (syndromes). In the DEM framework, a detector is said to be triggered when a stabilizer measurement differs between consecutive rounds. 

In general, the DEM defines a linear map between a space--time syndrome history and detection events, given by the collection of detector outcomes. To this end, we identify sets of error mechanisms (detectors) with their binary indicator vectors in $\mathbb{F}_2^{n}$ ($\mathbb{F}_2^{m}$).
The full DEM is then specified by a triple 

\[
(H,\, p,\, L),
\]
where
\begin{itemize}
  \item $H \in \mathbb{F}_2^{m \times n}$ is the \emph{detector matrix},
  specifying which of the $n$ error mechanisms flip which of the $m$ detectors;
  entry $H_{ij} = 1$ indicates that mechanism $j$ flips detector $i$.
  \item $p = (p_1,\dots,p_n)$ with $p_j \in [0,1]$ is the vector of
  \emph{error probabilities} associated with each of the independent error mechanisms.
  \item $L \in \mathbb{F}_2^{k \times n}$ is the \emph{logical map}, where column $L_{\ast j}$ indicates which of the $k$ logical representatives are flipped by mechanism $j$. We will denote each column of $L$ as $L_i \in \mathbb{F}_2^n$ for $i =  1,..,k$.
\end{itemize}

Given a detector outcome $s \in \mathbb{F}_2^m$, the joint probability of
observing $s$ and a logical vector $\gamma \in \mathbb{F}_2^k$ is
\begin{equation}\label{eq:prob_sector}
    \mathbb{P}(s,\gamma)
  = \sum_{\substack{x \in \mathbb{F}_2^n \\ Hx = s,\, Lx = \gamma}}
    \mathbb{P}(x)
\end{equation}
where $x_j = 1$ corresponds to the $j^\text{th}$ error mechanism occurring, and $x_j = 0$ otherwise. Due to the independence of error mechanisms
\begin{equation}\label{eq:error_prob}
    \mathbb{P}(x)
  = \prod_{j=1}^n (1-p_j)^{1-x_j} p_j^{x_j} \propto \prod_{j=1}^n  \Bigl(\frac{p_j}{1-p_j}\Bigr)^{x_j}.
\end{equation}

\subsection{Optimal decoding}\label{sec:optimal_decoding}

The process of \emph{optimal decoding} corresponds to determining the logical $\gamma^*$ which arises with the highest probability over all logicals $\gamma \in \mathbb{F}_2^k$, given an error model described by the DEM $(H,p,L)$, and a detector outcome $s \in \mathbb{F}_2^m$. Specifically this corresponds to finding
\begin{equation}\label{eq:ml_decoding}
    \gamma^*(s)
  = \arg\max_{\gamma \in \mathbb{F}_2^k} \mathbb{P}(s,\gamma),
\end{equation}
that is, the most likely logical configuration given the measured detector pattern. Once this logical has been identified, the appropriate recovery operation can be applied to return the system back to the code space.

In contrast, \emph{most-likely error} (MLE) decoding focuses on identifying the single physical error consistent with the observed detector outcome that occurs with highest probability. Rather than optimizing directly over logicals, this approach first selects the most probable error configuration and then infers its logical effect.

Specifically, MLE decoding determines
\begin{equation}\label{eq:mle_error}
    x^{(\mathrm{MLE})}(s)
  = \arg\max_{\substack{x \in \mathbb{F}_2^n \\ Hx = s}} \mathbb{P}(x),
\end{equation}
and assigns the corresponding logical
\begin{equation}\label{eq:mle_decoding}
    \gamma^{(\mathrm{MLE})}(s)
  = L \cdot x^{(\mathrm{MLE})}(s).
\end{equation}
After which the appropriate recovery operation is performed.

\subsection{Approximate optimal decoding}\label{sec:approximate_optimal}

Exact optimal decoding corresponds to computing the total probability of all physical error configurations consistent with a given detector outcome and logical vector $\gamma$, as in \autoref{eq:prob_sector}. This is generally intractable, even if generation and computation of the probability of any single such configuration is efficient.
Intuitively, this is because the number of such errors typically grows exponentially with system size.

This difficulty can be circumvented by performing \emph{approximate} optimal decoding, that is one computes the posterior probability $\mathbb{P}(s, \gamma)$ for all $\gamma$ approximately with some bounded error and makes a decoding decision based on these estimates. 

To be explicit, we consider a setup where one can approximately sample from this distribution, that is we take independent samples from a distribution $\mathbb{P}_{\epsilon}(x)$, that is within total variation distance $\epsilon$ from the distribution $\mathbb{P}(x)$. 
Then, we can obtain an estimator $\widehat{P}_{\epsilon, M}(s, \gamma)$ using $M$ samples such that with probability $1-\alpha$
\begin{equation}\label{eq:estimator_bound}
    \abs{\widehat{P}_{\epsilon, M}(s, \gamma) - \mathbb{P}(s, \gamma)} \leq 
        \epsilon  + \sqrt{\frac{\log(2/\alpha)}{2 M}}.
\end{equation}
This follows directly from the definition of the total variation distance, and Hoeffding's inequality.

In this paper, we will use a Markov chain to construct the estimator $\widehat{P}_{\epsilon, M}(s, \gamma)$. In this case, we can guarantee a small bias $\epsilon$ if we initially run the chain for its \emph{mixing time} $\tmix(\epsilon)$ before starting to take the samples. Then, because of correlation between successive steps of the chain, if we take samples at $T$ successive steps, a version of \autoref{eq:estimator_bound} remains true, but with the number of samples $M$ being replaced by an \emph{effective number} $\Tilde M = T / (\trel + 1/2)$, where $\trel$ is the \emph{relaxation time} of the chain (see e.g. Theorem 1 in Ref. \cite{leon2004hoeffding}). For the formal definition of $\tmix$, $\trel$ and the relation between the two, we refer the reader to \appref{app:mixing_time_proof}.

This means, in particular, to estimate logical class probabilities within $\delta$ of the true value with probability $1-\alpha$, we need $\epsilon < \delta$ and a number of steps $T$ of the chain such that
\begin{equation}\label{eq:relaxation_time}
    T \geq (\trel+\tfrac{1}{2}) \frac{1}{2(\delta - \epsilon)^2} \log(\frac{2}{\alpha}).
\end{equation}
Importantly, even if we want $\alpha$ to be exponentially small in the distance $\alpha = \exp(-\Omega(d))$, the above only grows linearly with $d$. This is important, since otherwise the finite-sample-size error could dominate the logical error rate. 

The parameter $\delta$ controls how well we are approximating the optimal decoder, which is exactly realized in the limit $\delta\to0$. Any finite $\delta$ sets the resolution at which the likelihood of different classes can be distinguished. A natural question is how a finite resolution $\delta > 0$ influences the performance of the decoder. Remarkably, we expect the error suppression for $\delta < 1/2$ to asymptotically be \emph{the same} as that of the optimal decoder, as we argue in the following.

Consider for simplicity the case of a single logical qubit $k=1$. In this case, there are two classes of errors, which we denote by $\gamma_{0/1}$ and we can write 
\begin{align}
\mathbb{P}(s,\gamma_0) &= \frac{Z(s, \gamma_0)}{\sum_{\gamma} Z(s,\gamma)} \\
    &= \frac{1}{1 + e^{-\beta f_{\rm dw}}}
\end{align}
where $Z(s, \gamma)\propto \mathbb{P}(s, \gamma)$ is the partition function of errors conditioned on the syndrome $s$ in the logical class $\gamma$. In the second line we have defined the \emph{domain wall free energy} $f_{\rm dw}(s) := -\tfrac{1}{\beta}\log[Z(s, \gamma_1) / Z(s, \gamma_0)]$. Without loss of generality, we can assume that $\gamma_0$ is always the class of the physical error, which means that decoding fails iff $f_{\rm dw}(s) \leq 0$. The domain wall free energy depends on the syndrome $s$, and hence is a random variable with distribution $P(f_{\rm dw})$. We can write the logical failure probability as a function of this distribution as
\begin{equation}
    \mathbb{P}_{\rm fail}^{\rm (opt)}  = \int_{-\infty}^0 P(f_{\rm dw}) \dd f_{\rm dw}.
\end{equation}

The decodable phase of the code corresponds to the ordered phase in a corresponding statistical mechanics model \cite{Dennis_2002}. In this case, we expect the domain wall free energy to scale with the distance of the code $f_{\rm dw} = \sigma d$. In particular, the \emph{domain wall tension} $\sigma$ is an intensive random variable, and one can argue that $f_{\rm dw}$ obeys a large-deviation principle\footnote{As pointed out in Ref. \onlinecite{chen2025scalableaccuracygainspostselection}, more formally one would view $f_{\rm dw}$ as a sum of weakly correlated random variables (and hence as a sequence of measures obeying an LDP) since for the domain wall free energy to be anomalously low in an ordered phase it has to pass through a sequence of `rare regions' of disorder, that is regions with an atypically small domain-wall tension.} (LDP) for $\sigma$ below its typical value $\bar \sigma$ \cite{chen2025scalableaccuracygainspostselection}. By this, we mean that for $\sigma < \bar \sigma$
\begin{equation}\label{eq:ldp}
    P(f_{\rm dw} = \sigma d) = e^{-I(\sigma) \cdot d + o(d)}
\end{equation}
where
\begin{equation}
    I(\sigma) = \lim_{d\to\infty} -\frac{1}{d}\log P(f_{\rm dw} = \sigma d)
\end{equation}
is called \emph{rate function}, and $o(d)$ denotes the presence of subleading contributions in the exponent. One also says that the left and right hand side of \autoref{eq:ldp} are logarithmically equivalent \cite{varadhan2008ldp,burenev2025ldp_intro}.
Then, by the Laplace principle for $\rho < \bar\sigma$
\begin{align}
    \int_{-\infty}^{\rho d} P(f_{\rm dw}) \dd f_{\rm dw} 
    &= \exp(-\left(\inf_{\sigma\leq \rho} I(\sigma)\right) \cdot d + o(d))\\
    &= e^{-I(\rho)\cdot d + o(d)}
\end{align}
where in the second line, we have assumed that the rate function is convex and monotonically decreasing (for $\sigma \leq \bar \sigma$), which is standard in statistical physics \cite{DemboZeitouni1998LargeDeviations}.
Setting $\rho = 0$ in the above, this immediately yields the result that for large $d$, the error rate of the optimal decoder is dominated by syndromes with \emph{zero} domain-wall tension.

Returning to our \emph{approximate} optimal decoder with resolution $\delta$, it will fail if the difference between class probabilities is less than $\delta$. This means that 
\begin{align}
    \mathbb{P}_{\rm fail}^{(\delta)} 
    &= \int_{\infty}^{2{\rm arctanh}(2\delta)} P(f_{\rm dw}) \dd f_{\rm dw} \\
    &= \exp(-I\left(\frac{2 {\rm artanh}(2 \delta)}{d}\right) \cdot d + o(d)) \\
    &= e^{-I(0^+) \cdot d + o(d)}.
\end{align}
Assuming that $I(0^+) = I(0)$, the above means that the approximate optimal decoder has asymptotically the same error suppression as the optimal decoder, up to sub-leading corrections.

We expect the argument to generalize with minor modifications to the case of finitely many logical sectors, $k=O(1)$, but the case where $k$ grows with system size may be more subtle. However, we note that for hyperbolic surface codes with constant rate, we observe in our numerics a strong concentration of the probability on very few logical classes (see \autoref{sec:numerical_results}). This may be interpreted as evidence that even in this case a version of the above argument applies.

We have introduced the DEM formalism, and discussed the notion of both exact and approximate optimal decoding and their relation, all in some generality. In what follows, we restrict our attention to \emph{matchable} decoding problems.

\subsection{Matchable decoding problems}\label{sec:matchable_decoding}

Using the DEM formalism introduced in \autoref{sec:DEM}, we define what is meant by a \emph{matchable} decoding problem. We show that for such decoding problems we can map the DEM onto a decoding graph, which naturally allows graph based decoders, such as minimum weight perfect matching (MWPM), to be used. Finally, we introduce a prescription for adding a virtual boundary node to a decoding graph whenever there are ``open’’ edges. This is a necessary step for certain graph based decoding algorithms, including the worm decoder, which will be introduced in \autoref{sec:worm_algorithm}.

A DEM is said to be \emph{matchable} when every error mechanism triggers at most two detectors. Equivalently, every column of the detector matrix \(H\) has support size at most two. For each column \(j\), define
\begin{equation}
    \operatorname{supp}(H_{\ast j})
  = \{\, i \in [m] \;|\; H_{ij} = 1 \,\}.
\end{equation}
The model is matchable iff
\begin{equation}
    |\operatorname{supp}(H_{\ast j})| \in \{1,2\},
    \qquad \forall j \in [n].
\end{equation}

Under this condition, we can map \((H,p)\) onto a weighted, undirected graph \(G = (V,E,w)\), referred to as a \emph{decoding graph}. The vertices of the graph represent detectors, while the edges represent error mechanisms. It is important to note, however, that as defined here \(G\) may contain self loops\footnote{the self loops correspond exactly to errors which trigger exactly one detector} (i.e., it is not necessarily a simple graph), which requires special treatment in most graph-based decoding algorithms. We discuss how this is handled in
\autoref{sec:boundary_node}. 

Each edge of the decoding graph is weighted by the
likelihood ratio of the corresponding error mechanism,
\begin{equation}\label{eq:graph_weights}
    w_e = \frac{p_e}{1 - p_e},
    \qquad e \in E.
\end{equation}

A subset of activated error mechanisms in the DEM picture corresponds to a subset
of edges, \(A \subseteq E\), of the decoding graph. In this context, the subset of edges is referred to as an \emph{error chain}. If a vertex in the decoding graph corresponds to a triggered detector, then we say a \emph{defect} is located at the vertex.

The logical maps of the DEM, given by $\{ L_i \in \mathbb{F}_2^n \}$, are naturally mapped to
subsets of edges on the graph, $L_i \mapsto \tilde{L}_i \subseteq E$.

Given an error chain, $A \subseteq E$, the value of the $i^\text{th}$ logical outcome is
\begin{equation}\label{eq:graph_logical}
    \tilde{L}_i[A] := |A \oplus \tilde{L}_i| \bmod 2 \in \{0,1\},
\end{equation}
where $\oplus$ corresponds to the symmetric difference (XOR). This quantity hence corresponds to counting whether the number of shared edges between $A$ and $L$ is even or odd. We denote the full logical vector by $\tilde{L}[A] = \{\tilde{L}_i[A]\}_{i=1}^k$.

When dealing with a \emph{simple} decoding graph (a decoding graph with no self loops), the defect configuration, \(S\) generated by an error chain $A$, is precisely its graph boundary:
\begin{equation}\label{eq:boundary_operator}
    S = \partial A := \{\, v \in V : d_A(v) \text{ is odd} \},
\end{equation}
where $\partial$ is referred to as the \emph{boundary operator}, and
\begin{equation}
    d_A(v) := |N_A(v)|, 
    \quad 
    N_A(v) := \{\, u \in V : uv \in A \,\}.
\end{equation}
That is, the defects in the decoding graph are the set of vertices in $V$ that
are connected to an odd number of edges in the error chain $A$.

This property naturally leads to matching-based decoders, which consider the space of error chains consistent with a given defect configuration. A canonical example is minimum-weight perfect matching (MWPM), which implements MLE decoding as introduced in \autoref{eq:mle_decoding}. In this setting, decoding amounts to identifying the most likely error chain consistent with the observed defects,
\begin{equation}\label{eq:mwpm_decoding}
    A^{(\mathrm{MWPM})}(S)
  = \arg\max_{\partial A = S} \mathbb{P}(A).
\end{equation}
While the problem of decoding via MLE is intractable in general \cite{iyer2015hardness}, for matchable decoding problems it can be done efficiently, since it corresponds to finding the minimum-weight path between a given set of vertices on a graph. To this end, one first uses Dijkstra's algorithm to construct the shortest paths between all pairs of triggered detectors, and then Edmond's Blossom algorithm \cite{Edmonds_1965} to identify which combination of pairs minimizes total weight. 
MWPM often provides good performance \cite{Wang_2003}, but is sub-optimal because it does not account for multiplicity of equivalent errors within the same logical class [cf. \autoref{eq:ml_decoding}].

\subsection{Removing self-loops by adding boundary nodes}\label{sec:boundary_node}

The assumption that every error mechanism triggers \emph{exactly} two detectors does not hold for all matchable decoding problems. In particular, it does not hold even for independent phase-- and bit--flip noise on topological codes with open boundary conditions---for example surface codes.     

In this case it is not guaranteed that there's an even number of defects, as they no longer need to be created in pairs. Technically speaking, nodes at the boundary have self-loops, and this leads to complications when implementing certain matching based decoders. To get around this, we can add a \emph{boundary node} to the decoding graph, creating a new graph without self-loops.\footnote{there are multiple variants of doing this, see \cite{deMarti_iOlius_2024} for another example that is slightly different from our version.}

Whenever a matchable DEM contains one or more columns of support size one, i.e.
\begin{equation}
    |\operatorname{supp}(H_{\ast j})| = 1
    \quad\text{for some } j,
\end{equation}
the decoding graph must be augmented with an additional \emph{boundary node}. 

The construction is straightforward. Let $A_B \subset E$, be the subset of edges in decoding graph connected only to a single vertex. We choose to connect all of the edges to a new boundary node, $v_B$, while keeping edge weights the same. This process forms a simple decoding graph, with one more additional node than before. 

In order to determine whether a defect is present at the boundary vertex $v_B$, the total parity of the detectors triggered has to be measured. If the detector outcome \(s \in \mathbb{F}_2^m\), satisfies \(|s| \bmod 2 = 1\) where $|s|$ corresponds to the Hamming weight, then there's a defect present at the boundary detector. This ensures that the total number of defects in the augmented decoding graph is always even, guaranteeing a global matching exists.

In \autoref{fig:boundary_detector} we show how this procedure is performed for the surface code with i.i.d.\ bit-flip noise and perfect measurements.

\begin{figure}
  \centering
    \includegraphics[width=0.9\linewidth]{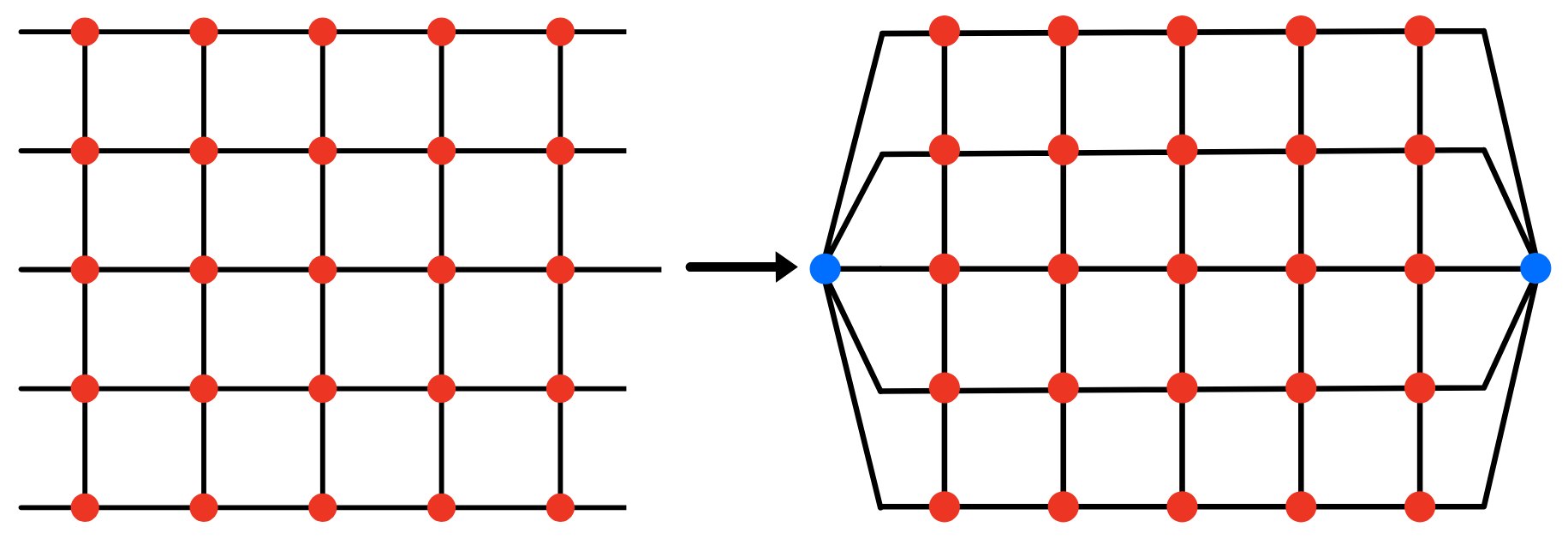}
   \caption{Illustration of the addition of an effective boundary detector for the DEM of the surface code with i.i.d.\ bit-flip noise. Red circles correspond to detectors (stabilizers), edges correspond to error mechanisms (qubits) and the blue circle correspond to the additional boundary node added. Note that both the left and right blue circles correspond to the \emph{same} boundary node.}
    \label{fig:boundary_detector}
\end{figure}

In all subsequent analysis, whenever necessary, it is assumed that this procedure has been performed, and hence we are always dealing with a simple decoding graph. As a result, \autoref{eq:boundary_operator} always holds.

Having introduced the decoding-graph formalism for matchable DEMs, we now turn to the problem of optimal decoding in this setting.

\subsection{Optimal decoding for matchable decoding problems}\label{sec:optimal_matchable}

In this section, we consider the problem of optimal decoding for matchable DEMs.
For such models, once a detection event is observed, the remaining uncertainty
in the error can be described in terms of the addition of cycles, that is,
configurations of edges with no boundary. Using this observation, we show that
optimal decoding can be reformulated in terms of properties of an ensemble of
cycles on an appropriately weighted graph. 

Consider a decoding graph \(G = (V,E,w)\), where each edge \(e \in E\) is assigned
a positive weight \(w_e \in \mathbb{R}_{>0}\), given by
\autoref{eq:graph_weights}. Since all error mechanisms in the DEM are assumed to
be independent, \autoref{eq:error_prob} translates directly to
\begin{equation}\label{eq:loop_measure}
    \mathbb{P}(A) \propto \prod_{e \in A} w_e .
\end{equation}

As discussed in \autoref{sec:matchable_decoding}, an error chain \(A\) has defects
located at its boundary \(S = \partial A\). The joint probability of observing a 
logical configuration \(\Gamma\) together with a detector outcome \(S\) is
therefore
\begin{equation}
    \mathbb{P}(S,\Gamma)
      = \sum_{\substack{A \subseteq E \\ \partial A = S,\, \tilde{L}[A] = \Gamma}}
        \mathbb{P}(A).
\end{equation}

For a fixed detector outcome \(S\), we can define a one-to-one correspondence
between the set of error chains on the decoding graph consistent with \(S\) and
the set of cycles
\begin{equation}
    \mathcal{C}_0 := \{\, A \subseteq E : \partial A = \emptyset \,\}.
    \label{eq:cycle_defintion}
\end{equation}
This correspondence is obtained by fixing a \emph{reference matching}
\(M \subseteq E\) satisfying \(\partial M = S\). Any error chain \(A\) with
boundary \(S\) can then be written uniquely as
\begin{equation}\label{eq:reference_cycle}
    A = C \oplus M, \qquad C \in \mathcal{C}_0.
\end{equation}

We can now define a probability distribution over the space of cycles,
denoted by \(\mathbb{P}_M(\cdot)\), which is directly related to the
distribution of error chains conditioned on the observed detection event.
Specifically,
\begin{equation}\label{eq:cycle_distribution}
   \mathbb{P}_M(C)
     := \mathbb{P}(C \oplus M)
     \propto \prod_{e \in C \oplus M} w_e
     \propto \prod_{e \in C} \tilde{w}_e ,
\end{equation}
where
\begin{equation}\label{eq:reweight_edges}
    \tilde{w}_e =
    \begin{cases}
        w_e, & \text{if } e \notin M, \\[4pt]
        w_e^{-1}, & \text{if } e \in M .
    \end{cases}
\end{equation}

This allows us to consider an ensemble of cycles on a reweighted graph
\(\tilde{G} = (V,E,\tilde{w})\), with edge weights \(\{\tilde{w}_e\}\) given by
\autoref{eq:reweight_edges}. In this representation, the problem of optimal
decoding can be formulated entirely in terms of properties of this ensemble of
cycles.

The probability of observing a logical configuration \(\Gamma\) given a detection event \(S\) is 
\begin{equation}\label{eq:p_lambda_s}
    \mathbb{P}(\Gamma \mid S)
     := \mathbb{P}_M(\Gamma_M)
      = \sum_{\substack{C \subseteq \mathcal{C}_0 \\ \tilde{L}[C] = \Gamma_M}}
        \mathbb{P}_M(C),
\end{equation}
where \(\Gamma = \Gamma_M \oplus \tilde{L}[M]\) and the reference matching satisfies \(\partial M = S\). 

In this formulation, optimal decoding corresponds to finding
\begin{equation}\label{eq:lambda_argmax}
    \Gamma^*_M = \arg\max_{\Gamma_M} \mathbb{P}_M(\Gamma_M),
\end{equation}
which is related to the most likely logical configuration on the original
decoding graph, \(\Gamma^*\), via
\begin{equation}
    \Gamma^* = \Gamma^*_M \oplus \tilde{L}[M].
\end{equation}
Once \(\Gamma^*\) has been identified, the corresponding recovery operation can be applied to perform decoding.

To summarize, we have shown that, for all matchable decoding problems, the ensemble of error chains conditioned on a detection event is equivalent to an ensemble of cycles on a weighted graph. 
In particular, optimal decoding can be reformulated as the task of maximizing the conditional probability $\mathbb{P}(\Gamma \vert S)$ over classes of cycles $\Gamma$.

This reformulation naturally leads to the worm process, a Markov chain that samples exactly the distribution in \autoref{eq:loop_measure}. More precisely, it acts on the combined space of cycles and near-cycles on weighted graphs, and its stationary distribution, when marginalized over near-cycles, coincides with \autoref{eq:loop_measure}.
Assuming fast mixing of the process, and sufficient weight on cycles, this will enable efficient approximate optimal decoding for matchable codes.

Before introducing the worm decoder in full detail in \autoref{sec:worm_algorithm}, we briefly relate the construction presented above to some of the well-known mappings between decoding problems and statistical mechanics models.

\subsection{Relation to statistical mechanics mappings}

This reformulation of the ensemble of errors in terms of an ensemble of cycles naturally connects to well-known statistical mechanics mappings for matchable decoding problems. Two canonical examples are the toric code with independent bit-flip noise, with and without measurement errors, which map respectively to the random-bond Ising model (RBIM) and the three-dimensional $\mathbb{Z}_2$ random plaquette gauge theory \cite{Dennis_2002, Wang_2003, Chubb_2021}.

In both of these disordered systems, the elementary excitations take the form of closed loops (cycles): domain walls in the RBIM and flux loops in the three-dimensional gauge theory. This loop structure is a general feature of statistical mechanics mappings for matchable decoding problems, and corresponds directly to the cycle condition given in \autoref{eq:reference_cycle}.

Moreover, the quenched disorder present in these statistical mechanics models corresponds directly to the reweighting procedure defined in \autoref{eq:reweight_edges}. In particular, the inversion of edge weights corresponds to the presence of negative couplings in the statistical mechanics formulation.

While the statistical mechanics mapping is more general \cite{Chubb_2021, aitchison2025spacetime_spins}, efficient sampling from the Gibbs distribution of the corresponding disordered spin model is only possible in special cases. We note that there are distinct ways to use the worm algorithm for sampling of partition functions. For the Ising model, if all couplings are ferromagnetic, one can use the high-temperature expansion to write the partition function in terms of a loop ensemble with positive weights \cite{Prokof_ev_2001}. This works on \emph{any} graph but does not work in the case where some couplings are antiferromagnetic (which is the case relevant for decoding). Another possibility, which is equivalent to the mapping presented in this work, is to use the low-temperature expansion to write the partition function as a sum over excitations on top of the ground state. In this case, the weights are always positive but one only obtains a loop ensemble (and can then use the worm) if the excitations are loop--like. It is easy to see that in general, the loop-like nature of excitations in the statistical mechanics model is equivalent to the corresponding decoding problem being matchable.

\section{The Worm Decoder\label{sec:worm_algorithm}}

In this section, we introduce the \emph{worm decoder}, a decoding algorithm that performs \emph{optimal} decoding for \emph{all} matchable decoding problems. 

Before introducing the algorithm itself, it is useful to briefly summarize what has been discussed so far. In \autoref{sec:DEM}, we introduced the detector error model (DEM) formalism for describing a broad class of noise models across a wide range of codes.

We then specialized, in \autoref{sec:matchable_decoding}, to \emph{matchable} decoding problems, in which each independent error mechanism triggers at most two detectors. For such problems, the DEM can be mapped onto a weighted graph \(G = (V,E,w)\), where vertices correspond to detectors, edges correspond to error mechanisms, and edge weights encode their likelihood ratios.

In \autoref{sec:optimal_matchable}, we showed that for matchable decoding problems, the ensemble of error chains consistent with an observed detection event can be mapped onto an ensemble of cycles on a reweighted graph \(\tilde{G} = (V,E,\tilde{w})\). As a consequence, the task of optimal decoding can be reformulated entirely in terms of the properties of this ensemble of cycles.

Optimal decoding can therefore be performed by sampling from this ensemble of cycles, which in turn corresponds to sampling over error chains consistent with the observed detection event. This is precisely the role played by the worm algorithm: a Markov chain Monte Carlo method that effectively samples from ensembles of cycles on weighted graphs.

In principle, many different algorithms could be used to sample from such ensembles of cycles. A simple approach is to employ local updates \cite{wootton2021high, hutter2014efficient, wichette2025partitionfunctionframeworkestimating, aitchison2025spacetime_spins}, such as flipping elementary plaquettes, which is equivalent to Glauber--type dynamics in the corresponding statistical mechanics models. However, such local-update schemes are typically severely limited in their applicability for disordered spin systems, as they suffer from severe slowdown in the presence of an unfavorable energy landscape. This problem is even more severe when one has to sample over distinct homology classes, which are only connected to each other via large non-local moves with low acceptance rates in this setting. 
This latter point is less problematic in the case where $k$ is small, and one can do sampling within each of the $2^k$ logical sectors separately. However, estimating the (relative) logical class probabilities then requires more sophisticated schemes involving e.g. thermodynamic integration \cite{hutter2014efficient,aitchison2025spacetime_spins} or Wang-Landau sampling \cite{wichette2025partitionfunctionframeworkestimating}. Efficiency of these schemes to the best of our knowledge has not been rigorously analyzed, but they involve significant overheads and would most likely be prohibitively expensive for high-rate codes. 

The key idea underlying the worm algorithm is that sampling efficiency can be dramatically improved by enlarging the configuration space beyond closed cycles alone, and then considering the marginal distribution restricted to cycles. This enlargement enables non-local loop updates that are inaccessible to purely local dynamics, allowing the Markov chain to rapidly explore different loop configurations. It has been shown that the worm algorithm can effectively circumvent bottlenecks present in local dynamics \cite{Wang_2005,Collevecchio_2016}, and it also naturally enables sampling across distinct logical sectors. For the present case, we will discuss the mixing time of the worm, and the resulting time complexity of the worm decoder in \autoref{sec:fast_mixing}.

Many variants of the worm algorithm exist in the statistical physics and computer science literature. In \autoref{sec:worm_algo_subsec}, we focus on the \emph{symmetric} worm algorithm, which is conceptually simple, most straightforward to analyze, and serves as a clear reference point. In \appref{App:worm_variants}, we discuss two additional variants that can also be used within the worm decoder framework. To obtain our numerical results on decoding performance, we use a slight modification of the symmetric worm algorithm, referred to as the \emph{directed} worm algorithm, which is described in \appref{app:directed_worm}.

\subsection{The Worm Algorithm}\label{sec:worm_algo_subsec}

We now introduce the (symmetric) worm algorithm~\cite{Prokof_ev_2001}, which forms the basis of the worm decoder introduced in \autoref{sec:worm_decoder}.

The worm algorithm is a Markov Chain Monte Carlo (MCMC) method acting on edge configurations of a weighted graph
\(\tilde{G} = (V,E,\tilde{w})\). 
The main idea of the algorithm is that it operates on the enlarged configuration space 

\begin{equation}\label{eq:chain_defintion}
    W := \mathcal{C}_0 \cup \mathcal{C}_2,
\end{equation}
where \(\mathcal{C}_0\) denotes the space of cycles, defined in \autoref{eq:cycle_defintion}, and
\begin{equation}
    \mathcal{C}_2 = \{ A \subseteq E : |\partial A| = 2 \},
\end{equation}
is the space of chains with exactly two boundary vertices, referred to as \emph{virtual defects}.

The basic update proceeds as follows. A vertex is selected at random, creating a pair of virtual defects. The two virtual defects then perform a biased random walk on the weighted graph until they encounter one-another, at which point the pair annihilates, resulting in a closed-loop update. By temporarily allowing configurations with open boundaries, the algorithm enables non-local loop updates that are
inaccessible to purely local dynamics. This allows the Markov chain to efficiently explore the space of cycle configurations, including transitions between distinct logical sectors. A schematic illustration of this process is shown in \autoref{fig:worm_process_schematic}.

\begin{figure}
  \centering
  \includegraphics[width=0.9\linewidth]{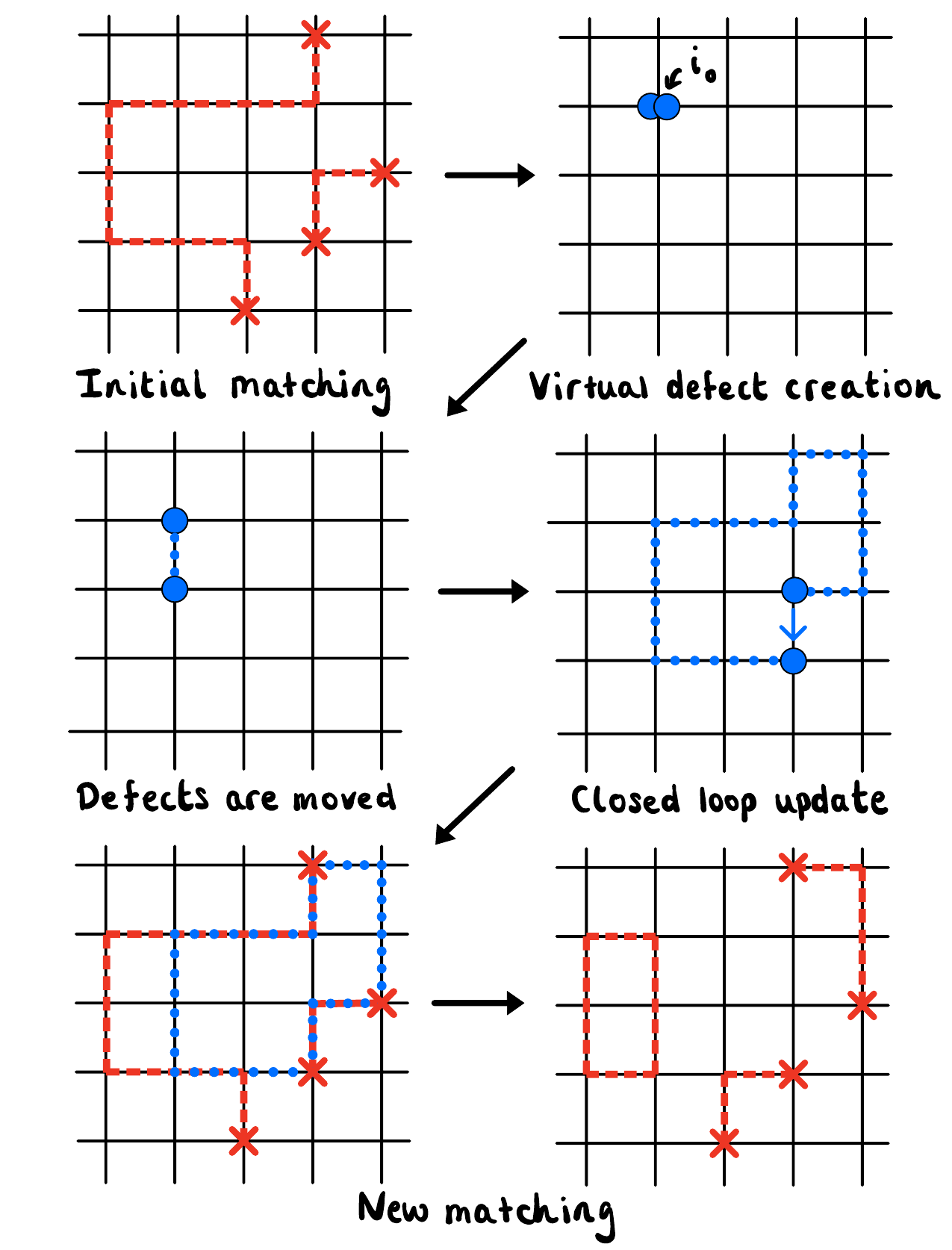}
  \caption{Schematic representation of the worm process. A pair of virtual defects is
  created at a random vertex. One performs a random walk on the graph until it
  meets the other, forming a loop update. Red crosses denote the observed syndrome.}
  \label{fig:worm_process_schematic}
\end{figure}

A formal description of this update procedure (including post-selection onto the space of cycles $\mathcal C_0$) is given in \autoref{alg:wormprocess}.

By construction, the stationary distribution of the worm process is given by the Prokof'ev--Svistunov (PS) measure~\cite{Prokof_ev_2001}. For \(A \in \mathcal{W}\),
\begin{subequations}\label{eq:PS_measure}
\begin{equation}
    \pi(A)
      = \frac{\Psi(A)\, f(A)}
             {\abs{V}\,f(\mathcal{C}_0) + 2\,f(\mathcal{C}_2)},
\end{equation}
where 
\begin{align}
f(A) &\equiv f_{\{w\},M}(A) = \prod_{e \in A} \tilde{w}_e, \\
f(\mathcal{S}) &= \sum_{A \in \mathcal{S}} f(A)
\end{align}
and
\begin{equation}
\Psi(A) :=
\begin{cases}
\abs{V}, & A \in \mathcal{C}_0, \\[3pt]
2, & A \in \mathcal{C}_2.
\end{cases}
\end{equation}
\end{subequations}

By restricting to the marginal distribution on \(\mathcal{C}_0\), the PS measure yields unbiased samples from the desired closed-loop ensemble \(\mathbb{P}(C) \propto \prod_{e \in C} \tilde{w}_e\). Therefore by starting with a DEM describing a matchable decoding problem, constructing the graph $G$, finding a reference matching $M$ to construct $\tilde{G}$, it is possible to use the worm algorithm to faithfully sample over error chains consistent with a detection event.

\begin{algorithm}[H]
\caption{The symmetric worm algorithm}
\label{alg:wormprocess}
\begin{algorithmic}[1]
\State \textbf{Input:} Weighted graph \(\tilde{G} = (V,E,\tilde{w})\); current loop configuration \(\Sigma\).
\State \textbf{Output:} Updated loop configuration \(\Sigma'\).

\State Choose \(i_0 \in V\) uniformly at random.
\State Set \(i_1 \gets i_0\), \(i_2 \gets i_0\), and \(\Sigma' \gets \Sigma\).

\Repeat
    \State Choose \(k \in \{1,2\}\) uniformly at random and set \(i \gets i_k\).
    \State Choose \(j \in N(i)\) uniformly at random.
    \State Propose $\Sigma' \to \Sigma'\oplus (i,j)$ and accept with Metropolis probability w.r.t. the PS measure.
    \If{accepted}
        \State $i_k \gets j$.
    \EndIf
\Until{\(i_1=i_2\)}

\State \Return \(\Sigma'\)
\end{algorithmic}
\end{algorithm}

\subsection{The Worm Decoder\label{sec:worm_decoder}}

So far, we have introduced a Markov process that can sample over the ensemble
of errors for \emph{any} matchable decoding problem. We now show how this
process can be implemented as a practical and efficient decoding algorithm
capable of performing \emph{optimal decoding}.

\begin{figure}
  \centering
  \includegraphics[width=0.9\linewidth]{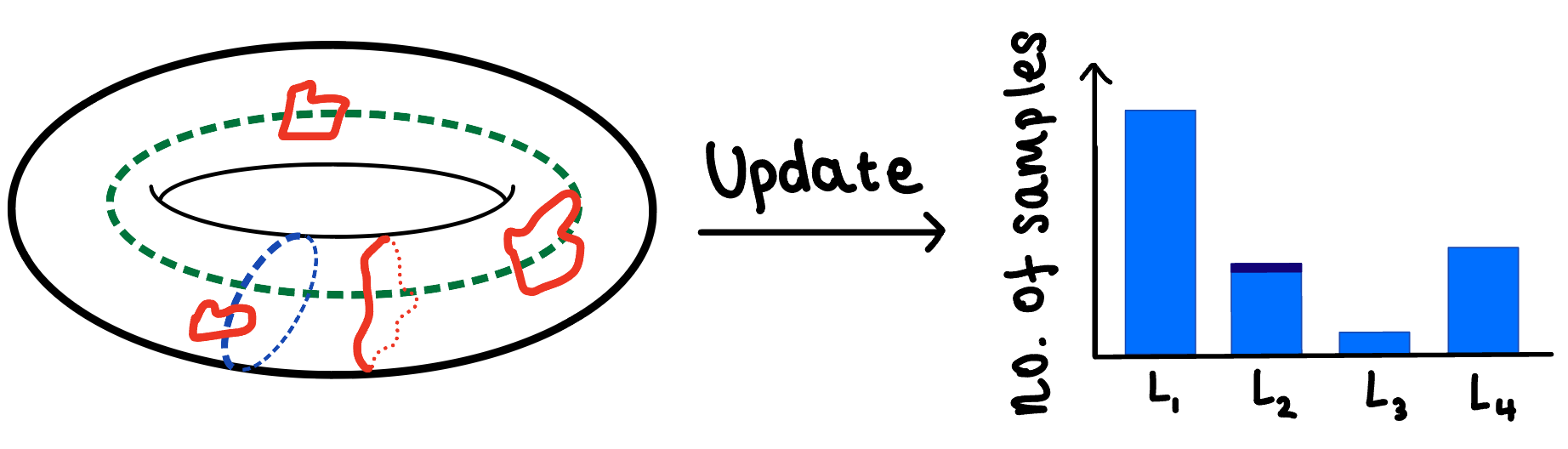}
  \caption{Schematic of the worm-decoding procedure. The logical sector of
  each sampled loop configuration is measured on the decoding graph of the
  toric code (here without measurement errors), and the corresponding logical tally is updated.}
  \label{fig:worm_decoder_schematic}
\end{figure}

The idea is straightforward. The worm process is used to generate samples of
error configurations consistent with a measured detection event. Since the
worm process has a stationary distribution given by the
Prokof'ev--Svistunov (PS) measure, the probability of observing a given closed
loop configuration is exactly the same as the probability that the
corresponding error caused the observed syndrome. Repeatedly running the worm process therefore produces independent error samples (given a sufficiently large number of steps between successive samples; see \autoref{eq:relaxation_time} for the precise requirement) drawn from the distribution defined by the DEM, conditioned on the observed detection event. Optimal decoding is then performed
by identifying the logical sector that appears most frequently among these
samples - provided sufficiently many such samples are taken. A schematic
illustrating this procedure for the toric code with bit-flip noise is shown in
\autoref{fig:worm_decoder_schematic}, and the full algorithm is presented in \autoref{alg:worm_decoder}.

Given the reference matching $M_0$ and the selected logical class $\Gamma^*$, a recovery operation is applied that returns the system to the codespace while implementing the   desired logical action. The procedure for constructing this recovery is identical to that used in MWPM decoding.

Although there are $2^k$ logical sectors in principle, the worm decoder does not need to construct an explicit histogram over all sectors. Instead, for each sampled loop configuration, the corresponding logical sector is determined by a matrix multiplication over $\mathbb{F}_2$, as shown in \autoref{eq:graph_logical}, yielding a $k$-bit string that labels the values of the logical representatives for that configuration. These logical sectors are stored in a sparse dictionary that records only those sectors actually observed during sampling. This representation remains efficient even when $k$ grows with the number of qubits, since the number of sectors with nonzero counts is bounded by the number of samples taken. In the sub-threshold regime, the posterior distribution over logical sectors is strongly peaked, so only a small number of sectors acquire non-zero counts in practice.

In practice, when using the worm algorithm as a decoder, it is often useful
to consider a slight variant of the above procedure, in which the number of
samples \(N_\text{samples}\) is not fixed. Instead, sampling may be stopped once,
within specified statistical uncertainties, one can reliably identify the
logical sector that solves \autoref{eq:lambda_argmax}, i.e., the most probable
sector consistent with the observed detection event. However, as discussed in
\autoref{sec:PS_and_estimators}, there remain important contexts in which it is
beneficial to use a fixed and sufficiently large value of \(N_\text{samples}\),
e.g. when leveraging the soft information generated by the worm process.

We have introduced a novel worm decoder that utilizes the worm process to
faithfully sample error chains of matchable decoding problems, conditioned on a
detection event. We now briefly discuss ways in which the worm decoder can be
adapted to exploit soft information during the decoding procedure, before
turning to analyzing the efficiency of our decoder in \autoref{sec:fast_mixing}.

\begin{algorithm}[H]
\caption{The Worm Decoder}
\label{alg:worm_decoder}
\begin{algorithmic}[1]
\State \textbf{Input:} Weighted graph \(G = (V, E, w)\); detection event \(S\);
logical representatives \(\mathcal{L} = \{\tilde{L}_i\}\); number of samples
\(N_{\text{samples}}\); worm autocorrelation time \(t_{\text{auto}}\).
\State \textbf{Output:} Recovery operation \(R\).
\State $M_0 \gets \Call{ReferenceMatching}{G,S}$.

\State Reweight $w_e \mapsto \tilde{w}_e$, to obtain $\tilde{G} = (V,E,\tilde{w})$.
\State Set \(\Sigma \gets \{\}\).
\State Initialize \texttt{logicals\_tally} as an empty dictionary.
\For{$n = 1$ to $N_{\text{samples}}$}
    \For{$t = 1$ to $t_{\text{auto}}$} \Comment{Perform decorrelation updates}
        \State \(\Sigma \gets\) \Call{WormUpdate}{$\tilde{G}, \Sigma$}.
    \EndFor
    \State \(\Gamma \gets\) \Call{MeasureLogical}{$\Sigma, \{\tilde{L}_i\}$}.
    \If{\(\Gamma \notin\) \texttt{logicals\_tally}}
    \State \texttt{logicals\_tally}[\(\Gamma\)] $\gets 1$
\Else
    \State \texttt{logicals\_tally}[\(\Gamma\)] $\gets$ \texttt{logicals\_tally}[\(\Gamma\)] $+ 1$
\EndIf

\EndFor

\State  $\Gamma^*_{M_0} \gets \text{argmax}_{\Gamma} (  \texttt{logicals\_tally}[\Gamma])$

\State $\Gamma^* \gets \Gamma^*_{M_0} \oplus \Call{MeasureLogical}{M_0, \{\tilde{L}_i\}}$ 

\State Determine recovery operation \(R\) from
\(\Gamma^*\) and \(M_0\).
\State \Return \(R\).
\end{algorithmic}
\end{algorithm}

\subsection{Soft Information for matchable decoding problems}\label{sec:PS_and_estimators}

Fundamentally, the worm decoder produces samples from the distribution of errors conditioned on a given syndrome, and thus produces data beyond a valid recovery operation. 
The utility of such `soft information' is well-appreciated for simulation purposes \cite{wichette2025partitionfunctionframeworkestimating} as well as practical decoding-related tasks such as confidence-based post-selection
\cite{lee2025efficient,chen2025scalableaccuracygainspostselection}, extracting noise model parameters \cite{Wagner_2021,Wagner2022paulichannelscanbe,wang2024dgrtacklingdriftedcorrelated}, and overall improved decoding performance~\cite{pattison2021improved}.

In this paper, we discuss three concrete examples.
Below, we first review how logical class probabilities can be used to estimate logical error rates of optimal decoders in simulations \cite{wichette2025partitionfunctionframeworkestimating}, and also a simple criterion for post-selection first proposed in Ref. \onlinecite{chen2025scalableaccuracygainspostselection}.
Later, in \autoref{sec:correl_decoding}, we further discuss how soft information generated by the worm can be used to include \emph{correlations} that arise when using it for \emph{non-matchable} DEMs. There, we introduce a correlated decoding scheme based on the worm algorithm that improves on existing schemes augmenting MWPM \cite{fowler2013optimalcomplexitycorrectioncorrelated,Bombin_2012}.

\subsubsection{Optimal decoding success estimators}

Using the worm decoder, we are able to generate optimal decoding success estimators for \emph{any} matchable decoding problem.

A naive way to estimate the logical error rate might be to draw an error $E$, perform decoding, and count which fraction of decoding attempts the optimal class $\Gamma_M^*$ coincides with the relative logical class of $E$ and $M$ (that is $E\oplus M$ should be in the class $\Gamma_M^*$).

As demonstrated in \cite{wichette2025partitionfunctionframeworkestimating}, one can do much better given access to the success-rate of decoding \emph{given a syndrome $S$}. Note that from \autoref{eq:p_lambda_s} and \autoref{eq:lambda_argmax}, this conditional probability is exactly given by
\begin{equation}
    \mathbb{P}_\text{succ}(S) = \mathbb{P}_M(\Gamma_M^*),
\end{equation}
where $\Gamma_M^*$ is given by \autoref{eq:lambda_argmax}. This quantity can be empirically estimated by the worm algorithm via
\begin{equation}\label{eq:decoding_succ_est}
    \texttt{logicals\_tally}[\Gamma_M^*]/N_\text{samples}
    \;\simeq\; \mathbb{P}_\text{succ}(S),
\end{equation}
where \texttt{logicals\_tally} is defined in \autoref{alg:worm_decoder}, and  $\tilde{\Gamma}_M^*$ denotes the logical sector that maximizes the observed value of \texttt{logicals\_tally}. 

For a given code and error model, the optimal decoding success rate is then
\begin{equation}
    \mathbb{P}_\text{succ} = \sum_S \mathbb{P}(S)\,\mathbb{P}_\text{succ}(S),
\end{equation}
and using \autoref{eq:decoding_succ_est}, this quantity can be estimated as
\begin{equation}
    \mathbb{P}_\text{succ} \simeq 
    \Bigg\langle 
        \frac{\texttt{logicals\_tally}[\Gamma_M^*]}{N_{\text{samples}}}
    \Bigg\rangle_S,
\end{equation}
where $\langle \cdot \rangle_S$ denotes averaging over detection events $S$. This estimator offers significant advantages when benchmarking the worm decoder: the statistical fluctuations are fundamentally smaller than those obtained from direct sampling of decoding outcomes. Each error realization returns a value in the interval $[0,1]$, whereas decoding success is a binary quantity. Consequently, the estimator exhibits markedly improved self-averaging. We present numerical results on this quantity in \autoref{sec:numerical_results} when investigating the decoding success rates of the worm decoder.

\subsubsection{Post-selection with the worm}

Soft information obtained in the decoding procedure can also be used to implement a post-selection criterion. For a given detection event $S$, using the result in \autoref{eq:decoding_succ_est}, we can obtain an estimate of the decoding success probability $\mathbb{P}_\text{succ}(S)$.

For example, one can define a procedure that post-selects on trials whose estimated probability of decoding error lies below a chosen threshold; that is, for a detection event $S=\partial M$, we retain only trials satisfying
\begin{equation}
    1 - \frac{\texttt{logicals\_tally}[\Gamma_M^*]}{N_{\text{samples}}}
    < \epsilon_\text{threshold}.
\end{equation}
This criterion was first introduced in~\cite{chen2025scalableaccuracygainspostselection}, where the focus was on the 2D toric code. It was proposed that such a post-selection rule can yield a logical error rate suppression from $p_f$ to $p_f^b$ with $b \geq 2$, while discarding only an exponentially small fraction of trials, thereby providing a scalable post-selection criterion. The underlying arguments presented in this work apply to general topological codes. The worm algorithm provides a natural platform for investigating scalable accuracy gains from post-selection for matchable DEMs, representing an interesting avenue for future work.

\section{Mixing time of the worm process and efficiency of the decoder}\label{sec:fast_mixing}

In this section, we discuss the time complexity of the worm decoder. 
The challenge to analyzing this is twofold. 
First, it depends on the \emph{mixing time} of the worm process, that is roughly speaking the time it takes to reach the steady-state distribution [\autoref{eq:PS_measure}] starting from an arbitrary initial state. 
Once the steady-state distribution is reached, the mixing also upper-bounds the time the worm process must be simulated between taking samples to guarantee they are (sufficiently) independent.
Second, the worm process is defined on the enlarged configuration space $\mathcal W = \mathcal C_0 \cup \mathcal C_2$ but we are ultimately interested in sampling from a marginal on $\mathcal C_0$. Therefore, we have to post-select on states in $\mathcal C_0$. This post-selection process is efficient if $f(\mathcal C_2) / f(\mathcal C_0) = {\rm poly}(\abs{V})$, that is if the worm process does not spend `too much' time being a worm.

Below, we will show that the latter condition is related to the first, and in fact assuming a slight generalization, which we call \emph{bounded defect susceptibility}, is sufficient to guarantee both fast mixing and efficient post-selection onto $\mathcal C_0$. The condition is a property of a given syndrome. We do not expect the run-time of the worm decoder to be polynomial for worst-case syndrome instances (see also \appref{app:matching_reduction}).
However, we will argue that \emph{typical} syndromes in the decodable phase can be decoded in polynomial time. 

We begin the rest of the section with a brief review of prior work on fast mixing results for the worm process. We then provide a formal definition of bounded defect susceptibility, and state our mixing time guarantee. 
Finally, we connect the notion of bounded defect susceptibility to the behavior of \emph{disorder operators} from statistical physics \cite{fradkin2017disorder_ops}. We use this to (i) argue (non-rigorously) that all but a vanishing fraction of syndromes are decodable in polynomial time, throughout the entire decodable phase, and (ii) discuss the average-case runtime of the worm decoder and connect the question of whether it is polynomial (in system size) to the existence of a previously unstudied kind of `Griffiths phase' in the corresponding statistical mechanics model.

\subsection{A brief history of mixing-time guarantees for the worm}

In physics, the worm algorithm (also called the ``directed loop update''  \cite{syljuaasen2002qmc}, or ``Broder's chain''\cite{Broder1986worm}) has been appreciated as a means of empirically efficient sampling in a wide range of settings. This includes simulation of the Ising model \cite{Prokof_ev_2001},
use in updates for quantum Monte Carlo schemes \cite{PROKOFEV1998253,syljuaasen2002qmc}, studying dimer models \cite{sandvik2003deconfined,Sandvik_2006}, and very recently also studying decoherence-induced transitions \cite{temkin2025charge_informed_qec}. In the last application (Ref. \onlinecite{temkin2025charge_informed_qec}) the worm is effectively used as a decoder for a $U(1)$ symmetry-enriched toric code under charge-conserving noise, although neither its efficiency nor the effect of finite-sample size and bias were studied systematically in this context. 
Remarkably, the same Markov chain has received much attention in computer science, where rigorous guarantees have been derived on its efficiency for sampling from a number of distributions.

Perhaps most famously, the worm underlies Jerrum and Sinclair's FPRAS for the permanent of (most) matrices with binary entries \cite{jerrum_sinclair_1989_permanent}, as well as the extension to all matrices with positive entries~\cite{jerrum_sinclair_vigoda_2004_permanent}. 
In physics language, the latter corresponds to approximate computation of the partition function of a dimer model on any bipartite graph and with arbitrary positive weights.

Another setting in which efficiency has been shown is sampling the partition function of the ferromagnetic Ising model on finite-degree graphs via its high-temperature series expansion (HTSE) \cite{Collevecchio_2016}. 
The latter case corresponds to the same loop ensemble as we consider in the decoding of matchable codes---usually called an `even subgraph model' in the CS literature---but relative to the even-subgraph measure considered in Ref.~\onlinecite{Collevecchio_2016} [$\lambda_w(A) = w^{\abs{A}}$ for $A \subseteq E$ and some $w\in(0, 1)$], our measure [\autoref{eq:loop_measure}] is `shifted' by the reference matching $M$ ($A \to A \oplus M$).
Equivalently, one could invert the weights on the edges in $M$. 
We emphasize that this does \emph{not} correspond to an Ising model with anti-ferromagnetic couplings, where the edge weights (under the standard HTSE mapping) would become negative, and which is known to be hard in general~\cite{Barahona1982complexity_spin_glass}.
However, we also do not expect the worm process to be fast mixing for general reference configurations $M$ and on general graphs $G$. We detail some physical intuition for this below, but also show in \appref{app:matching_reduction} that an FPRAS for the partition function of even subgraphs of the form of \autoref{eq:loop_measure} would yield a FPRAS for approximate counting of perfect matchings on general graphs, and the existence of the latter is a major open problem \cite{jerrum2003book}.

To the best of our knowledge, no rigorous results exist on using the worm algorithm to sample loop models with a measure of the form we define in \autoref{eq:cycle_distribution}. 
Fast mixing \emph{is} known for general positive weights on any graph, for the so-called Jerrum-Sinclair chain, which can be seen as a generalization of the worm process that can create an arbitrary number of defects \cite{sinclair1992improved,jerrum2003book,chen2025faster}. However, it is known that the same guarantee does not hold for the worm, where at most two defects are present at any given time \cite{stefankovi2018torpid}. 
For the worm process that approximately counts the number of perfect matchings in a graph, it was already shown by Jerrum and Sinclair \cite{jerrum_sinclair_1989_permanent} that the process mixes fast on general graphs, if one assumes that the ratio of near-perfect matchings and perfect matchings is at most polynomially large.

Below, we will derive a similar statement, but for the worm process defined for the even subgraph model that we consider, cf.\ \autoref{eq:cycle_distribution}. 
We can show fast mixing on general graphs, conditioned on the fact that the ratio of mass on the near-even subgraphs ($\mathcal C_2$) and even subgraphs ($\mathcal C_0$) is at most polynomially large. We make this precise in the following.

\subsection{Fast mixing from bounded defect susceptibility}

For simplicity, in the following we restrict the discussion to the setting where the priors are uniform in space ($w_e = w$ for all $e\in E$).
In this setting, the measure on subgraphs of $G = (V, E)$ takes the form 
\begin{equation}\label{eq:lambda_measure}
    f_{w, M}(A) = \prod_{e\in A} \tilde{w}_e \propto w^{|A \oplus M|} =: \lambda_{w,M}(A),
\end{equation}
where $A\in \mathcal{C}_0$, $w \in (0, 1)$, $M \subseteq E$ is an arbitrary `reference' subgraph and
\begin{equation}
    \tilde{w}_e = \begin{cases} w & \text{if } e \notin M,  \\ w^{-1} & \text{if } e \in M,\end{cases}
\end{equation}
consistent with \autoref{eq:cycle_distribution}. 
In the uniform-weight case, it is most convenient to consider the form of the measure~$\lambda_{w, M}$ above, which we will do in the following. Note that since $\lambda$ differs from $f$ only by a ($M$-dependent) constant, the PS measure can be expressed in terms of $\lambda$ simply by replacing $f$ by $\lambda$ everywhere in \autoref{eq:PS_measure}.

In the decoding setting, the boundary of $M$ is given by the syndrome $\partial M = S \subseteq V$ (remember that the measure~$\lambda$ only depends on the syndrome, up to relabeling of configurations $A \in \mathcal C_0$).

Before stating our main result, we briefly recall some relevant definitions. For a Markov chain on state space $\Omega$ with transition matrix $P$ and unique stationary distribution $\pi$, the mixing time is defined as
\begin{equation}\label{eq:tmix_def}
    \tmix(x_0, \epsilon) := \min\{t \in \mathbb{N} : 
    \|P^t(x_0) - \pi\|_{\rm TV} \leq \epsilon\},
\end{equation}
where $x_0 \in \Omega$ is the initial state and~$\epsilon$ is a target accuracy. While the mixing time depends in general on both $x_0$ and~$\epsilon$, this dependence can be separated from the intrinsic timescale of the chain via the bound~\cite[Theorem 12.4]{levin2017markov}
\begin{equation}\label{eq:t_rel_def}
    \tmix(x_0, \epsilon) \leq 
    \left[ \log\pi(x_0)^{-1} + \log\epsilon^{-1} \right] \trel,
\end{equation}
where $\trel$ is the \emph{relaxation time}, which is independent of both the initial state and the target accuracy.  
Informally, this is the time scale on which the chain contracts towards its stationary distribution.

In our setting, we consider the \emph{lazy} worm process, in which at each step, with probability $1/2$ no transition is proposed.
Laziness ensures that the chain is aperiodic in addition to being reversible, and hence converges to a unique stationary distribution from any initial state.
Furthermore, since all configurations $A \in \mathcal{W} = \mathcal{C}_0 \cup \mathcal{C}_2$ have weight $\lambda_{w,M}(A) \geq w^{|E|}$, the stationary probability [\autoref{eq:PS_measure}] of any state satisfies $\log \pi(x_0)^{-1} = O(\abs{E})$. Consequently, for the worm process, fast relaxation implies fast mixing.

We will show that the lazy worm process is fast mixing, conditioned  on the following property of the measure~$\lambda_{w, M}$.

\begin{property}[Bounded defect susceptibility]\label{prop:bounded_sus}
For an even subgraph model defined on a graph $G = (V, E)$ with measure $\lambda_{w, M}$ as defined in \autoref{eq:lambda_measure}, we say it has $(\chi_2, \chi_4)$-bounded defect susceptibility if
\begin{subequations}\label{eq:bounded_sus_ratios}
\begin{align}
    \frac{\lambda_{w, M}(\mathcal{C}_2)}{\lambda_{w, M}(\mathcal{C}_0)} 
    &\leq \chi_2, \\
    \frac{\lambda_{w, M}(\mathcal{C}_4)}{\lambda_{w, M}(\mathcal{C}_0)} 
    &\leq \chi_4,
\end{align}
\end{subequations}
where $\mathcal{C}_k$ denotes the set of subgraphs of $G$ with exactly~$k$ odd-degree vertices.
\end{property}
The name of the property is motivated by the fact that the ratios of partition functions in \autoref{eq:bounded_sus_ratios} can be interpreted as a sum over all two-point and four-point disorder-operators, respectively. We will detail this below.

Remarkably, one can obtain a mixing time guarantee for the worm process in terms of bounded defect susceptibility alone.

\begin{theorem}[short version]\label{thm:fast_mixing}
    Consider the lazy worm process for an even subgraph model on $G=(V, E)$ with $(\chi_2, \chi_4)$-bounded defect susceptibility (Property \autoref{prop:bounded_sus}).
    The relaxation time of this process is bounded by 
    \begin{equation}
        \trel \leq 8 \Delta w^{-1} \abs{E} \left[(1 + \chi_2)(2 + \abs{V}) + 2\chi_4 \right].
    \end{equation}
\end{theorem}

We provide a longer version of the statement and the proof of this result in \appref{app:mixing_time_proof}. The proof is a direct generalization of the existing guarantee for the worm process of the Ising model \cite{Collevecchio_2016} to our setting (and our result contains the case considered in Ref.\ \onlinecite{Collevecchio_2016} as $M = \emptyset$). In particular, we use the same set of canonical paths, and then bound the congestion using property \autoref{prop:bounded_sus}. In the Ising model case, where the worm process is obtained via the high-temperature series expansion, the same ratio of partition function corresponds to the actual Ising susceptibility, which is always polynomial in system size. 
In our case, this is not guaranteed in general, but for random error models (where $M$ is chosen from a distribution) we do expect property \autoref{prop:bounded_sus} to hold with polynomial $\chi_2$ and $\chi_4$ for \emph{typical} errors in the decodable phase. 
We discuss this in detail in the following.

\subsection{Bounded defect susceptibility: typical vs. average-case complexity}

The mixing time guarantee in theorem \autoref{thm:fast_mixing} is valid for all syndromes. As mentioned above, we do not expect the quantities $\chi_2$, $\chi_4$ to be bounded by a polynomial in system size in general. To understand this better, it is insightful to connect these ratios to the notion of a \emph{disorder operator} in statistical physics, which is usually defined as
\begin{equation}\label{eq:disorder_op}
    \expval{\mu(u) \mu(v)}_{\lambda} := \frac{\lambda_{w, M}(\mathcal C_{u, v})}{\lambda_{w, M}(\mathcal C_0)}
\end{equation}
where $\mathcal{C}_{u,v} = \{ A \subseteq E : \partial A = \{ u,v \} \}$ is the set of subgraphs with only vertices $u$ and $v$ having odd degree.
Intuitively, this quantity encodes the change in free energy when inserting two defects at site $u$ and $v$.

In a conventional ordered phase without disorder ($M = \emptyset$), this defect operator is expected to decay exponentially with the distance of $u, v$. In particular, defining the domain-wall free energy of a domain wall `pinned' at $u$ and $v$ as $\beta f_{\rm dw}(u, v) := -\log(\expval{\mu(u) \mu(v)}_\lambda)$, this will behave as $f_{\rm dw} (u, v) = \sigma_{\rm dw}\cdot {\rm dist}(u, v)$ for large distances between~$u$ and~$v$ \cite{fradkin2017disorder_ops}. 
Note that $\chi_2 = \sum_{u, v} \expval{\mu(u) \mu(v)}= O(\abs{V})$ in this case, and a similar argument applies to the ratio $\chi_4=O(\abs{V}^2)$.

The behavior of \autoref{eq:disorder_op} is more subtle, however, in a disordered system, even in ordered phases (which correspond to the decodable phase of the code under the usual statistical mechanics mapping \cite{Dennis_2002}). 
The reason is that its behavior in this case can be dominated by rare regions. This will lead, in general, to a stark difference between the behavior for typical disorder configurations, and the average. 

The left hand side of \autoref{eq:disorder_op} only depends on the positions of the defects $u$ and $v$.
In order to sample the syndromes we equivalently sample the errors directly from their distribution and take this to be our reference~$M$.
With this convention, successful decoding is equivalent to a positive domain-wall free energy for non-contractible domain walls (see also the discussion in \autoref{sec:approximate_optimal}). Note that for any configuration where the free-energy of non-contractible domains is positive, we also expect a positive limit ${\rm dist}(u, v) \to d$ for $f_{\rm dw}(u, v)$, since pinned domain walls have smaller or equal free energy.

\begin{figure}
\includegraphics[width=0.35\textwidth]{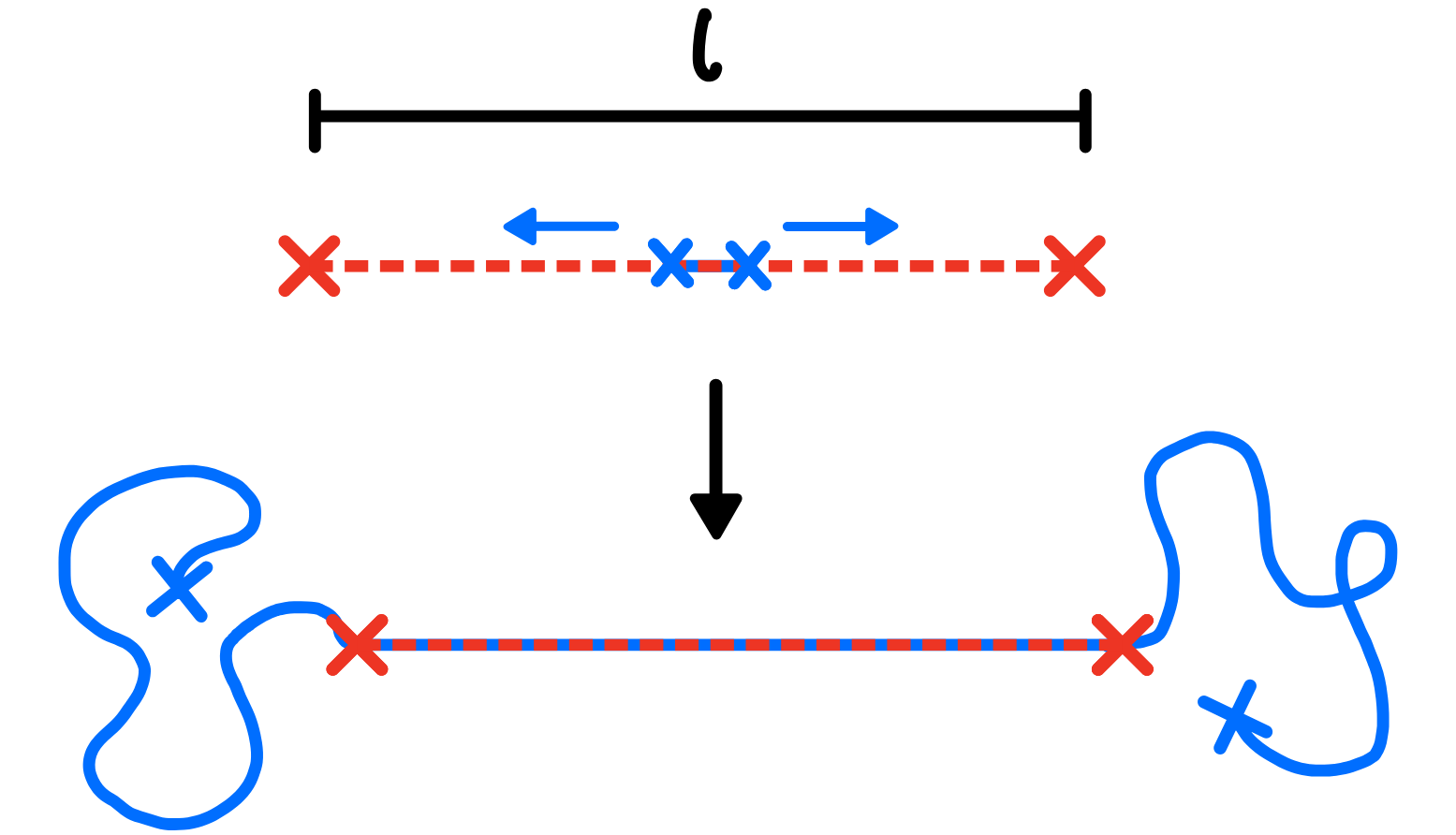}
\caption{Schematic illustrating how a pair of virtual defects (blue crosses) can form near a reference matching (red dashed line) of length $\ell$, and subsequently become trapped in the vicinity of the syndrome locations (red crosses). While such configurations of size $\ell$ contribute a mixing time that grows exponentially in $\ell$, they are also exponentially rare in $\ell$ for local random error models.}
\label{fig:worm_trap}
\end{figure}
\subsubsection{Typical errors are decodable in polynomial time}

As an example, consider a configuration of disorder $M$ as sketched in \autoref{fig:worm_trap}. In this configuration, two syndromes at distance $\ell$ are separated by a long string of edges in the reference matching (in the 2D RBIM corresponding to decoding of the surface code with perfect measurements, this would be a long dual string of `flipped' bonds connecting two far-away frustrated plaquettes). In the ordered phase, such a configuration leads to $\chi_2 \geq \expval{\mu(u)\mu(v)} = \exp(\beta \sigma \ell)$ for some positive constant $\sigma>0$. However, such configurations are also exponentially rare. Standard percolation theory bounds \cite{grimmett1999percolation} state that in a typical sample, the longest configuration of this kind would have size $\bar \ell = z\log \abs{V}$ for some $z > 0$, which, when plugged into our mixing time guarantee yields a polynomial upper bound in system size.

More generally, we expect in random local error models that syndrome configurations for which it is possible to gain a free energy on the order $\max_{u, v} \log \expval{\mu(u)\mu(v)}\propto \ell$ occur with a probability that is exponentially small in $\ell$, $p_{\ell} \propto e^{-r \ell}$, again with some positive constant $r$ and beyond, potentially, some finite-distance cutoff scale $\xi$. 
Note that this is exactly the intuition behind the exponential decay of logical error rate with distance: error configurations where this defect free energy can be negative on the scale of the distance are exactly the configurations where even optimal decoding fails, and such configurations are exponentially rare in the distance, see e.g. \cite{Dennis_2002}. 
Stated differently, assuming that configurations where free energy $\propto\ell$ can be gained are exponentially rare in~$\ell$ assumes that a neighborhood of diameter $\ell$ of the code behaves statistically similarly to a (smaller) code of size~$\ell$ (up to a change of length scales).
This is expected to hold for most spatially homogeneous codes and error models, beyond some finite-size scale, and the same idea is the basis for the notion of universality of phases in statistical mechanics. 
In these cases, we expect \emph{typical} errors (that is all but a vanishing fraction) to be decodable in polynomial time. 

If typical errors can be decoded in polynomial time, this means that we expect a vanishing error rate even when running the worm decoder with a polynomial time-out, as long as the power of the polynomial is greater than the parameter $z$ defined above. Then, even if we assume failure in all instances where the mixing time exceeds the polynomial cutoff, this will affect only a vanishing fraction of all errors. By the arguments above, we expect this to be true in \emph{entire} decodable phase, which means that the threshold of the worm decoder, with a polynomial runtime, is the same as that of the optimal decoder up to finite-sample errors. 

We note two caveats. First, the parameter $z$ may diverge as the threshold transition is approached, and therefore so does the power of the polynomial bounding the typical decoder run-time.
Second, for any polynomial runtime cutoff, we also expect (asymptotically) only a polynomial in distance suppression of logical errors. This is because instances where the mixing time exceeds a given polynomial are only polynomially rare. 
To be concrete, say we choose the maximum time to be $t_{\rm max}(\abs{V})$ as some finite-degree polynomial. Then an adversarial defect configuration as the one sketched in \autoref{fig:worm_trap} of scale $\ell = \log t_{\rm max}(\abs{V})$ will make the decoder fail. Assuming that errors happen i.i.d. on each edge with probability $p$, such a configuration has probability at least 
\begin{equation}
    p_{\rm fail} \gtrsim p^{\ell} = \exp(\log p \log t_{\rm max}) = \frac{1}{t_{\rm max}(\abs{V})^{- \log p}}.
\end{equation}
The first proportionality here is obtained by fixing only the edges between the two syndromes and marginalizing over all other edges. 
Local disorder fluctuations may lead to a normalization of distances by a constant factor to $\ell_{\rm eff} = \Theta(\ell)$, but the above will stay valid with this renormalized length scale.

As before, similar arguments apply for the four-point disorder operators and the corresponding susceptibility~$\chi_4$.

\subsubsection{The average mixing time is dominated by rare errors}

While so far, we have discussed the case of \emph{typical} errors and the runtime of the decoder, another interesting question is that of average runtime. This would be relevant for example, to the setting where one decodes many codes in parallel, with the possibility to distribute a polynomial amount of computing resources per code.

In a disordered system within the ordered phase one expects the average domain-wall free energy to scale with the length $[f_{\rm dw}(\ell)]_{M} = \ell\cdot [\sigma_{\rm dw}]_{M}$ for some positive average tension $[\sigma_{\rm dw}]_{M} > 0$ and sufficiently large $\ell$. Here, $[\dots]_{M}$ denotes the average over reference configurations $M$.
However, this does not necessarily imply an exponential decay of the disorder-average of its exponential
\begin{equation}
    \left[ \expval{\mu(u)\mu(v)}_\lambda \right]_M = \left[ e^{-f_{\rm dw}(\ell={\rm dist}(u, v))} \right]_{M}. \label{eq:disorder_op_av}
\end{equation}

Note that our mixing time guarantee implies
\begin{align}
    [ \trel ]_{M} \leq& 8 \Delta w^{-1} \abs{E} \left[ (1 + \chi_2)(2+\abs{V}) + 2\chi_4 \right]_M\\
        =& 8 \Delta w^{-1} \abs{E} \left( (1 + [\chi_2]_M)(2+\abs{V}) + 2[\chi_4]_M \right)
\end{align}
and 
\begin{align}
    [\chi_2]_M =& \left[ \sum_{u, v} \expval{\mu(u) \mu(v)}_\lambda \right]_M \\
        =& \sum_{u, v} \left[ \expval{\mu(u) \mu(v)}_\lambda \right]_M
\end{align}
where in the last line, the sum is now over the disorder-averaged defect operator. The same argument can be made for the four-point susceptibility $\chi_4$.

The average runtime is thus directly related to the question of whether the disorder-averaged disorder operator in \autoref{eq:disorder_op_av}, in the decodable phase, has the same behavior as in the disorder-free case. 
Following the discussion in the previous section, and assuming for simplicity that our error model is (on average) translational invariant and local in some finite euclidean dimension $D$, then the average will depend on a tradeoff between the free-energy gain and the unlikeliness of rare regions of diameter $\ell$. 
In other words, the distribution of $\chi_{2, 4}$ has a power-law tail, and the power of the tail will determine which moments of the distribution diverge exponentially with system size.
In particular, in a system of linear size $L$ we can write 
\begin{align}
    \sum_{u, v} \left[ \expval{\mu(u)\mu(v)}_\lambda \right]_M 
    &\propto L^D\sum_{\ell = 0}^{L} \ell^{D-1} \left[ e^{-\beta f_{\rm dw}(\ell)} \right]_{M} \\
    &= L^D \sum_{\ell = 0}^{L} \ell^{D-1} e^{-(r - \beta\sigma) \ell}
\end{align}
where $r$ is the constant that governs the exponential decay of probability of having disorder configurations gaining free energy $\beta\sigma \ell$. If $r \geq \beta\sigma$, then the above sum (and hence the average mixing time) is bounded by a polynomial function of $L$, while for $r < \beta\sigma$ it is exponentially large in $L$. 

The trade-off between exponentially rare configurations of disorder having exponentially large contributions to certain quantities is exactly what governs Griffiths effects in statistical mechanics \cite{griffiths1969effects}. The situation here is most reminiscent of one encountered in the fermionic representation of the random bond Ising model, where rare disorder regions of the same kind as sketched in \autoref{fig:worm_trap} contribute exponentially to the low-energy density of states \cite{gruzberg2001rbim,mildenberger2006griffiths,pandey2026low}. 
The constants~$r$ and~$\sigma$ are non-universal and depend in general on the details of the disorder distribution as well as the weights~$w$.

The constant $\sigma$ is just the domain wall tension (for a pinned domain wall pinned at two specific sites). This quantity vanishes at the transition into the paramagnet (assuming that the transition is continuous, as it is e.g. in the 2D RBIM), but stays finite in the zero-temperature limit $\beta \to \infty$. Fixing the error rate $p$ (which fixes~$r$), and varying the edge weight $w=\exp(-2\beta)$ (which varies also $\sigma$), the average will be polynomial-bounded ($r \geq \beta\sigma$) close to the transition into the paramagnet at large $w$, while it will diverge exponentially at sufficiently small $w$.
The question of whether decoding is efficient on average is then exactly the question of whether the system enters a `Griffiths phase' (at $r = \beta\sigma$) only below or on the Nishimori line ($w = p/(1-p)$) as the temperature is lowered. If the Nishimori line lies within the Griffiths phase ($r < \beta\sigma$), then decoding will be inefficient at least with the simple scheme presented here.

Studying the exact phase boundary of such Griffiths phases is often subtle and necessitates high-quality numerics. This is especially true here, since at least in principle a bound on both the two-point and four-point defect operators is necessary (although one does not expect the two to behave qualitatively different in practice). We hence leave settling the question of the average mixing time for future work. We also note that even if the average mixing time is bounded by a polynomial in system size for some realistic noise model, to obtain an efficient decoder from this one needs to define a `stopping condition' for the decoder at a given instance. This stopping condition must reliably identify syndromes where the mixing time is large, and one would have to show that the stopping condition reliably terminates the algorithm in polynomial time on average, while not degrading performance. 
Finally, we note that even if the defect susceptibilities are not bounded by any polynomial, this in some cases can be remedied  by introducing a tunable defect fugacity (that is a weight factor $w(u, v)$ for configurations in $\mathcal C_2$) \cite{jerrum_sinclair_vigoda_2004_permanent}. Even if such schemes do not suffice for worst-case instances \cite{stefankovi2018torpid}, they may suffice to make `bad' errors sufficiently rare to guarantee a polynomial average mixing time.

\section{Numerical Results}\label{sec:numerical_results}

We now present numerical results on the decoding performance of the worm algorithm for two classes of matchable decoding problems, both of which lie beyond the reach of existing methods such as tensor-network decoding \cite{chubb2021generaltensornetworkdecoding,Piveteau_2024}. All results are obtained using the \emph{directed} worm algorithm, discussed in \autoref{app:directed_worm}.

First, we present decoding success rates obtained using the worm algorithm for the rotated surface code with independent and identically distributed (i.i.d.) measurement errors together with i.i.d.\ bit-flip noise. 
Second, we present corresponding results for a family of hyperbolic surface codes generated from different compactifications of a fixed infinite tessellation of the hyperbolic plane, assuming perfect syndrome measurements but i.i.d. bit-flip noise.

For the hyperbolic codes in particular, we additionally compare the decoding performance of the worm algorithm to MWPM. The two decoders exhibit strikingly similar performance for these codes and noise models. We discuss how this feature may be related to some of the geometric properties of hyperbolic tilings.

\subsection{The surface code with measurement errors}
We consider the rotated surface code subjected to independent and identically distributed (i.i.d.) bit-flip noise and measurement errors, each occurring with probability \(p\).

To apply the worm decoder in this setting, an additional boundary node must be added to the decoding graph. This is done using the prescription of \autoref{sec:boundary_node}. The resulting decoding performance is shown in \autoref{fig:RSC_ME_data}.

\begin{figure}
\centering
\includegraphics[width=1.0\linewidth]{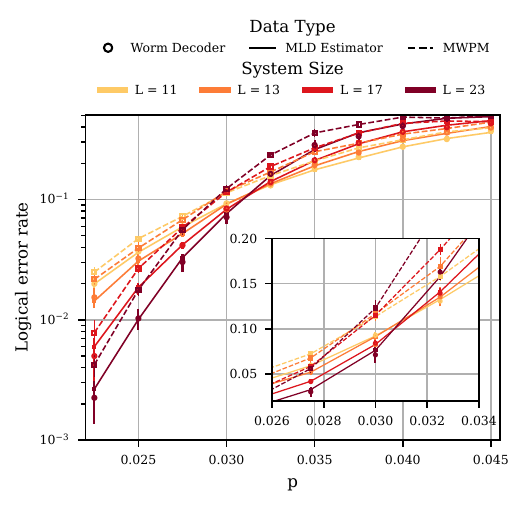}
\caption{Decoding performance of the worm algorithm for the rotated surface code with i.i.d.\ bit-flip noise and measurement errors.}
\label{fig:RSC_ME_data}
\end{figure}

From these data, we estimate a threshold of
\[
p_c = 0.0310 \pm 0.0005,
\]

    which differs slightly from the previously reported value \(p_c \approx 0.033\)~\cite{Ohno_2004}. We expect our estimate to be more reliable.

\subsection{Hyperbolic Surface Codes}

We investigate the performance of the worm decoder for a family of hyperbolic surface codes obtained from \(\{5,5\}\) tilings, under i.i.d. bit-flip noise occurring with probability \(p\), and in the absence of measurement errors. The resulting decoding performance is shown in \autoref{fig:hyp_worm}. From our numerical results, we estimate a threshold of
\[
p_c = 0.019 \pm 0.001.
\]
Moreover, we find empirically that the decoding performance of the worm algorithm closely matches that of minimum-weight perfect matching (MWPM). We discuss how this similarity may be connected to geometric properties of negatively curved spaces, in particular the notion of \(\delta\)-hyperbolicity.

\begin{figure}
\centering
\includegraphics[width=1.0\linewidth]{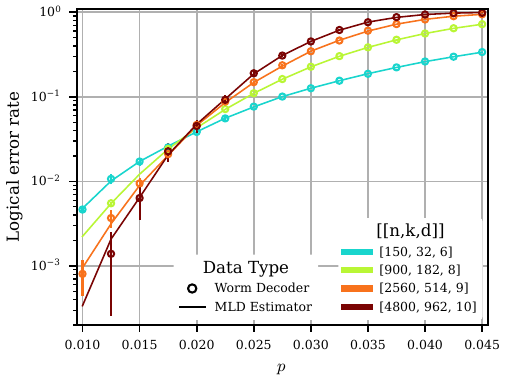}
    \caption{Decoding success rates under the worm algorithm for a family of \(\{5,5\}\) hyperbolic surface codes.}
    \label{fig:hyp_worm}
\end{figure}

Hyperbolic surface codes have the property that the number of encoded qubits \(k\) scales linearly with the number of physical qubits \(n\), and hence the total number of logical sectors grows exponentially with system size. As mentioned in \autoref{sec:worm_decoder}, this can be a problem when trying to perform the maximization over classes in \autoref{eq:ml_decoding}. However, the worm decoder does not attempt to enumerate or store the probability of all possible logical sectors. Instead, the worm automatically \emph{importance samples} the sectors, which are labeled using a $k$-bit string and recorded only when they are encountered during sampling.

The sparse dictionary, \texttt{logicals\_tally} used in line~6 of \autoref{alg:worm_decoder} grows dynamically and contains at most~\(N_{\text{samples}}\) entries. In practice, below threshold, the number of observed logical sectors is much smaller than this bound, reflecting the fact that the posterior distribution over logical sectors is strongly concentrated on a finite set (see also the discussion in \autoref{sec:approximate_optimal}). 

\begin{figure}
\centering
\includegraphics[width=1.0\linewidth]{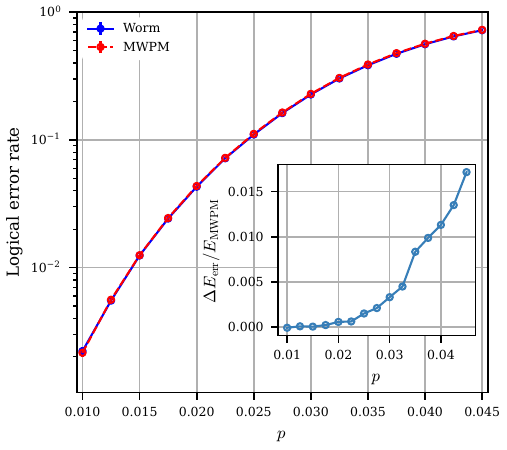}
    \caption{Comparison of decoding success rates for the $[[900, 182, 8]]$ hyperbolic code under MWPM and the worm decoder.}
    \label{fig:hypmwpm}
\end{figure}

Our numerical results reveal a somewhat surprising result regarding the decoding success rates of the worm algorithm for these hyperbolic codes under our  phenomenological noise model, compared to those of MWPM. Empirically, we find that they are \emph{extremely} close (cf.\ \autoref{fig:hypmwpm}).
This suggests that the most-likely-error is almost always in the most likely error class in this setting, in contrast to euclidean surface codes. 
We suspect that this is actually a general feature of the underlying hyperbolic geometry, and in \appref{app:hyperbolicity} we suggest a concrete connection to a property called  $\delta$-hyperbolicity. This property, roughly speaking, means that taking detours relative to a geodesic is expensive, and therefore there are very few equivalent paths that contribute to the weights in \autoref{eq:ml_decoding}. This means the minimum weight error dominates, and one would indeed expect MWPM to be a much better approximation to optimal decoding in the hyperbolic case, than it is in the euclidean setting.

\section{Beyond Matching: Correlated Decoding with the Worm}\label{sec:correl_decoding}
Up to this point, we have focused on \emph{matchable} decoding problems, in which each error mechanism triggers at most two detectors. In this setting, the DEM maps naturally onto a weighted graph, and the worm algorithm can be used to sample error chains consistent with an observed detection event, enabling optimal decoding.

Many physically realistic noise models, however, do not yield matchable decoding problems. Important examples include surface codes subject to circuit-level noise or to single-qubit depolarizing noise, with or without measurement errors. In such cases, individual error mechanisms may trigger more than two detectors, and the corresponding DEM defines a weighted \emph{hypergraph} rather than a graph. As a result, graph-based decoding methods---including the worm algorithm---cannot be applied directly.

A common strategy for addressing this difficulty is to approximate the
hypergraph decoding problem by a graph-based one where correlations between errors are incorporated through an iterative scheme with a suitable reweighting procedure for the edge weights \cite{fowler2013optimalcomplexitycorrectioncorrelated,Delfosse_2014,Paler_2023}. 
These so-called \emph{correlated matching} approaches aim to (heuristically) capture the effects of non-matchable error mechanisms while retaining the computational advantages of graph-based decoding.

In this section, we first introduce a correlated decoding scheme for depolarizing noise based on the worm algorithm. To this end, we make use of the fact that the worm naturally provides \emph{soft information} in the form of empirical single-edge error marginals obtained by sampling over error chains. These marginals can be used to reweight the edges of the effective decoding graph, thereby effectively incorporating correlations between errors. 

We then illustrate this approach using the rotated surface code subject to i.i.d.\ single-qubit depolarizing noise with perfect measurements. In this setting, $Y$ errors correspond to non-matchable error mechanisms, as they trigger four detectors (two $X$ and two $Z$ detectors). As a result, the presence of $Y$ errors induce correlations between the $X$- and $Z$-decoding problems, which we capture through an appropriate reweighting scheme. We show numerically that this method achieves a significantly higher threshold, and improved low-noise performance, compared to existing correlated MWPM techniques. 

At the end of the section, we discuss how the same strategy may be generalized to more complex noise models.

\subsection{Correlated decoding of surface codes with depolarizing noise}

We consider a (rotated) surface code subject to depolarizing noise of strength $p$. 
A single decoding experiment then consists of application of the single-qubit channel
\begin{equation}
    \Phi(\rho)
    =
    (1-p)\rho
    + \frac{p}{3}\bigl( X\rho X + Y\rho Y + Z\rho Z \bigr),
\end{equation}
followed by measurements of all $X$- and $Z$-type stabilizers. We refer to the outcomes of these measurements as the $X$-syndrome $S_X$ and the $Z$-syndrome $S_Z$. Without loss of generality, we assume that the goal of the experiment is to successfully decode the $X$-type error configuration.

The $X$- and $Z$-decoding problems are correlated, since $iY = XZ$, that is, up to a phase, a $Y$ error corresponds to a simultaneous bit- and phase-flip. In the DEM representation, this correlation appears as a four-body hyperedge: a single-qubit $Y$ error produces two $X$ detection events and two $Z$ detection events.

This four-body hyperedge can be decomposed into a pair of two-body edges—one belonging to the $X$-decoding graph and the other to the $Z$-decoding graph. This decomposition induces a natural one-to-one mapping between edges of the $Z$-decoding graph and corresponding edges of the $X$-decoding graph. For an edge $e$ in the $Z$-decoding graph, we denote the corresponding edge in the $X$-decoding graph by $\tilde{e}$.

In order to construct a reweighting scheme for this decoding problem, it is necessary to quantify the correlations between the $X$ and $Z$ decoding problems induced by this decomposition.

More generally, the procedure of decomposing non-matchable noise models into matchable ones follows the same principle~\cite{Paler_2023}. The matchable error mechanisms—those that trigger at most two detectors—define a natural edge basis for the decoding graph, while non-matchable mechanisms correspond to hyperedges that can be decomposed onto this basis. 

It is precisely this decomposition that induces correlations between edges, which can then be incorporated through an appropriate reweighting scheme. A special feature of the depolarizing noise model considered here is that the resulting approximate decoding graph separates into two disconnected components: the $X$ and $Z$ decoding graphs.

\subsubsection{Quantifying correlations}

Consider the conditional probabilities for a single qubit to experience an $X$-error given whether or not it has a $Z$-error. Writing $x_{\tilde{e}}, z_e \in \{0,1\}$ for the binary random variables indicating the presence of $X$ and $Z$ errors on the qubit associated with the pair $(e,\tilde{e})$, and let $\mathbb{P}(x_{\tilde{e}} = A \mid z_e = B)$ correspond to the probability that $x_e = A$ given $z_e = B$. A direct calculation from the depolarizing channel yields
\begin{eqnarray}
    \mathbb{P}(x_{\tilde{e}} = 1 \mid z_e = 1)
    &=& \mathbb{P}(x_{\tilde{e}} = 0 \mid z_e = 1) = \frac{1}{2}, \\
    \mathbb{P}(x_{\tilde{e}} = 1 \mid z_e = 0)
    &=& \frac{p/3}{1 - 2p/3}.
\end{eqnarray}
Suppose that for each qubit \(e\), we are told the probability that it has a $Z$ error is, 
$\mathbb{P}(z_e = 1) = \alpha_e \in [0,1]$. Given this knowledge, the single-qubit $X$-error probability becomes
\begin{eqnarray}
    \mathbb{P}(x_{\tilde{e}} = 1)\label{eq:quantify_correl}
    &=&
    \mathbb{P}(x_{\tilde{e}} = 1 \mid z_e = 1)\,\alpha_e
    \nonumber\\
    &&\quad
    + \mathbb{P}(x_{\tilde{e}} = 1 \mid z_e = 0)\,(1 - \alpha_e)
    \\
    &=&
    \frac{\alpha_e}{2}
    + (1 - \alpha_e)\,\frac{p/3}{1 - 2p/3}.
\end{eqnarray}

The idea behind our correlated decoding scheme is that using the worm we can estimate the quantity $\alpha_e$ for all edges of the $Z$-decoding graph. This can be used to then reweight each edge of the $X$ decoding graph individually.

\subsubsection{Correlated decoding algorithm}

We use the worm algorithm to estimate $\alpha_e$ for each edge of the $Z$-decoding problem, given no information about the $X$-decoding problem. These estimated marginals are then used to reweight the corresponding edges in the $X$-decoding graph.

In the absence of any $X$-information, each qubit has the same prior probability of experiencing a $Z$ error, namely $p_Z = 2p/3$. Consequently, all edges of the $Z$-decoding graph $G_Z$ are weighted accordingly.

Sampling over error chains defined by this weighted graph, together with the measured syndrome $S_Z$, yields estimates of the single-qubit $Z$-error marginals, $\alpha = \{\alpha_e\}$.

For each corresponding edge $\tilde{e}$ in the $X$-decoding problem, we then assign a reweighted edge
\begin{equation}
    w^X_{\tilde{e}}
    =
    \frac{\mathbb{P}(x_{\tilde{e}} = 1)}{1 - \mathbb{P}(x_{\tilde{e}} = 1)},
\end{equation}
where $\mathbb{P}(x_{\tilde{e}} = 1)$ is given by \autoref{eq:quantify_correl}. Finally, the worm decoder is applied to this reweighted $X$-decoding graph. The full procedure is summarized in Algorithm~\ref{alg:correl_worm_decoder}.

\begin{algorithm}[H]
\caption{Correlated worm decoder (Depolarizing noise)}
\label{alg:correl_worm_decoder}
\begin{algorithmic}[1]

\State \textbf{Input:} 
$G_X=(V_X,E_X,w^X)$, $G_Z=(V_Z,E_Z,w^Z)$; 
edge map $f: e \mapsto \tilde{e}$; 
syndromes $S_X,S_Z$; 
logical representatives $\mathcal{L}_X$; 
sample counts $N_Z,N_X$; 
autocorrelation time $t_{\text{auto}}$.

\State \textbf{Output:} Recovery $R_X$.

\State $M_0^Z \gets \Call{ReferenceMatching}{G_Z, S_Z}$
\State Reweight $w^z_e \mapsto \tilde{w}^Z_e$, to obtain $\tilde{G}_Z = (V_Z,E_Z,\tilde{w}^Z)$.
\State $\Sigma \gets 0^{|E_Z|}$, \quad $\alpha \gets 0^{|E_Z|}$

\For{$n = 1$ to $N_Z$}
    \For{$t = 1$ to $t_{\text{auto}}$}
        \State $\Sigma \gets \Call{WormUpdate}{\tilde{G}_Z,\Sigma}$
    \EndFor
    \State $\Lambda \gets \Sigma \oplus M_0^Z$
    
    \State $\alpha[e] \gets \alpha[e] + \Lambda[e]/N_Z \quad \forall e \in E_Z$
\EndFor

\For{$e \in E_Z$}
    \State Compute $P(x_{\tilde{e}} = 1) 
    = \frac{\alpha[e]}{2} 
      + (1-\alpha[e])\frac{p/3}{1-2p/3}$
    \State $w^X_{\tilde{e}} \gets 
    P(x_{\tilde{e}}=1)/(1-P(x_{\tilde{e}}=1))$
\EndFor

\State $R_X \gets \Call{WormDecoder}{G_X, S_X, t_{\text{auto}}, \mathcal{L}_X, N_X}$
\State \Return $R_X$

\end{algorithmic}
\end{algorithm}

\begin{figure}
\centering
\includegraphics[width=1.0\linewidth]{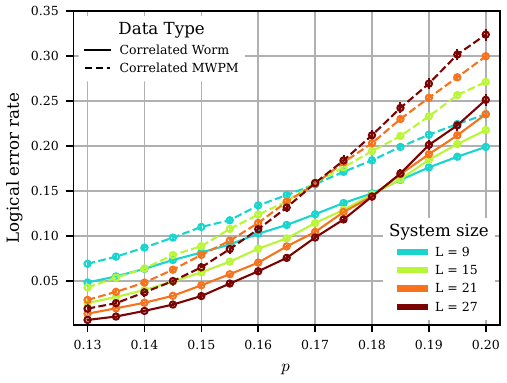}
\caption{Plot showing decoding success rates of the correlated worm decoder for the rotated surface code under depolarizing noise and perfect measurements, for a variety of system sizes and noise strengths $p$.}
    \label{fig:correl_threshold}
\end{figure}

\begin{figure}
    \centering
    \includegraphics[width=1.0\linewidth]{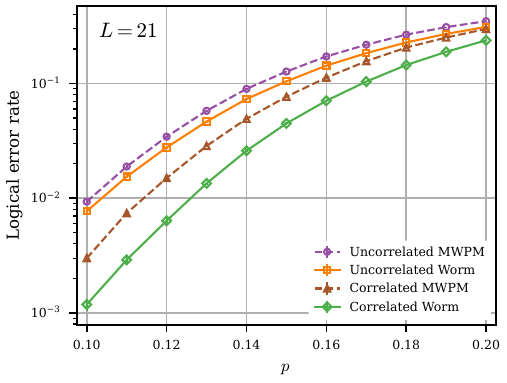}
    \caption{Plot showing the low noise regime behavior of the correlated worm decoder compared to other decoding methods. All data is for $L = 21$.}
    \label{fig:correl_low_p}
\end{figure}

We present numerical results on the performance of the correlated worm decoding scheme for the rotated surface code with perfect measurements. We observe a significant improvement in the decoding threshold compared to correlated MWPM, as well as a substantial enhancement in decoding performance in the low-noise regime. These results are shown in \autoref{fig:correl_threshold} and
\autoref{fig:correl_low_p}, respectively.

In \appref{App:correlated_parsing}, we introduce an iterative parsing scheme in which the correlated reweighting procedure outlined in \autoref{alg:correl_worm_decoder} is applied multiple times. In this approach, the decoder alternates between the $X$- and $Z$-decoding problems, using the worm algorithm at each stage to update edge weights based on newly estimated marginal information. We find that this iterative procedure yields a modest but consistent further improvement in decoding success rates.

Overall, these results demonstrate that exploiting the soft information produced by the worm algorithm yields clear benefits over existing correlated MWPM-based decoding methods. In particular, we observe both a noticeable increase in the threshold and a marked improvement in decoding performance at low physical error rates.

We now turn to the question of how this correlated worm decoding strategy can be generalized to other non-matchable decoding problems. We expect that similar gains can be achieved for a broad class of noise models, potentially leading to systematic improvements over graph-based correlated decoding schemes based on MWPM.

\subsection{Generalizing to all detector models}\label{sec:general_correlated}

We now discuss how the correlated decoding approach introduced for the surface code with depolarizing noise and perfect measurements can be extended to more general non-matchable noise models.

As noted previously, depolarizing noise is special in that it naturally decomposes into two disconnected—but correlated—decoding problems. For more general noise models, such as circuit-level noise, this separation does not occur. In these cases, mapping the DEM from a hypergraph to a correlated graph produces a \emph{single} connected decoding graph whose edges exhibit nontrivial
correlations.

The first step in this generalized setting is to project the hypergraph DEM onto a correlated decoding graph. This is achieved by identifying a basis of edges corresponding to the matchable error mechanisms (those that trigger exactly one or two detectors) of the original decoding
problem~\cite{Paler_2023}. For the surface code with depolarizing noise, these correspond to single-qubit \(X\) and \(Z\) error mechanisms. All hyperedges are then decomposed onto this edge basis, with their contributions incorporated by adjusting the resulting edge weights appropriately. In addition, it is necessary to retain information about which error mechanisms in the original
hypergraph contribute to each edge of the correlated decoding graph.

After this projection step, one obtains a correlated decoding graph in which all correlations are captured implicitly through the decomposition of the original hypergraph. At this stage, the worm algorithm can be applied directly to this graph, allowing sampling over error chains and the estimation of single-edge error marginals.

The next stage consists of a reweighting procedure. Using the empirically estimated edge marginals together with the stored decomposition data, the edge weights of the correlated decoding graph are updated so as to enforce the
correlations induced by non-matchable error mechanisms.

Finally, decoding is performed on this reweighted graph, again using the worm algorithm. As in the depolarizing noise case, this procedure can be straightforwardly extended into an iterative reweighting scheme, where sampling and reweighting are alternated for multiple rounds before a final decoding step is applied.

The construction presented here is schematic but concrete, and outlines a path toward a fully general correlated worm decoding framework. We believe that such a scheme can effectively exploit the soft information
produced by the worm algorithm. A detailed numerical study of scenarios relevant for practical decoding problems, and investigation of optimal projection and reweighting strategies is left for future work.

\section{Discussion and conclusions\label{sec:conclusion}}

In this paper, we have introduced a new algorithm, which efficiently approximates optimal decoding of matchable codes, and which we call the ``worm decoder''. It uses a Markov Chain Monte Carlo algorithm (``the worm'') to sample from the conditional distribution of errors given a syndrome.

The efficiency of this procedure is contingent on fast mixing time of the underlying worm process. We presented a rigorous guarantee for fast mixing conditional on a condition which we call `bounded defect susceptibility'. We connected this notion to the statistical physics of disorder operators in the corresponding statistical mechanics model, and argued fast mixing (and hence efficiency of the worm) for typical errors all the way up to the optimal decoding threshold, implying that the worm realizes this threshold up to finite sample-size errors. 
We also connected the average mixing time of the worm process to the location of the phase boundary of a specific kind of previously unstudied Griffiths phase in the statistical mechanics model. Studying its exact phase boundary remains important future work.

Finally, we also demonstrated the effectiveness of our decoder numerically by performing decoding of both the surface code with measurement errors, and a family of hyperbolic surface codes with constant rate, up to distances $d=23$ and $d=10$, respectively. 
We highlight that the worm algorithm does not only output a recovery operation, but also soft information that is useful for further processing. We demonstrated the usefulness of this soft information explicitly by using it to extend the applicability of the worm algorithm beyond just matchable codes. 
To this end we introduced an iterative reweighting scheme similar to correlated matching and showed significant performance improvement over both uncorrelated worm decoding, as well as correlated matching. 
This demonstrates the feasibility and usefulness of this ``correlated worm'' decoding in the simplest non-trivial setting. 
To move towards utility in realistic decoding scenarios, future work should explore more realistic noise models, and benchmark different reweighting schemes in this setting.

Our paper demonstrates that Markov Chain Monte Carlo (MCMC) methods can be highly competitive in achieving high decoding performance for certain noise models and codes. Crucial to this success is the fact that the worm algorithm works in an extended configuration space, enabling highly non-local moves.  
More generally, ``cluster updates'' have been highly successful in enabling simulation of constrained systems in classical statistical mechanics \cite{wang_swedsen1990review,krauth2003review,ninarello2017review}, but, to the best of our knowledge, have not been previously explored for decoding. 
Cluster updates yield big speedups when they exploit some underlying structure of the model that is being sampled \cite{swedsen_wang1987algorithm,wolff1989algorithm,krauth2003pocket,grigera2001swap,Sandvik_2006,placke2024fracton}. 
In the case presented here, it is the ``matchability'' of the problem that makes the worm useful. A natural question is whether one can design effective update rules, and hence MCMC decoders, for other structured decoding problems.

Another interesting, numerical finding of our work is the fact that the performance gap between optimal (maximum-likelihood) and most-likely-error (MLE) decoding seems to be very small for hyperbolic surface codes (see \autoref{fig:hypmwpm}). We argued that this is a direct consequence of the hyperbolicity of the underlying space, connecting it to ``\emph{$\delta$-hyperbolicity}'', which, roughly speaking, implies that geodesics and almost-geodesics are much less degenerate in hyperbolic space than in Euclidean spaces. 
One natural question is whether one can argue for this relation more rigorously, and, for example, derive explicit asymptotic bounds on the difference between optimal and MLE decoding.
Another interesting question is whether such an implication could generalize to other qLDPC codes. Asymptotically good qLDPC codes, for example, are based on high-dimensional expanders, and they have the property that the size of syndromes grows linearly with the size of (sufficiently small) errors. It would be intriguing if such generalized notions of hyperbolicity would also imply that most-likely-error decoding is close to optimal.

\vspace{1em}

\textbf{Acknowledgements}---We thank Thiago Bergamaschi, Oscar Higgot, David Huse, Harsh Patil, Vedika Khemani, Matt McEwen, and Roderich Moessner for useful discussions. We are particularly grateful to Andreas Galanis for valuable discussions about mixing-time guarantees for the worm process, and for providing the reduction from approximate sampling of even subgraphs to approximate counting of perfect matchings. We are also particularly grateful to Akshat Pandey for many helpful discussions about rare-region effects and the distribution of domain wall free energies in disordered spin models.\\
Z.T.\ acknowledges funding through Vedika Khemani who acknowledges support from the Office
of Naval Research Young Investigator Program (ONR
YIP) under Award Number N000142412098 and from
the Packard Foundation through a Packard Fellowship
in Science and Engineering.
B.P.\ acknowledges funding through a Leverhulme-Peierls Fellowship at the University of Oxford and the Alexander von Humboldt foundation through a Feodor-Lynen fellowship.\\
Part of this research was conducted while Z.T. and B.P were visiting the Kavli Institute for Theoretical Physics (KITP) supported by grant NSF PHY-2309135.
Part of this work was conducted while N.P.B.\ was visiting the Simons Institute for the Theory of Computing supported by DOE
QSA grant \#FP00010905.

\vspace{1em}

\textbf{Code and Data availability}---The simulation data and code necessary to reproduce the results of this work are available at
\url{https://github.com/zac-tobias/Optimal-Decoding-with-the-Worm-Data}.

\bibliography{references} %

@misc{fowler2013optimalcomplexitycorrectioncorrelated,
      title={Optimal complexity correction of correlated errors in the surface code}, 
      author={Austin G. Fowler},
      year={2013},
      eprint={1310.0863},
      archivePrefix={arXiv},
      primaryClass={quant-ph},
      url={https://arxiv.org/abs/1310.0863}, 
}

@article{terhal2015quantum,
  title={Quantum error correction for quantum memories},
  author={Terhal, Barbara M},
  journal={Reviews of Modern Physics},
  volume={87},
  number={2},
  pages={307--346},
  year={2015},
  publisher={APS}
}

@book{bridson2013metric,
  title={Metric spaces of non-positive curvature},
  author={Bridson, Martin R and Haefliger, Andr{\'e}},
  volume={319},
  year={2013},
  publisher={Springer Science \& Business Media}
}

@article{kitaev1997quantum,
  title={Quantum computations: algorithms and error correction},
  author={Kitaev, A Yu},
  journal={Russian Mathematical Surveys},
  volume={52},
  number={6},
  pages={1191},
  year={1997},
  publisher={IOP Publishing}
}

@article{wu2025minimum,
  title={Minimum-weight parity factor decoder for quantum error correction},
  author={Wu, Yue and Li, Binghong and Chang, Kathleen and Puri, Shruti and Zhong, Lin},
  journal={arXiv preprint arXiv:2508.04969},
  year={2025}
}

@article{beni2025tesseract,
  title={Tesseract: A search-based decoder for quantum error correction},
  author={Beni, Laleh Aghababaie and Higgott, Oscar and Shutty, Noah},
  journal={arXiv preprint arXiv:2503.10988},
  year={2025}
}

@article{bravyi2014efficient,
  title={Efficient algorithms for maximum likelihood decoding in the surface code},
  author={Bravyi, Sergey and Suchara, Martin and Vargo, Alexander},
  journal={arXiv preprint arXiv:1405.4883},
  year={2014}
}

@article{iyer2015hardness,
  title={Hardness of decoding quantum stabilizer codes},
  author={Iyer, Pavithran and Poulin, David},
  journal={IEEE Transactions on Information Theory},
  volume={61},
  number={9},
  pages={5209--5223},
  year={2015},
  publisher={IEEE}
}

@article{freedman2001projective,
  title={Projective plane and planar quantum codes},
  author={Freedman, Michael H and Meyer, David A},
  journal={Foundations of Computational Mathematics},
  volume={1},
  number={3},
  pages={325--332},
  year={2001},
  publisher={Springer}
}

@article{bravyi1998quantum,
  title={Quantum codes on a lattice with boundary},
  author={Bravyi, Sergey B and Kitaev, A Yu},
  journal={arXiv preprint quant-ph/9811052},
  year={1998}
}

@incollection{kitaev1997quantum2,
  title={Quantum error correction with imperfect gates},
  author={Kitaev, A Yu},
  booktitle={Quantum communication, computing, and measurement},
  pages={181--188},
  year={1997},
  publisher={Springer}
}

@inproceedings{aharonov1997fault,
  title={Fault-tolerant quantum computation with constant error},
  author={Aharonov, Dorit and Ben-Or, Michael},
  booktitle={Proceedings of the twenty-ninth annual ACM symposium on Theory of computing},
  pages={176--188},
  year={1997}
}

@article{Bombin_2012,
   title={Strong Resilience of Topological Codes to Depolarization},
   volume={2},
   ISSN={2160-3308},
   url={http://dx.doi.org/10.1103/PhysRevX.2.021004},
   DOI={10.1103/physrevx.2.021004},
   number={2},
   journal={Physical Review X},
   publisher={American Physical Society (APS)},
   author={Bombin, H. and Andrist, Ruben S. and Ohzeki, Masayuki and Katzgraber, Helmut G. and Martin-Delgado, M. A.},
   year={2012},
   month=apr }

@misc{chubb2021generaltensornetworkdecoding,
      title={General tensor network decoding of 2D Pauli codes}, 
      author={Christopher T. Chubb},
      year={2021},
      eprint={2101.04125},
      archivePrefix={arXiv},
      primaryClass={quant-ph},
      url={https://arxiv.org/abs/2101.04125}, 
}

@article{PRXQuantum.5.040303,
  title = {Tensor-Network Decoding Beyond 2D},
  author = {Piveteau, Christophe and Chubb, Christopher T. and Renes, Joseph M.},
  journal = {PRX Quantum},
  volume = {5},
  issue = {4},
  pages = {040303},
  numpages = {20},
  year = {2024},
  month = {Oct},
  publisher = {American Physical Society},
  doi = {10.1103/PRXQuantum.5.040303},
  url = {https://link.aps.org/doi/10.1103/PRXQuantum.5.040303}
}

@article{Ohno_2004,
   title={Phase structure of the random-plaquette gauge model: accuracy threshold for a toric quantum memory},
   volume={697},
   ISSN={0550-3213},
   url={http://dx.doi.org/10.1016/j.nuclphysb.2004.07.003},
   DOI={10.1016/j.nuclphysb.2004.07.003},
   number={3},
   journal={Nuclear Physics B},
   publisher={Elsevier BV},
   author={Ohno, Takuya and Arakawa, Gaku and Ichinose, Ikuo and Matsui, Tetsuo},
   year={2004},
   month=oct, pages={462–480} }

@article{Wang_2005,
   title={Worm algorithm for two-dimensional spin glasses},
   volume={72},
   ISSN={1550-2376},
   url={http://dx.doi.org/10.1103/PhysRevE.72.036706},
   DOI={10.1103/physreve.72.036706},
   number={3},
   journal={Physical Review E},
   publisher={American Physical Society (APS)},
   author={Wang, Jian-Sheng},
   year={2005},
   month=sep }

@article{Collevecchio_2016,
   title={The Worm Process for the Ising Model is Rapidly Mixing},
   volume={164},
   ISSN={1572-9613},
   url={http://dx.doi.org/10.1007/s10955-016-1572-2},
   DOI={10.1007/s10955-016-1572-2},
   number={5},
   journal={Journal of Statistical Physics},
   publisher={Springer Science and Business Media LLC},
   author={Collevecchio, Andrea and Garoni, Timothy M. and Hyndman, Timothy and Tokarev, Daniel},
   year={2016},
   month=jul, pages={1082–1102} }

@article{Prokof_ev_2001,
   title={Worm Algorithms for Classical Statistical Models},
   volume={87},
   ISSN={1079-7114},
   url={http://dx.doi.org/10.1103/PhysRevLett.87.160601},
   DOI={10.1103/physrevlett.87.160601},
   number={16},
   journal={Physical Review Letters},
   publisher={American Physical Society (APS)},
   author={Prokof’ev, Nikolay and Svistunov, Boris},
   year={2001},
   month=sep }

@article{Dennis_2002,
   title={Topological quantum memory},
   volume={43},
   ISSN={1089-7658},
   url={http://dx.doi.org/10.1063/1.1499754},
   DOI={10.1063/1.1499754},
   number={9},
   journal={Journal of Mathematical Physics},
   publisher={AIP Publishing},
   author={Dennis, Eric and Kitaev, Alexei and Landahl, Andrew and Preskill, John},
   year={2002},
   month=sep, pages={4452–4505} }

@article{Paler_2023,
   title={Pipelined correlated minimum weight perfect matching of the surface code},
   volume={7},
   ISSN={2521-327X},
   url={http://dx.doi.org/10.22331/q-2023-12-12-1205},
   DOI={10.22331/q-2023-12-12-1205},
   journal={Quantum},
   publisher={Verein zur Forderung des Open Access Publizierens in den Quantenwissenschaften},
   author={Paler, Alexandru and Fowler, Austin G.},
   year={2023},
   month=dec, pages={1205} }

@misc{chen2025scalableaccuracygainspostselection,
      title={Scalable accuracy gains from postselection in quantum error correcting codes}, 
      author={Hongkun Chen and Daohong Xu and Grace M. Sommers and David A. Huse and Jeff D. Thompson and Sarang Gopalakrishnan},
      year={2025},
      eprint={2510.05222},
      archivePrefix={arXiv},
      primaryClass={cond-mat.stat-mech},
      url={https://arxiv.org/abs/2510.05222}, 
}

@article{pattison2021improved,
  title={Improved quantum error correction using soft information},
  author={Pattison, Christopher A and Beverland, Michael E and da Silva, Marcus P and Delfosse, Nicolas},
  journal={arXiv preprint arXiv:2107.13589},
  year={2021}
}

@inproceedings{Delfosse_2014,
   title={A decoding algorithm for CSS codes using the X/Z correlations},
   url={http://dx.doi.org/10.1109/ISIT.2014.6874997},
   DOI={10.1109/isit.2014.6874997},
   booktitle={2014 IEEE International Symposium on Information Theory},
   publisher={IEEE},
   author={Delfosse, Nicolas and Tillich, Jean-Pierre},
   year={2014},
   month=jun, pages={1071–1075} }

@article{PROKOFEV1998253,
title = {“Worm” algorithm in quantum Monte Carlo simulations},
journal = {Physics Letters A},
volume = {238},
number = {4},
pages = {253-257},
year = {1998},
issn = {0375-9601},
doi = {https://doi.org/10.1016/S0375-9601(97)00957-2},
url = {https://www.sciencedirect.com/science/article/pii/S0375960197009572},
author = {N.V Prokof'ev and B.V Svistunov and I.S Tupitsyn}
}

@article{deMarti_iOlius_2024,
   title={Decoding algorithms for surface codes},
   volume={8},
   ISSN={2521-327X},
   url={http://dx.doi.org/10.22331/q-2024-10-10-1498},
   DOI={10.22331/q-2024-10-10-1498},
   journal={Quantum},
   publisher={Verein zur Forderung des Open Access Publizierens in den Quantenwissenschaften},
   author={deMarti iOlius, Antonio and Fuentes, Patricio and Orús, Román and Crespo, Pedro M. and Etxezarreta Martinez, Josu},
   year={2024},
   month=oct, pages={1498} }

@article{Edmonds_1965, title={Paths, Trees, and Flowers}, volume={17}, DOI={10.4153/CJM-1965-045-4}, journal={Canadian Journal of Mathematics}, author={Edmonds, Jack}, year={1965}, pages={449–467}}

@article{Wang_2003,
   title={Confinement-Higgs transition in a disordered gauge theory and the accuracy threshold for quantum memory},
   volume={303},
   ISSN={0003-4916},
   url={http://dx.doi.org/10.1016/S0003-4916(02)00019-2},
   DOI={10.1016/s0003-4916(02)00019-2},
   number={1},
   journal={Annals of Physics},
   publisher={Elsevier BV},
   author={Wang, Chenyang and Harrington, Jim and Preskill, John},
   year={2003},
   month=jan, pages={31–58} }

@article{Chubb_2021,
   title={Statistical mechanical models for quantum codes with correlated noise},
   volume={8},
   ISSN={2308-5835},
   url={http://dx.doi.org/10.4171/AIHPD/105},
   DOI={10.4171/aihpd/105},
   number={2},
   journal={Annales de l’Institut Henri Poincaré D, Combinatorics, Physics and their Interactions},
   publisher={European Mathematical Society - EMS - Publishing House GmbH},
   author={Chubb, Christopher T. and Flammia, Steven T.},
   year={2021},
   month=may, pages={269–321} }

@article{gidney2021stim,
  doi = {10.22331/q-2021-07-06-497},
  url = {https://doi.org/10.22331/q-2021-07-06-497},
  title = {Stim: a fast stabilizer circuit simulator},
  author = {Gidney, Craig},
  journal = {{Quantum}},
  issn = {2521-327X},
  publisher = {{Verein zur F{\"{o}}rderung des Open Access Publizierens
                in den Quantenwissenschaften}},
  volume = 5,
  pages = 497,
  month = jul,
  year = 2021
}

@article{old2023generalized,
  title={Generalized Belief Propagation Decoders for Quantum Error-Correcting Codes},
  author={Old, Julien and Rispler, Manuel},
  journal={Quantum},
  volume={7},
  pages={1037},
  year={2023}
}

@article{Piveteau_2024,
   title={Tensor-Network Decoding Beyond 2D},
   volume={5},
   ISSN={2691-3399},
   url={http://dx.doi.org/10.1103/PRXQuantum.5.040303},
   DOI={10.1103/prxquantum.5.040303},
   number={4},
   journal={PRX Quantum},
   publisher={American Physical Society (APS)},
   author={Piveteau, Christophe and Chubb, Christopher T. and Renes, Joseph M.},
   year={2024},
   month=oct }

@misc{gottesman1997stabilizercodesquantumerror,
      title={Stabilizer Codes and Quantum Error Correction}, 
      author={Daniel Gottesman},
      year={1997},
      eprint={quant-ph/9705052},
      archivePrefix={arXiv},
      primaryClass={quant-ph},
      url={https://arxiv.org/abs/quant-ph/9705052}, 
}

@misc{wichette2025partitionfunctionframeworkestimating,
      title={A partition function framework for estimating logical error curves in stabilizer codes}, 
      author={Leon Wichette and Hans Hohenfeld and Elie Mounzer and Linnea Grans-Samuelsson},
      year={2025},
      eprint={2505.15758},
      archivePrefix={arXiv},
      primaryClass={quant-ph},
      url={https://arxiv.org/abs/2505.15758}, 
}

@article{Wagner2022paulichannelscanbe,
  doi = {10.22331/q-2022-09-19-809},
  url = {https://doi.org/10.22331/q-2022-09-19-809},
  title = {Pauli channels can be estimated from syndrome measurements in quantum error correction},
  author = {Wagner, Thomas and Kampermann, Hermann and Bru{\ss{}}, Dagmar and Kliesch, Martin},
  journal = {{Quantum}},
  issn = {2521-327X},
  publisher = {{Verein zur F{\"{o}}rderung des Open Access Publizierens in den Quantenwissenschaften}},
  volume = {6},
  pages = {809},
  month = sep,
  year = {2022}
}

@article{Wagner_2021,
   title={Optimal noise estimation from syndrome statistics of quantum codes},
   volume={3},
   ISSN={2643-1564},
   url={http://dx.doi.org/10.1103/PhysRevResearch.3.013292},
   DOI={10.1103/physrevresearch.3.013292},
   number={1},
   journal={Physical Review Research},
   publisher={American Physical Society (APS)},
   author={Wagner, Thomas and Kampermann, Hermann and Bruß, Dagmar and Kliesch, Martin},
   year={2021},
   month=mar }

@article{Hastings2021dynamically,
  doi = {10.22331/q-2021-10-19-564},
  url = {https://doi.org/10.22331/q-2021-10-19-564},
  title = {Dynamically {G}enerated {L}ogical {Q}ubits},
  author = {Hastings, Matthew B. and Haah, Jeongwan},
  journal = {{Quantum}},
  issn = {2521-327X},
  publisher = {{Verein zur F{\"{o}}rderung des Open Access Publizierens in den Quantenwissenschaften}},
  volume = {5},
  pages = {564},
  month = oct,
  year = {2021}
}

@misc{kribs2006operatorquantumerrorcorrection,
      title={Operator quantum error correction}, 
      author={David W. Kribs and Raymond Laflamme and David Poulin and Maia Lesosky},
      year={2006},
      eprint={quant-ph/0504189},
      archivePrefix={arXiv},
      primaryClass={quant-ph},
      url={https://arxiv.org/abs/quant-ph/0504189}, 
}

@article{Poulin_2005,
   title={Stabilizer Formalism for Operator Quantum Error Correction},
   volume={95},
   ISSN={1079-7114},
   url={http://dx.doi.org/10.1103/PhysRevLett.95.230504},
   DOI={10.1103/physrevlett.95.230504},
   number={23},
   journal={Physical Review Letters},
   publisher={American Physical Society (APS)},
   author={Poulin, David},
   year={2005},
   month=dec }

@article{hutter2014efficient,
  title = {Efficient Markov chain Monte Carlo algorithm for the surface code},
  author = {Hutter, Adrian and Wootton, James R. and Loss, Daniel},
  journal = {Phys. Rev. A},
  volume = {89},
  issue = {2},
  pages = {022326},
  numpages = {10},
  year = {2014},
  month = {Feb},
  publisher = {American Physical Society},
  doi = {10.1103/PhysRevA.89.022326},
  url = {https://link.aps.org/doi/10.1103/PhysRevA.89.022326}
}

@article{wootton2021high,
  title = {High Threshold Error Correction for the Surface Code},
  author = {Wootton, James R. and Loss, Daniel},
  journal = {Phys. Rev. Lett.},
  volume = {109},
  issue = {16},
  pages = {160503},
  numpages = {5},
  year = {2012},
  month = {Oct},
  publisher = {American Physical Society},
  doi = {10.1103/PhysRevLett.109.160503},
  url = {https://link.aps.org/doi/10.1103/PhysRevLett.109.160503}
}

@misc{lee2025efficient,
Author = {Seok-Hyung Lee and Lucas English and Stephen D. Bartlett},
Title = {Efficient Post-Selection for General Quantum LDPC Codes},
Year = {2025},
Eprint = {arXiv:2510.05795},
}

@article{Sandvik_2006,
   title={Correlations and confinement in nonplanar two-dimensional dimer models},
   volume={73},
   ISSN={1550-235X},
   url={http://dx.doi.org/10.1103/PhysRevB.73.144504},
   DOI={10.1103/physrevb.73.144504},
   number={14},
   journal={Physical Review B},
   publisher={American Physical Society (APS)},
   author={Sandvik, Anders W. and Moessner, R.},
   year={2006},
   month=apr }

@article{10.1063/1.1703953,
    author = {Kasteleyn, P. W.},
    title = {Dimer Statistics and Phase Transitions},
    journal = {Journal of Mathematical Physics},
    volume = {4},
    number = {2},
    pages = {287-293},
    year = {1963},
    month = {02},
    issn = {0022-2488},
    doi = {10.1063/1.1703953},
    url = {https://doi.org/10.1063/1.1703953}
}

@article{10.1063/1.1704825,
    author = {Fisher, Michael E.},
    title = {On the Dimer Solution of Planar Ising Models},
    journal = {Journal of Mathematical Physics},
    volume = {7},
    number = {10},
    pages = {1776-1781},
    year = {1966},
    month = {10},
    issn = {0022-2488},
    doi = {10.1063/1.1704825},
    url = {https://doi.org/10.1063/1.1704825}
}

@misc{wang2024dgrtacklingdriftedcorrelated,
      title={DGR: Tackling Drifted and Correlated Noise in Quantum Error Correction via Decoding Graph Re-weighting}, 
      author={Hanrui Wang and Pengyu Liu and Yilian Liu and Jiaqi Gu and Jonathan Baker and Frederic T. Chong and Song Han},
      year={2024},
      eprint={2311.16214},
      archivePrefix={arXiv},
      primaryClass={quant-ph},
      url={https://arxiv.org/abs/2311.16214}, 
}

@article{jerrum_sinclair_1989_permanent,
author = {Jerrum, Mark and Sinclair, Alistair},
title = {Approximating the Permanent},
journal = {SIAM Journal on Computing},
volume = {18},
number = {6},
pages = {1149-1178},
year = {1989},
doi = {10.1137/0218077},
URL = { https://doi.org/10.1137/0218077 },
eprint = { https://doi.org/10.1137/0218077 }
}

@BOOK{jerrum2003book,
  title     = "Counting, sampling and integrating: Algorithms and complexity",
  author    = "Jerrum, Mark",
  publisher = "Birkhauser Verlag AG",
  series    = "Lectures in Mathematics. ETH Z{\"u}rich",
  edition   =  2003,
  month     =  jan,
  year      =  2003,
  address   = "Basel, Switzerland"
}

@article{jerrum_sinclair_vigoda_2004_permanent,
author = {Jerrum, Mark and Sinclair, Alistair and Vigoda, Eric},
title = {A polynomial-time approximation algorithm for the permanent of a matrix with nonnegative entries},
year = {2004},
issue_date = {July 2004},
publisher = {Association for Computing Machinery},
address = {New York, NY, USA},
volume = {51},
number = {4},
issn = {0004-5411},
url = {https://doi.org/10.1145/1008731.1008738},
doi = {10.1145/1008731.1008738},
journal = {J. ACM},
month = jul,
pages = {671–697},
numpages = {27}
}

@article{wolff1989algorithm,
  title = {Collective Monte Carlo Updating for Spin Systems},
  author = {Wolff, Ulli},
  journal = {Phys. Rev. Lett.},
  volume = {62},
  issue = {4},
  pages = {361--364},
  numpages = {0},
  year = {1989},
  month = {Jan},
  publisher = {American Physical Society},
  doi = {10.1103/PhysRevLett.62.361},
  url = {https://link.aps.org/doi/10.1103/PhysRevLett.62.361}
}

@article{swedsen_wang1987algorithm,
  title = {Nonuniversal critical dynamics in Monte Carlo simulations},
  author = {Swendsen, Robert H. and Wang, Jian-Sheng},
  journal = {Phys. Rev. Lett.},
  volume = {58},
  issue = {2},
  pages = {86--88},
  numpages = {0},
  year = {1987},
  month = {Jan},
  publisher = {American Physical Society},
  doi = {10.1103/PhysRevLett.58.86},
  url = {https://link.aps.org/doi/10.1103/PhysRevLett.58.86}
}

@article{wang_swedsen1990review,
title = {Cluster Monte Carlo algorithms},
journal = {Physica A: Statistical Mechanics and its Applications},
volume = {167},
number = {3},
pages = {565-579},
year = {1990},
issn = {0378-4371},
doi = {https://doi.org/10.1016/0378-4371(90)90275-W},
url = {https://www.sciencedirect.com/science/article/pii/037843719090275W},
author = {Jian-Sheng Wang and Robert H. Swendsen}
}

@article{krauth2003pocket,
  title = {Pocket Monte Carlo algorithm for classical doped dimer models},
  author = {Krauth, Werner and Moessner, R.},
  journal = {Phys. Rev. B},
  volume = {67},
  issue = {6},
  pages = {064503},
  numpages = {6},
  year = {2003},
  month = {Feb},
  publisher = {American Physical Society},
  doi = {10.1103/PhysRevB.67.064503},
  url = {https://link.aps.org/doi/10.1103/PhysRevB.67.064503}
}

@misc{krauth2003review,
Author = {Werner Krauth},
Title = {Cluster Monte Carlo algorithms},
Year = {2003},
Eprint = {arXiv:cond-mat/0311623},
}

@article{grigera2001swap,
  title = {Fast Monte Carlo algorithm for supercooled soft spheres},
  author = {Grigera, Tom\'as S. and Parisi, Giorgio},
  journal = {Phys. Rev. E},
  volume = {63},
  issue = {4},
  pages = {045102},
  numpages = {4},
  year = {2001},
  month = {Mar},
  publisher = {American Physical Society},
  doi = {10.1103/PhysRevE.63.045102},
  url = {https://link.aps.org/doi/10.1103/PhysRevE.63.045102}
}

@article{ninarello2017review,
  title = {Models and Algorithms for the Next Generation of Glass Transition Studies},
  author = {Ninarello, Andrea and Berthier, Ludovic and Coslovich, Daniele},
  journal = {Phys. Rev. X},
  volume = {7},
  issue = {2},
  pages = {021039},
  numpages = {22},
  year = {2017},
  month = {Jun},
  publisher = {American Physical Society},
  doi = {10.1103/PhysRevX.7.021039},
  url = {https://link.aps.org/doi/10.1103/PhysRevX.7.021039}
}

@article{placke2024fracton,
  title = {Ising fracton spin liquid on the honeycomb lattice},
  author = {Placke, Benedikt and Benton, Owen and Moessner, Roderich},
  journal = {Phys. Rev. B},
  volume = {110},
  issue = {2},
  pages = {L020401},
  numpages = {5},
  year = {2024},
  month = {Jul},
  publisher = {American Physical Society},
  doi = {10.1103/PhysRevB.110.L020401},
  url = {https://link.aps.org/doi/10.1103/PhysRevB.110.L020401}
}

@article{syljuaasen2002qmc,
  title = {Quantum Monte Carlo with directed loops},
  author = {Sylju\aa{}sen, Olav F. and Sandvik, Anders W.},
  journal = {Phys. Rev. E},
  volume = {66},
  issue = {4},
  pages = {046701},
  numpages = {28},
  year = {2002},
  month = {Oct},
  publisher = {American Physical Society},
  doi = {10.1103/PhysRevE.66.046701},
  url = {https://link.aps.org/doi/10.1103/PhysRevE.66.046701}
}

@misc{sandvik2003deconfined,
Author = {Anders W. Sandvik},
Title = {Deconfinement and criticality in extended two-dimensional dimer models},
Year = {2003},
Eprint = {arXiv:cond-mat/0312097},
}

@inproceedings{sinclair1992improved,
author = {Sinclair, Alistair},
title = {Improved Bounds for Mixing Rates of Marked Chains and Multicommodity Flow},
year = {1992},
isbn = {3540552847},
publisher = {Springer-Verlag},
address = {Berlin, Heidelberg},
booktitle = {Proceedings of the 1st Latin American Symposium on Theoretical Informatics},
pages = {474–487},
numpages = {14},
series = {LATIN '92}
}

@misc{chen2025faster,
Author = {Xiaoyu Chen and Weiming Feng and Zhe Ju and Tianshun Miao and Yitong Yin and Xinyuan Zhang},
Title = {Faster Mixing of the Jerrum-Sinclair Chain},
Year = {2025},
Eprint = {arXiv:2504.02740},
}

@InProceedings{stefankovi2018torpid,
author="{\v{S}}tefankovi{\v{c}}, Daniel
and Vigoda, Eric
and Wilmes, John",
editor="Bender, Michael A.
and Farach-Colton, Mart{\'i}n
and Mosteiro, Miguel A.",
title="On Counting Perfect Matchings in General Graphs",
booktitle="LATIN 2018: Theoretical Informatics",
year="2018",
publisher="Springer International Publishing",
address="Cham",
pages="873--885",
isbn="978-3-319-77404-6"
}

@article{Barahona1982complexity_spin_glass,
doi = {10.1088/0305-4470/15/10/028},
url = {https://doi.org/10.1088/0305-4470/15/10/028},
year = {1982},
month = {oct},
publisher = {},
volume = {15},
number = {10},
pages = {3241},
author = {F Barahona},
title = {On the computational complexity of Ising spin glass models},
journal = {Journal of Physics A: Mathematical and General}
}

@inproceedings{Broder1986worm,
  title={How hard is it to marry at random? (On the approximation of the permanent)},
  author={Andrei Z. Broder},
  booktitle={Symposium on the Theory of Computing},
  year={1986},
  url={https://api.semanticscholar.org/CorpusID:10144946}
}

@inproceedings{guo2019zeros,
author = {Guo, Heng and Liao, Chao and Lu, Pinyan and Zhang, Chihao},
title = {Zeros of holant problems: locations and algorithms},
year = {2019},
publisher = {Society for Industrial and Applied Mathematics},
address = {USA},
booktitle = {Proceedings of the Thirtieth Annual ACM-SIAM Symposium on Discrete Algorithms},
pages = {2262–2278},
numpages = {17},
location = {San Diego, California},
series = {SODA '19}
}

@article{fradkin2017disorder_ops,
	author = {Fradkin, Eduardo},
	date = {2017/05/01},
	doi = {10.1007/s10955-017-1737-7},
	id = {Fradkin2017},
	isbn = {1572-9613},
	journal = {Journal of Statistical Physics},
	number = {3},
	pages = {427--461},
	title = {Disorder Operators and Their Descendants},
	url = {https://doi.org/10.1007/s10955-017-1737-7},
	volume = {167},
	year = {2017}}

@article{griffiths1969effects,
  title = {Nonanalytic Behavior Above the Critical Point in a Random Ising Ferromagnet},
  author = {Griffiths, Robert B.},
  journal = {Phys. Rev. Lett.},
  volume = {23},
  issue = {1},
  pages = {17--19},
  numpages = {0},
  year = {1969},
  month = {Jul},
  publisher = {American Physical Society},
  doi = {10.1103/PhysRevLett.23.17},
  url = {https://link.aps.org/doi/10.1103/PhysRevLett.23.17}
}

@article{gruzberg2001rbim,
  title = {Random-bond Ising model in two dimensions: The Nishimori line and supersymmetry},
  author = {Gruzberg, Ilya A. and Read, N. and Ludwig, Andreas W. W.},
  journal = {Phys. Rev. B},
  volume = {63},
  issue = {10},
  pages = {104422},
  numpages = {27},
  year = {2001},
  month = {Feb},
  publisher = {American Physical Society},
  doi = {10.1103/PhysRevB.63.104422},
  url = {https://link.aps.org/doi/10.1103/PhysRevB.63.104422}
}

@article{mildenberger2006griffiths,
  title = {Griffiths phase in the thermal quantum Hall effect},
  author = {Mildenberger, A. and Evers, F. and Narayanan, R. and Mirlin, A. D. and Damle, K.},
  journal = {Phys. Rev. B},
  volume = {73},
  issue = {12},
  pages = {121301},
  numpages = {4},
  year = {2006},
  month = {Mar},
  publisher = {American Physical Society},
  doi = {10.1103/PhysRevB.73.121301},
  url = {https://link.aps.org/doi/10.1103/PhysRevB.73.121301}
}

@article{google2025below_threshold,
	author = {Acharya, Rajeev and Abanin, Dmitry A. and Aghababaie-Beni, Laleh and Aleiner, Igor and Andersen, Trond I. and Ansmann, Markus and Arute, Frank and Arya, Kunal and Asfaw, Abraham and Astrakhantsev, Nikita and Atalaya, Juan and Babbush, Ryan and Bacon, Dave and Ballard, Brian and Bardin, Joseph C. and Bausch, Johannes and Bengtsson, Andreas and Bilmes, Alexander and Blackwell, Sam and Boixo, Sergio and Bortoli, Gina and Bourassa, Alexandre and Bovaird, Jenna and Brill, Leon and Broughton, Michael and Browne, David A. and Buchea, Brett and Buckley, Bob B. and Buell, David A. and Burger, Tim and Burkett, Brian and Bushnell, Nicholas and Cabrera, Anthony and Campero, Juan and Chang, Hung-Shen and Chen, Yu and Chen, Zijun and Chiaro, Ben and Chik, Desmond and Chou, Charina and Claes, Jahan and Cleland, Agnetta Y. and Cogan, Josh and Collins, Roberto and Conner, Paul and Courtney, William and Crook, Alexander L. and Curtin, Ben and Das, Sayan and Davies, Alex and De Lorenzo, Laura and Debroy, Dripto M. and Demura, Sean and Devoret, Michel and Di Paolo, Agustin and Donohoe, Paul and Drozdov, Ilya and Dunsworth, Andrew and Earle, Clint and Edlich, Thomas and Eickbusch, Alec and Elbag, Aviv Moshe and Elzouka, Mahmoud and Erickson, Catherine and Faoro, Lara and Farhi, Edward and Ferreira, Vinicius S. and Burgos, Leslie Flores and Forati, Ebrahim and Fowler, Austin G. and Foxen, Brooks and Ganjam, Suhas and Garcia, Gonzalo and Gasca, Robert and Genois, {\'E}lie and Giang, William and Gidney, Craig and Gilboa, Dar and Gosula, Raja and Dau, Alejandro Grajales and Graumann, Dietrich and Greene, Alex and Gross, Jonathan A. and Habegger, Steve and Hall, John and Hamilton, Michael C. and Hansen, Monica and Harrigan, Matthew P. and Harrington, Sean D. and Heras, Francisco J. H. and Heslin, Stephen and Heu, Paula and Higgott, Oscar and Hill, Gordon and Hilton, Jeremy and Holland, George and Hong, Sabrina and Huang, Hsin-Yuan and Huff, Ashley and Huggins, William J. and Ioffe, Lev B. and Isakov, Sergei V. and Iveland, Justin and Jeffrey, Evan and Jiang, Zhang and Jones, Cody and Jordan, Stephen and Joshi, Chaitali and Juhas, Pavol and Kafri, Dvir and Kang, Hui and Karamlou, Amir H. and Kechedzhi, Kostyantyn and Kelly, Julian and Khaire, Trupti and Khattar, Tanuj and Khezri, Mostafa and Kim, Seon and Klimov, Paul V. and Klots, Andrey R. and Kobrin, Bryce and Kohli, Pushmeet and Korotkov, Alexander N. and Kostritsa, Fedor and Kothari, Robin and Kozlovskii, Borislav and Kreikebaum, John Mark and Kurilovich, Vladislav D. and Lacroix, Nathan and Landhuis, David and Lange-Dei, Tiano and Langley, Brandon W. and Laptev, Pavel and Lau, Kim-Ming and Le Guevel, Lo{\"\i}ck and Ledford, Justin and Lee, Joonho and Lee, Kenny and Lensky, Yuri D. and Leon, Shannon and Lester, Brian J. and Li, Wing Yan and Li, Yin and Lill, Alexander T. and Liu, Wayne and Livingston, William P. and Locharla, Aditya and Lucero, Erik and Lundahl, Daniel and Lunt, Aaron and Madhuk, Sid and Malone, Fionn D. and Maloney, Ashley and Mandr{\`a}, Salvatore and Manyika, James and Martin, Leigh S. and Martin, Orion and Martin, Steven and Maxfield, Cameron and McClean, Jarrod R. and McEwen, Matt and Meeks, Seneca and Megrant, Anthony and Mi, Xiao and Miao, Kevin C. and Mieszala, Amanda and Molavi, Reza and Molina, Sebastian and Montazeri, Shirin and Morvan, Alexis and Movassagh, Ramis and Mruczkiewicz, Wojciech and Naaman, Ofer and Neeley, Matthew and Neill, Charles and Nersisyan, Ani and Neven, Hartmut and Newman, Michael and Ng, Jiun How and Nguyen, Anthony and Nguyen, Murray and Ni, Chia-Hung and Niu, Murphy Yuezhen and O'Brien, Thomas E. and Oliver, William D. and Opremcak, Alex and Ottosson, Kristoffer and Petukhov, Andre and Pizzuto, Alex and Platt, John and Potter, Rebecca and Pritchard, Orion and Pryadko, Leonid P. and Quintana, Chris and Ramachandran, Ganesh and Reagor, Matthew J. and Redding, John and Rhodes, David M. and Roberts, Gabrielle and Rosenberg, Eliott and Rosenfeld, Emma and Roushan, Pedram and Rubin, Nicholas C. and Saei, Negar and Sank, Daniel and Sankaragomathi, Kannan and Satzinger, Kevin J. and Schurkus, Henry F. and Schuster, Christopher and Senior, Andrew W. and Shearn, Michael J. and Shorter, Aaron and Shutty, Noah and Shvarts, Vladimir and Singh, Shraddha and Sivak, Volodymyr and Skruzny, Jindra and Small, Spencer and Smelyanskiy, Vadim and Smith, W. Clarke and Somma, Rolando D. and Springer, Sofia and Sterling, George and Strain, Doug and Suchard, Jordan and Szasz, Aaron and Sztein, Alex and Thor, Douglas and Torres, Alfredo and Torunbalci, M. Mert and Vaishnav, Abeer and Vargas, Justin and Vdovichev, Sergey and Vidal, Guifre and Villalonga, Benjamin and Heidweiller, Catherine Vollgraff and Waltman, Steven and Wang, Shannon X. and Ware, Brayden and Weber, Kate and Weidel, Travis and White, Theodore and Wong, Kristi and Woo, Bryan W. K. and Xing, Cheng and Yao, Z. Jamie and Yeh, Ping and Ying, Bicheng and Yoo, Juhwan and Yosri, Noureldin and Young, Grayson and Zalcman, Adam and Zhang, Yaxing and Zhu, Ningfeng and Zobrist, Nicholas and Google Quantum AI and Collaborators},
	date = {2025/02/01},
	doi = {10.1038/s41586-024-08449-y},
	id = {Acharya2025},
	isbn = {1476-4687},
	journal = {Nature},
	number = {8052},
	pages = {920--926},
	title = {Quantum error correction below the surface code threshold},
	url = {https://doi.org/10.1038/s41586-024-08449-y},
	volume = {638},
	year = {2025}}

@ARTICLE{breuckmann2016constructions,
  author={Breuckmann, Nikolas P. and Terhal, Barbara M.},
  journal={IEEE Transactions on Information Theory}, 
  title={Constructions and Noise Threshold of Hyperbolic Surface Codes}, 
  year={2016},
  volume={62},
  number={6},
  pages={3731-3744},
  doi={10.1109/TIT.2016.2555700}}

@book{levin2017markov,
  title={Markov chains and mixing times},
  author={Levin, David A and Peres, Yuval},
  volume={107},
  year={2017},
  publisher={American Mathematical Soc.}
}

@article{schweinsberg2002relaxation_time,
author = {Schweinsberg, Jason},
title = {An O(n2) bound for the relaxation time of a Markov chain on cladograms},
journal = {Random Structures \& Algorithms},
volume = {20},
number = {1},
pages = {59-70},
doi = {https://doi.org/10.1002/rsa.1029},
url = {https://onlinelibrary.wiley.com/doi/abs/10.1002/rsa.1029},
eprint = {https://onlinelibrary.wiley.com/doi/pdf/10.1002/rsa.1029},
year = {2002}
}

@book{grimmett1999percolation,
  title={Percolation},
  author={Grimmett, G.},
  isbn={9783540649021},
  lccn={99021041},
  series={Die Grundlehren der mathematischen Wissenschaften in Einzeldarstellungen},
  url={https://books.google.co.uk/books?id=lTSzVXTOCpIC},
  year={1999},
  publisher={Springer}
}

@book{DemboZeitouni1998LargeDeviations,
  author    = {Dembo, Amir and Zeitouni, Ofer},
  title     = {Large Deviations Techniques and Applications},
  series    = {Stochastic Modelling and Applied Probability},
  volume    = {38},
  publisher = {Springer},
  address   = {Berlin},
  edition   = {2},
  year      = {2009},
  doi       = {10.1007/978-3-642-03311-7}
}

@misc{pandey2026low,
Author = {Akshat Pandey and Aditya Mahadevan and A. Alan Middleton and Daniel S. Fisher},
Title = {Low-temperature transition of 2d random-bond Ising model and quantum infinite randomness},
Year = {2026},
Eprint = {arXiv:2603.02308},
}

@article{varadhan2008ldp,
 author = {Varadhan, S. R. S.},
 title = {Large deviations},
 fjournal = {The Annals of Probability},
 journal = {Ann. Probab.},
 issn = {0091-1798},
 volume = {36},
 number = {2},
 pages = {397--419},
 year = {2008},
 doi = {10.1214/07-AOP348},
 zbMATH = {5264694},
 Zbl = {1146.60003}
}

@Article{burenev2025ldp_intro,
	title={{An introduction to large deviations with applications in physics}},
	author={Ivan N. Burenev and Daniël W. H. Cloete and Vansh Kharbanda and Hugo Touchette},
	journal={SciPost Phys. Lect. Notes},
	pages={104},
	year={2025},
	publisher={SciPost},
	doi={10.21468/SciPostPhysLectNotes.104},
	url={https://scipost.org/10.21468/SciPostPhysLectNotes.104},
}

@misc{temkin2025charge_informed_qec,
Author = {Vlad Temkin and Zack Weinstein and Ruihua Fan and Daniel Podolsky and Ehud Altman},
Title = {Charge-Informed Quantum Error Correction},
Year = {2025},
Eprint = {arXiv:2512.22119},
}

@misc{aitchison2025spacetime_spins,
Author = {Cory T. Aitchison and Benjamin Béri},
Title = {Spacetime Spins: Statistical mechanics for error correction with stabilizer circuits},
Year = {2025},
Eprint = {arXiv:2512.21991},
}

@article{leon2004hoeffding,
author = {Carlos A. Le{\'o}n and Fran{\c{c}}ois Perron},
title = {{Optimal Hoeffding bounds for discrete reversible Markov chains}},
volume = {14},
journal = {The Annals of Applied Probability},
number = {2},
publisher = {Institute of Mathematical Statistics},
pages = {958 -- 970},
year = {2004},
doi = {10.1214/105051604000000170},
URL = {https://doi.org/10.1214/105051604000000170}
}

@article{delfosse2020erasure_ml,
  title = {Linear-time maximum likelihood decoding of surface codes over the quantum erasure channel},
  author = {Delfosse, Nicolas and Z\'emor, Gilles},
  journal = {Phys. Rev. Res.},
  volume = {2},
  issue = {3},
  pages = {033042},
  numpages = {5},
  year = {2020},
  month = {Jul},
  publisher = {American Physical Society},
  doi = {10.1103/PhysRevResearch.2.033042},
  url = {https://link.aps.org/doi/10.1103/PhysRevResearch.2.033042}
}

\appendix

\section{Variations of the Worm Algorithm}\label{App:worm_variants}

The worm algorithm has been studied extensively in the statistical physics literature, resulting in several closely related variants. We briefly describe two such variants below, both of which can be incorporated into the worm decoder framework. In \autoref{alg:worm_decoder}, the \texttt{WormUpdate} step may be implemented using any of these algorithms.

All of these algorithms are based on the same underlying principle: they operate on an enlarged configuration space rather than restricting updates to closed cycles alone. This enlargement enables non-local updates of cycle configurations and is the primary reason for their empirical efficiency.

Despite differences in their update rules, all of these variants define Markov chain Monte Carlo processes whose marginal distribution on closed cycles is identical. As a result, each variant can be used to sample error chains for matchable decoding problems.

In the following, we focus on two specific variants: the directed worm algorithm and the directed dimer worm algorithm.

\subsection{Directed worm algorithm}\label{app:directed_worm}
The directed worm algorithm is a slight variation of the symmetric worm
algorithm introduced in \autoref{alg:wormprocess}. In this variant, one virtual defect
is kept fixed while the other is allowed to move, whereas in the symmetric worm
algorithm both defects are mobile and one is chosen at random at each update
step \cite{Wang_2005}. In our numerical implementation of the worm decoder, presented in \autoref{sec:numerical_results}, we employ this directed variant of the worm algorithm.

Because the two virtual defects play distinct roles in the update steps, the space of configurations explored by the directed worm algorithm differs slightly from that of the symmetric version. Specifically, the configuration space is given by

\begin{equation}
    W^* = \mathcal{C}_0 \cup \mathcal{C}_2^*,
\end{equation}

where \(\mathcal{C}_0\) is the space of cycles defined in \autoref{eq:cycle_defintion}, and \(\mathcal{C}_2^*\) denotes the set of error chains with two open ends, together with a specified choice of \emph{head} and \emph{tail}, corresponding respectively to the mobile and fixed virtual defects. Equivalently, each element of $\mathcal{C}_2$ admits an orientation so $\mathcal{C}_2^* \cong \mathcal{C}_2 \times \{0,1\}$.

The directed worm algorithm is formally described in \autoref{alg:wormprocess_directed}.

When the distinction between head and tail is ignored, the stationary measure of this Markov chain coincides exactly with the Prokof'ev--Svistunov measure introduced in \autoref{eq:PS_measure}.

We now turn to a further variant of the worm approach, which maps the ensemble of cycles onto an ensemble of dimer configurations of a dressed graph.

\begin{algorithm}[H]
\caption{The Directed Worm Algorithm}
\label{alg:wormprocess_directed}
\begin{algorithmic}[1]
\State \textbf{Input:} Weighted graph \(\tilde{G} = (V,E,\tilde{w})\); current loop configuration \(\Sigma\).
\State \textbf{Output:} Updated loop configuration \(\Sigma'\).

\State Choose \(i_0 \in V\) uniformly at random.
\State Set \(i \gets i_0\) and \(\Sigma' \gets \Sigma\).

\Repeat
    \State Choose \(j \in N(i)\) uniformly at random.
    \State Accept with probability
    \[
    a(i\!\to\! j)=\min\!\left(1,\,\frac{d(i)}{d(j)}\,\tilde{w}_{ij}^{\,1-2\Sigma'_{ij}}\right).
    \]
    \If{accepted}
        \State \(\Sigma'_{ij} \gets 1-\Sigma'_{ij}\), \quad \(i \gets j\).
    \EndIf
\Until{\(i = i_0\)}

\State \Return \(\Sigma'\)
\end{algorithmic}
\end{algorithm}

\subsection{Dimer worm algorithm}

We discuss another variant of the worm algorithm, which can also be implemented within the worm decoder. This version of the worm algorithm acts on the space of \emph{dimer configurations}, also known as \emph{perfect matchings} in the graph-theoretic sense. We caution that this use of the term `matching' is distinct from its use in minimum-weight perfect matching (MWPM), where a `matching' refers to a set of error chains consistent with an observed syndrome---despite the fact that MWPM is itself implemented via a perfect matching algorithm applied to the graph of triggered detectors.

Given a graph \( G = (V, E) \), a \emph{dimer} is an edge \( e \in E \) that occupies (covers) the two vertices incident to it. A \emph{dimer configuration} on \( G \) is a subset of edges \( M \subset E \) such that no two edges in \( M \) share a common vertex, and every vertex in \( V \) is incident to exactly one edge in \( M \). 
Each dimer configurations therefore corresponds to a permutation of the vertices.

In a manner entirely analogous to the enlarged state space encountered in cycle-based variants of the worm algorithm, the worm algorithm on dimer configurations also introduces \emph{virtual defects} (monomers). These virtual defects correspond to vertices that are not part of any dimer. The MCMC algorithm introduced in \autoref{alg:wormprocess_dimer} therefore operates on an enlarged state space that includes such configurations.

The relationship between cycles and dimers arises due to the introduction of a \emph{decorated graph}. To this end, one replaces all vertices of the original graph with \emph{gadgets}. It follows that there is an induced $2^{\abs{V}}$-to-one mapping between valid dimer configurations on the decorated graph and cycles on the original graph. We discuss the decoration procedure in more detail in the following.

\subsubsection{Introducing the decorated graph}

By replacing the vertices of a weighted graph \( G = (V, E, w) \) with appropriate gadgets, one can construct a new weighted graph. Here we introduce an adaptation of the Kasteleyn construction \cite{10.1063/1.1703953,10.1063/1.1704825}, following the procedure described below.

Consider a vertex \( v \in V \) of degree (coordination number) \( d_v \) in the original graph. We replace this vertex by a local gadget using the following prescription:

\begin{enumerate}
  \item Construct a regular \( d_v \)-gon and attach an outward-pointing triangle to each of its edges.
  \item Add an external edge to the tip of each triangle.
  \item Assign the original edge weights to the external edges, and set all internal edge weights to \( w = 1 \).
\end{enumerate}

This construction is illustrated in \autoref{fig:dressed_lattice} for the cases \( d_v = 4 \) and \( d_v = 6 \).

\begin{figure}[htbp]
\centering
\includegraphics[width=1.0\linewidth]{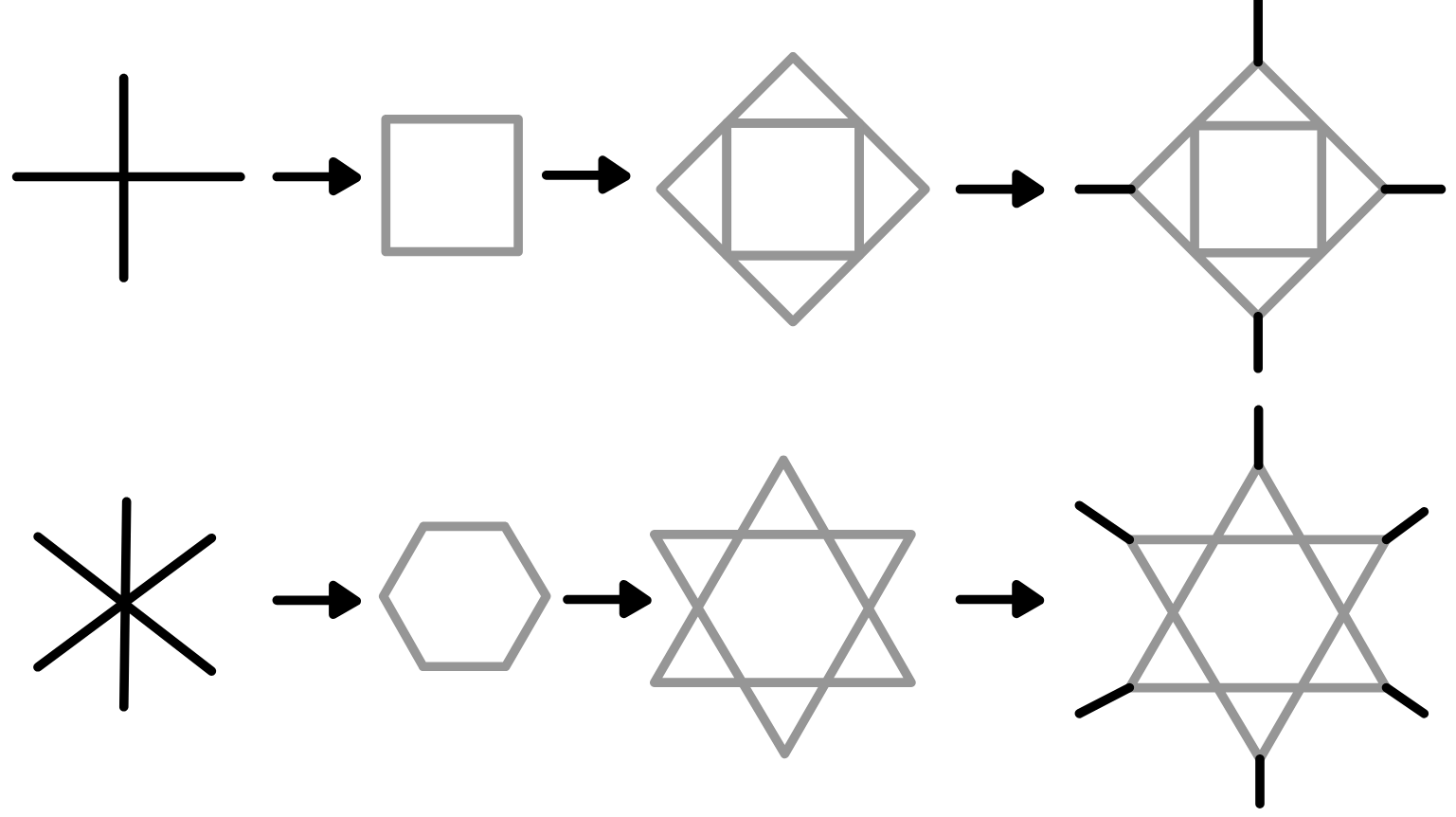}
\caption{Illustration of the gadget construction for vertices with coordination numbers \( d_v = 4 \) and \( d_v = 6 \). The internal edges are shown in gray and have weight $w = 1$.}
\label{fig:dressed_lattice}
\end{figure}

This construction induces a one-to-$2^{\abs{V}}$ mapping between loop configurations on the original graph and dimer configurations on the dressed lattice. In particular, if a dimer occupies an external edge of a gadget, then the corresponding loop configuration on the original graph includes the associated edge. Examples of this mapping for the case \( d_v = 6 \) are shown in \autoref{fig:dimer_mapping}.

\begin{figure}[htbp]
\centering
\includegraphics[width=1.0\linewidth]{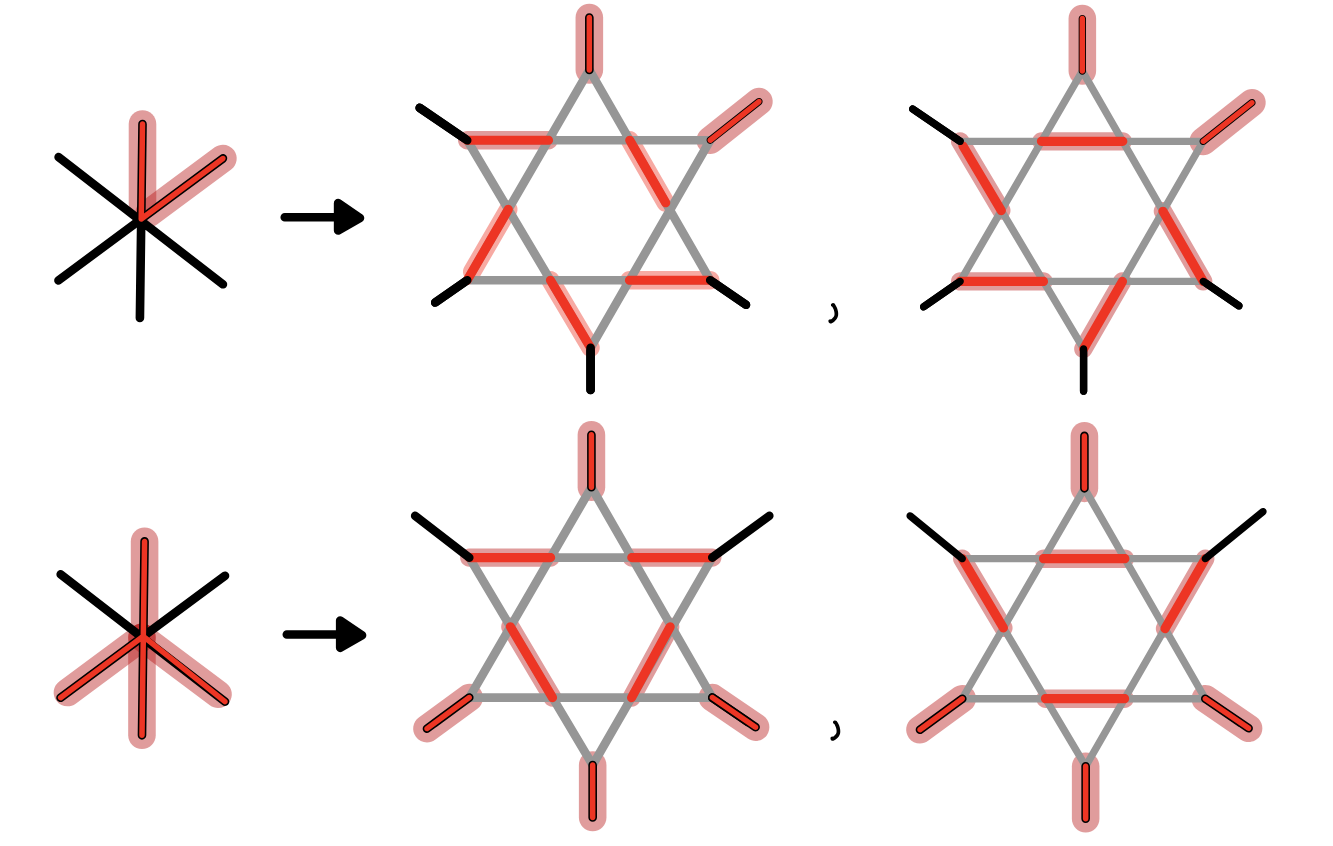}
\caption{Examples illustrating the one-to-two mapping between local edge configurations on a vertex and dimer configurations on the dressed lattice.}
\label{fig:dimer_mapping}
\end{figure}

One can further verify that configurations containing non-cyclic (open) loops on the original graph do not admit a corresponding mapping to valid dimer configurations on the dressed lattice. Additionally, consider the Gibbs distribution over dimer configurations,
\begin{equation}\label{eq:dimer_Gibbs}
    \mathbb{P}_{\text{dimer}}(M) \propto \prod_{e \in M} w_e .
\end{equation}
This distribution coincides with the original Gibbs distribution defined over cycle configurations on the original graph. This follows from the fact that the weights of the external edges of the dressed lattice are unchanged, while all internal edges are assigned unit weight. Consequently, sampling from the Gibbs distribution of dimer configurations on the dressed lattice induces samples
from the corresponding cycle ensemble on the original weighted graph.

This observation naturally motivates the use of dimer-based worm algorithms, which provide a faithful means of sampling from Gibbs distributions over loop configurations on weighted graphs. As a result, such algorithms can be employed to perform optimal decoding for matchable codes.

\subsubsection{The worm algorithm for dimers}
We now introduce a Markov Chain Monte Carlo (MCMC) algorithm \cite{Sandvik_2006}, which when marginalized over perfect matchings, can be used to faithfully sample over the Gibbs distribution given in \autoref{eq:dimer_Gibbs}.

Conceptually, this algorithm shares many similarities with the other variants of the worm algorithm introduced previously. The basic move consists of breaking a dimer into two virtual defects and let (one of) the defects to perform a stochastic random walk until it returns to the other defect and a new valid dimer configuration is formed.

We present the directed dimer worm algorithm in \autoref{alg:wormprocess_dimer}.
Given a configuration of dimers (potentially also with a pair of virtual defects), if a site $v$ is part of a dimer, we denote the site matched to $v$ by $\Sigma_D[v]$.

\begin{algorithm}[H]
\caption{The Directed Dimer Worm Algorithm}
\label{alg:wormprocess_dimer}
\begin{algorithmic}[1]
\State \textbf{Input:} Weighted decorated graph \(\tilde{G}_D = (V,E,\tilde{w})\); current dimer configuration \(\Sigma_D\).
\State \textbf{Output:} Updated dimer configuration \(\Sigma'_D\).

\State Choose a vertex \(i_0 \in V\) uniformly at random.
\State Set \(j \gets \Sigma_D[i_0]\).
\State Set \(\Sigma'_D \gets \Sigma_D\).

\While{\(j \neq i_0\)}
    \State Choose a neighbor \(i_1^\ast \in N(j)\) with probability
    \Statex \hspace{1.2em} proportional to \(\tilde{w}_{j i_1^\ast}\).
    \State Let \(i_2^\ast \gets \Sigma'_D[i_1^\ast]\).
    \State Remove the dimer \((i_1^\ast, i_2^\ast)\) from \(\Sigma'_D\).
    \State Add the dimer \((i_1^\ast, j)\) to \(\Sigma'_D\).
    \State Set \(j \gets i_2^\ast\).
\EndWhile

\State \Return \(\Sigma'_D\)
\end{algorithmic}
\end{algorithm}
An apparent advantage of this algorithm is that it is rejection-free, in contrast to both the symmetric and directed worm algorithms. However, a virtual defect can still become effectively ``stuck'' locally within a gadget of the decorated lattice. This behavior is directly analogous to a rejected move in cycle-based worm algorithms and can lead to similar dynamical slowdowns.

As with the other worm variants, this algorithm can be used in place of line~9 in \autoref{alg:worm_decoder}. This is achieved by mapping a cycle configuration \(\Sigma\) on the original decoding graph to a corresponding dimer configuration~\(\Sigma_D\) on the appropriately dressed graph. A new, independent dimer configuration is then generated using \autoref{alg:wormprocess_dimer}. Finally, to obtain a sample from the distribution of cycles, this updated dimer configuration is mapped back to a cycle on the original graph.

\section{Iterative parsing for correlated decoding}
\label{App:correlated_parsing}

Here, we introduce an iterative parsing procedure for correlated decoding of the surface code with depolarizing noise. The central idea is that correlations between the $X$- and $Z$-decoding problems can be captured more accurately by performing multiple rounds of edge reweighting before the final decoding step.

Concretely, the procedure alternates between the two decoding problems. Initially, the worm algorithm is applied to the $Z$-decoding graph, with edge weights given by $w = p_z/(1-p_z)$, where $p_z = 2p/3$. Using the worm algorithm, empirical edge marginals are estimated and used to reweigh the edges of the $X$-decoding graph. The worm algorithm is subsequently applied to the reweighted $X$-decoding graph, and the resulting marginals are used to update the edge weights of the $Z$-decoding graph. This alternating reweighting procedure may be repeated multiple times before performing the final decoding step.

Numerical results for this iterative parsing scheme are shown in \autoref{fig:parsing_results_depol}. We observe a modest but consistent improvement in decoding success rates when increasing the number of iterations for the surface code with depolarizing noise and perfect measurements.

\begin{figure}[htbp]
    \centering
    \includegraphics[width=\linewidth]{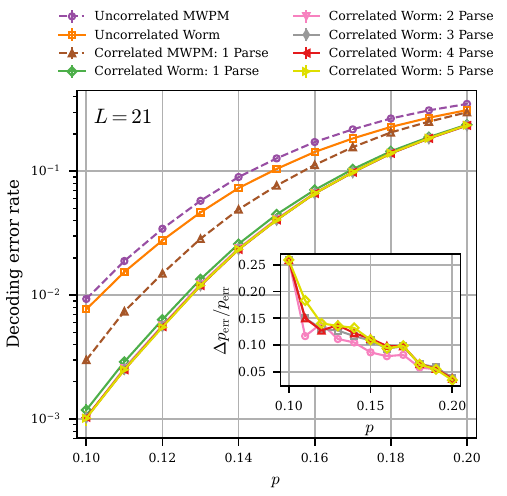}
   \caption{Convergence of the correlated worm decoder under multiple parsing iterations. Results are shown for the rotated surface code with \(L = 21\) under i.i.d.\ depolarizing noise of strength \(p\). The inset shows the relative improvement in decoding success rate when increasing the number of parses from one to multiple.}
    \label{fig:parsing_results_depol}
\end{figure}

An analogous construction can be applied to more general correlated decoding problems (which do not split into two separate matchable DEMs). In this case, multiple rounds of estimating edge marginals, are performed prior to decoding, with each round followed by a reweighting step. This iterative procedure progressively refines the effective correlations between edges before the final decoding step.

\section{MWPM vs.\ MLD for hyperbolic codes and \(\delta\)-hyperbolicity}\label{app:hyperbolicity}

We can understand the similarity in decoding performance between MWPM and MLD for hyperbolic surface codes through a geometric property known as \emph{\(\delta\)-hyperbolicity}.

To understand this property, let us for simplicity consider a connected tree-graph.
If we pick any three vertices and connect them pairwise by geodesics (shortest paths), they will form a ``tripod''-shape where the three paths meet at a common point (unless all three points lie on a single geodesic).
The \(\delta\)-hyperbolicity property generalizes this to more general metric spaces where this property approximately holds.

To this end, let $(X, d)$ be a metric space and let ${x, y, z \in X}$. 
A \emph{geodesic triangle} $\Delta(x,y,z)$ is the union of three geodesic segments $[x,y] \cup [y,z] \cup [z,x]$.
Given $\delta \geq 0$, the triangle $\Delta(x,y,z)$ is called \emph{$\delta$-slim} if each of its sides is contained within the $\delta$-neighborhood of the other two sides. 
More precisely, considering the geodesic $[x,y]$, then for every point $p \in [x,y]$, there exists a point in $[y,z] \cup [z,x]$ whose distance from $p$ is less than $\delta$, i.e.
\[
  [x,y] \subseteq B_\delta\bigl([y,z]\bigr) \cup B_\delta\bigl([z,x]\bigr),
\]
where $B_\delta(S) = \{ p \in X : d(p, S) < \delta \}$ denotes the open $\delta$-neighborhood of a set $S$.
The space $(X, d)$ is said to be \emph{$\delta$-hyperbolic} if every geodesic triangle in $X$ is $\delta$-slim.
Note that for our construction of hyperbolic codes we consider surfaces with Gaussian curvature $\kappa = -1$, which are $\delta$-hyperbolic with $\delta = \ln(1+\sqrt{2})$.

Let $\gamma$ be a path in $X$. 
If $[p, q]$ is a geodesic segment connecting the endpoints of $\gamma$, then for every $x \in [p, q]$
\begin{align}
\label{eq:quasi_geodesic}
d(x, \operatorname{im}(\gamma)) \leq \delta\, |\log_2 \operatorname{length}(\gamma)| + 1.
\end{align}
This can be proved by subdividing the area between the geodesic $[p,q]$ and the curve $\gamma$ into triangles and repeatedly applying the definition of $\delta$-hyperbolicity, see \cite[Chapter III.H, Proposition I.6]{bridson2013metric}.

Intuitively, \autoref{eq:quasi_geodesic} tells us that taking detours is expensive in hyperbolic spaces, as for a path to deviate from a geodesic necessarily requires an exponential increase in the length of the path.
This is in stark contrast to Euclidean space:
we can approximate a geodesic (straight line) by a path $\gamma$ that ``zig-zags'' in two perpendicular directions.
We can then rearrange the segments of this path to deviate from the straight line by $\propto \operatorname{length}(\gamma)$.

Hence, we have that in the hyperbolic setting there are far fewer equivalent paths and the minimum weight error dominates.
More concretely, we have that the probability of an error class $[E]$ given a syndrome $s$ is
\begin{align*}
\mathbb{P}( [E] \mid s) &\propto \sum_{E'\in [ E ]} \left( \frac{p}{1-p} \right)^{|E'|}.
\end{align*}
Now, the number of cycles of length $\ell$ passing through a given edge in the hyperbolic lattice grows only polynomially in $\ell$. 
This follows from \autoref{eq:quasi_geodesic}: 
any cycle of length $\ell$ is contained within a $O(\delta \log \ell)$-neighborhood of any geodesic connecting two of its points, so that cycles are confined to tube-like regions whose volume grows polynomially. 
This is in contrast to Euclidean lattices, where the number of such cycles grows exponentially as $\sim \mu^\ell$ for some connectivity constant $\mu > 1$, see \autoref{fig:euclidean_paths}.

\begin{figure}
  \centering
    \includegraphics[width=0.5\linewidth]{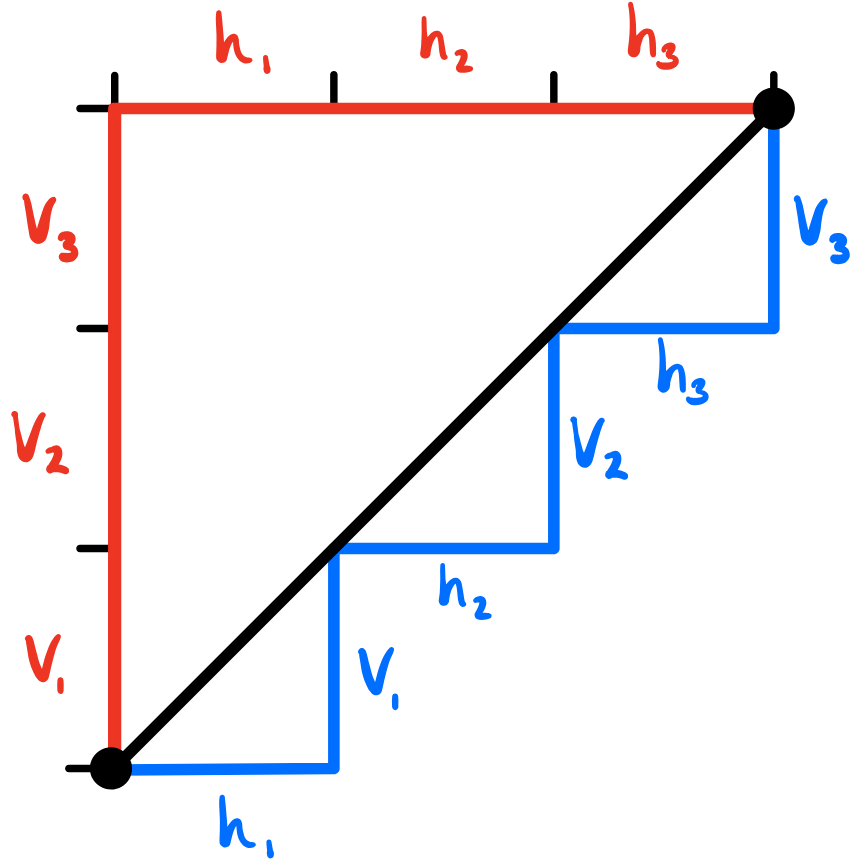}
   \caption{On the square grid there can exist exponentially many shortest paths between two points. Take for example two points that are separated by taking three steps in the vertical direction and then three steps in the horizontal direction. However, as translations commute in Euclidean space, any permutation of these steps gives a shortest path.}
    \label{fig:euclidean_paths}
\end{figure}
As a consequence, in a hyperbolic space we have that
\begin{align*}
    \mathbb{P}([E] \mid s) \propto \left(\frac{p}{1-p}\right)^{|E^*|} \cdot \mathrm{poly}(\abs{V}),
\end{align*}
where $E^*$ denotes the minimum-weight representative in class $[E]$, and the polynomial prefactor accounts for the sub-exponential number of near-minimum-weight configurations. 
Since this prefactor is similar across logical classes, the class maximizing $\mathbb{P}(\cdot \mid s)$ is generically the one containing the minimum-weight error, which is the class selected by MWPM.

\section{From approximate sampling of even subgraphs to approximate counting of perfect matchings}\label{app:matching_reduction}

Remember the even subgraph measure defined in the main text
\begin{equation}
    \lambda_{w, M} = w^{\abs{A \oplus M}}
\end{equation}
where $A\in \mathcal C_0$ are even subgraphs on a graph $G=(V, E)$, $w\in (0, 1)$ and $M\subseteq E$.
In this appendix, we explain that any polynomial-time random approximate sampling scheme for the above measure that works on any graph~$G$, for arbitrary $M$ and arbitrary weights $w$, would yield a fully polynomial time random approximation scheme (FPRAS) for counting the number of perfect matchings in general graphs. 

In the following, we will take $w=1/2$ and $M = E$ but the argument generalizes with minor modifications to the case where sampling is possible for any $w\in (0, 1)$. Now, consider the existence of an FPRAS for the partition function
\begin{align}
    \lambda_{1/2, E}(\mathcal C_0) &= \sum_{A\in \mathcal C_0} \left(\frac{1}{2}\right)^{\abs{A\oplus M}} \\
    &\propto \sum_{A\in \mathcal C_0} 2^{\abs{A}} =: \Tilde\lambda(G) \label{eq:app:even_subg_Z}
\end{align}
on general graphs $G=(V, E)$. In this case, one can show that this yields a FPRAS for counting perfect matchings on general graphs. 

To start, consider the following reduction to graphs of degree three.

\begin{fact}[Lemma 28 in \cite{guo2019zeros}]
Suppose there is an FPRAS for \texttt{\# Perfect-Matchings} on 3-regular graphs. Then there is an FPRAS for \texttt{\#Perfect-Matchings} on general graphs.
\end{fact}

Using this, one can then show that indeed approximate computation of the partition function in \autoref{eq:app:even_subg_Z} enables approximate counting of perfect matchings.

\begin{proposition}[Andreas Galanis, private communication]\label{prop:approx_counting_reduction}
Suppose there is an FPRAS for computing $\Tilde\lambda$ [\autoref{eq:app:even_subg_Z}] on graphs G of maximum degree~$\leq 3$. Then there is an FPRAS for \texttt{\#Perfect-Matchings} on 3-regular graphs.
\end{proposition}

\begin{proof}

Let $G=(V, E)$ be a 3-regular graph with n vertices, and let $\#M$ denote the number of perfect matchings in $G$, and assume without loss of generality that $\# M \geq 1$. 

Consider the graph $G'=(V', E')$ obtained from $G$ by replacing every edge $e\in E$ with a path $P_e$ of $2*\abs{V}$ vertices.
We will show that 
\begin{equation}
    \#M \leq \Tilde \lambda(G') / 2^{2\abs{V}^2} < 1+\#M,
\end{equation} 
so an approximation for $\Tilde \lambda(G')$ allows us to approx \#M.

Indeed any even subgraph $A'$ of $G'$ has the following property: for every edge e of $G$, either none of the edges of the path $P_e$ are included in $A'$, or all of the edges in $P_e$ are included in $A'$. Therefore there is a simple 1-1 correspondence between even subgraphs $A$ of $G$ and even subgraphs $A'$ of $G'$: to get from $A$ to $A'$, for every edge $e$ in $A$, add all edges of $P_e$ into $A'$. The weight of $A$ is $2^|A|$, whereas the weight of $S'$ is $2^{2\abs{V} \abs{A}}$.

Since $G$ is 3-regular, every even subgraph has at most $\abs{V}$ edges; in fact, for an even subgraph $A$ with $\abs{V}$ edges, every vertex in $(V,A)$ must have degree 2, and therefore the set $E\setminus A$ is a perfect matching of $G$. So, $\#M$ is actually equal to number of even subgraphs $A$ with $\abs{A} = \abs{V}$.

Using the mapping from $A$ to $A'$, we conclude that the contribution to $\Tilde\lambda(G')$ of even subgraphs $A$ with $\abs{A} = \abs{V}$ is $\#M \times 2^{2\abs{V}^2}$, whereas the total contribution to $\Tilde\lambda(G')$ of even subgraphs $A$ with $\abs{A} < \abs{V}$ is at most $2^{3\abs{V}/2} \times 2^{2\abs{V}(\abs{V}-1)} \leq 2^{2\abs{V}^2-1}$. Here, $2^{3\abs{V}/2}$ is a (loose) bound on the number of edge subsets $A$.
It follows that 
\begin{equation}
    \#M\, 2^{2\abs{V}^2} \leq \Tilde\lambda(G') < (\#M+1)2^{2\abs{V}^2},
\end{equation}
establishing the claim.

\end{proof}

\begin{widetext}

\section{Mixing time guarantee for the worm process}\label{app:mixing_time_proof}

In this appendix, we provide the formal proof of our mixing time guarantee in theorem \autoref{thm:fast_mixing}.
For completeness, we begin by stating the concrete setup, providing a brief introduction to the formal definitions of mixing time and canonical paths. We then provide the proof, which is a slight generalization of the proof for the worm process for the Ising model presented in Ref. \cite{Collevecchio_2016}.

\subsection{Setup}

Consider a graph, $G = (V, E)$ with maximum degree $\Delta$. Then we will consider the measure
\begin{equation}\label{eq:lambda_measure_app}
    \lambda(A) = w^{\abs{A \oplus M}}
\end{equation}
where $M\subset E$ is some fixed subset of edges, $\oplus$ denotes the symmetric difference of sets, and $w \in (0, 1)$.

The \emph{worm process} is defined on the set of even subgraphs $\mathcal C_0$ and near-even subgraphs $\mathcal C_2$, the latter being the set of subgraphs where exactly two vertices have odd degree. On the space $\mathcal W := \mathcal C_0 \cup \mathcal C_2$, consider the so-called Proko\'ev-Svistunov (PS) measure
\begin{subequations}\label{eq:PS_measure_app}
\begin{equation}
    \pi(A) := \frac{ \Psi(A) \lambda(A)}{\abs{V} \lambda(\mathcal C_0) + 2 \lambda(\mathcal C_2)}
\end{equation}
where
\begin{equation}
    \Psi(A) := \begin{cases}
        \abs{V} & \text{if}~A\in \mathcal C_0 \\
        2 & \text{if}~A\in \mathcal C_2.
    \end{cases}
\end{equation}
\end{subequations}

Now, we define worm process as the Markov chain obtained by Metropolizing the following proposals with respect to the PS measure
\begin{itemize}
    \item if $A \in \mathcal C_0$
    \begin{enumerate}
        \item choose $v\in V$ uniformly at random (u.a.r.)
        \item choose $u\sim v$ u.a.r.
        \item propose $A \to A \oplus uv$
    \end{enumerate}
    \item if $A \in \mathcal C_2$
    \begin{enumerate}
        \item choose $v \in \partial A$ u.a.r.
        \item choose $u\sim v$ u.a.r.
        \item propose $A \to A \oplus uv$
    \end{enumerate}
\end{itemize}
here, $\partial A$ denotes the set of odd vertices in $A$, and $u\sim v$ if $uv\in E$.
One can verify that this is indeed a well-defined process on $\mathcal W$.

In the following, we consider the \emph{lazy} worm process, which is defined as above, but does nothing at every step with probability one-half (this is a technical step to avoid potential issues of periodicity).
The transition matrix of the lazy worm process is then given by

\begin{equation}\label{eq:transition_matrix}
P_w(A, A \oplus uv) := \begin{cases}
    w^{\id[uv \notin (A \oplus M)]} \frac{1}{2\abs{V}}\left( \frac{1}{d(u)} + \frac{1}{d(v)} \right) & A \in \mathcal C_0 \\
    w^{\id[uv \notin (A \oplus M)]} \frac{1}{4}\left( \frac{1}{d(u)} + \frac{1}{d(v)} \right) & A\oplus uv \in \mathcal C_0 \\
    \min \left(\frac{d(u)}{d(v)} w^{\id[uv\notin (A\oplus M)] - \id[uv \in (A\oplus M)]}, 1\right) \frac{1}{4 d(u)} & A \in \mathcal C_2 \wedge A\oplus uv \in \mathcal C_2 \wedge u \in \partial A
\end{cases}. 
\end{equation}
where $\id[X]$ denotes the indicator function for the event $X$, and $d(u)$ is the degree of the vertex $u$ in $G$.

\subsection{Mixing time and canonical paths}

We start by reviewing some general concepts related to the mixing time of Markov Chains.

\subsubsection{Mixing time}

We want to prove that the worm process is fast mixing. To formalize this, for a Markov chain on state space $\Omega \simeq \mathbb{F}_2^n$, with transition matrix $P : \Omega \times \Omega \to [0, 1]$, and unique steady state $\pi$, we define the \emph{mixing time}
\begin{equation}\label{eq:tmix_def_app}
    \tmix(x_0, \epsilon) = \min\{t \in \N : \abs{P^t(x_0) - \pi}_{\rm TV} \leq \epsilon\}
\end{equation}
where $x_0 \in \Omega$ is some initial state, $\epsilon$ is some target accuracy, and $\pi$ denotes the steady-state of the chain $P$.

As indicated, the mixing time in general depends on both the initial state $x_0$ and the target accuracy $\epsilon$. However, there exists a general bound on that dependence of the form \cite[Theorem 12.4]{levin2017markov}
\begin{equation}\label{eq:rel_tmix_def}
    \tmix(x_0, \epsilon) \leq \left[ \log(\pi(x_0)^{-1}) + \log(\epsilon^{-1}) \right] \trel
\end{equation}
where $\trel$ is the \emph{relaxation time}, which is independent of both the initial state and the exact accuracy threshold. Intuitively, it is the time scale on which the chain geometrically contracts towards its stationary distribution. 
We will bound $\trel$ for the worm process below.
We say that a Markov chain is \emph{fast mixing} if $\trel = \textrm{poly}(n)$.

\subsubsection{Canonical Paths}

To show fast mixing for the chain of interest, we will use the method of \emph{canonical paths} \cite{jerrum2003book}. 
On a high level, the idea is to show the absence of bottlenecks in the transition graph of the Markov chain. This can be done by defining (choosing carefully) a set of paths (on the transition graph) between pairs of configurations $I, F \in \Omega$, and showing that no single transition in the chain is used `too much'. 

Formally, let 
\begin{equation}
    \transset := \{(x, y) \in \Omega \times \Omega~;~P(x, y) > 0\}
\end{equation}
the \emph{set of transitions} in the chain $P$. The graph $\transgraph = (\Omega, \transset)$, with vertex set $\Omega$ and edge set $\transset$, is called the transition graph of the chain $P$.

Now, for some non-empty subset of states $\mathscr S \subseteq \Omega$, define a set of \emph{canonical paths} 
\begin{equation}\label{eq:canonical_paths_def}
    \Gamma = \{\gamma_{I,F} : (I, F) \in \Omega \times \mathscr S\}
\end{equation}
where each $\gamma_{I, F}$ is a path from $I$ to $F$ on $\transgraph$. Note that there are many possible choices of a set of canonical paths for a given chain. Remarkably, one can obtain an upper bound on the relaxation time from any such set, in terms of its \emph{congestion}\footnote{Note that we here do not include the length of the paths $\gamma_{I, F}$ in the definition of congestion, but will have the length of the longest path as a separate factor in the resulting mixing time bound in theorem \autoref{thm:canonical_paths}. This is in contrast to some other definitions (see e.g. \cite{jerrum2003book}) but leads to the same bound on the mixing time.}
\begin{equation}\label{eq:congestion_def}
    \varphi(\Gamma) := \max_{t:=(u, v) \in \transset} \biggl\{ 
        \underbrace{\frac{1}{\pi(u)P(u, v)}}_{(\textrm{capacity of $t$})^{-1}} 
        ~
        \underbrace{\sum_{t \in \gamma_{I, F}} \pi(I) \pi(F)}_{\textrm{flow through $t$}} \biggr\}.
\end{equation}
In particular, we will use the following result.

\begin{theorem}[Schweinsberg (2002) \cite{schweinsberg2002relaxation_time}]\label{thm:canonical_paths}

Consider an irreducible and lazy Markov chain with transition matrix $P$, state space $\Omega$, and let $P$ be reversible w.r.t. the distribution $\pi$.
Let $\mathscr S \subseteq \Omega$ be nonempty, let $\Gamma$ be a set of canonical paths from $\mathscr S$ to $\Omega$ [\autoref{eq:canonical_paths_def}]. Let $\varphi(\Gamma)$ be the congestion of $\Gamma$ [\autoref{eq:congestion_def}], and let $\ell_{\rm max}$ be the maximal length of any path in $\Gamma$.
Then the relaxation time of $P$ is bounded by
\begin{equation}
    \trel \leq 4\, \ell_{\rm max} \, \varphi(\Gamma).
\end{equation}
\end{theorem}

With this result, in the canonical paths framework getting a good upper bound on the relaxation time becomes a routing problem, that is choosing a `good' set of canonical paths with low congestion. We will define such a set explicitly for the worm process below.

\subsection{Proof of the mixing time guarantee}

We will bound the mixing time in terms of the \emph{defect susceptibilities} as defined in property \autoref{prop:bounded_sus}. In particular, we will show the following result.

\begin{theorem}\label{thm:fast_mixing_long}
    Consider the worm process for an even subgraph model on $G=(V, E)$ with $(\chi_2, \chi_4)$-bounded defect susceptibility (property \autoref{prop:bounded_sus}).
    There exists a set of canonical paths $\Gamma$ such that $\ell_{\rm max} \leq \abs{E}$, and the congestion is bounded by
    \begin{equation}
        \varphi(\Gamma) \leq 2 \Delta w^{-1} \left[(1 + \chi_2)(2+\abs{V}) + 2\chi_4 \right].
    \end{equation}
    where $\Delta$ is the largest degree of any vertex in $G$. This implies that the relaxation time of the lazy worm is bounded by
    \begin{equation}
        \trel \leq 8 \Delta w^{-1} \abs{E} \left[(1 + \chi_2)(2 + \abs{V}) + 2\chi_4 \right].
    \end{equation}
\end{theorem}

To prove the above, we first show a small number of helpful lemmas. The first one directly follows from the form of the transition matrix in \autoref{eq:transition_matrix}.

\clearpage

\begin{lemma}\label{lem:lemma24modified}
    Let $G = (V, E)$ be a finite graph, with maximum degree $\Delta$, and $w\in(0, 1)$. Then, for all $A\in \mathcal W= \mathcal C_0 \cup \mathcal C_2$ and all $e= \{u,v\} \in E$ 
    \begin{equation}
        \psi(A) \lambda(A) P(A, A\oplus e) \geq \frac{w^{(A \oplus M) \cup e}}{2\Delta}
        \geq \frac{w^{(A \cup e) \oplus M}}{2 \Delta w^{-1}}
    \end{equation}
\end{lemma}

\begin{proof}
We begin by observing that
\[
|A\oplus M|+\mathbf 1(e\notin A\oplus M)=|(A\oplus M)\cup e|.
\]

\medskip
\noindent\textbf{Case 1: $A\in\mathcal C_0$.}
\begin{align*}
\psi(A)\lambda_w(A)P_w(A,A\triangle e)
&=
\Bigl(\abs{V}\,w^{|A\oplus M|}\Bigr)\cdot
\frac{1}{2\abs{V}}\,w^{\mathbf 1(e\notin A\oplus M)}
\Bigl(\frac1{d(u)}+\frac1{d(v)}\Bigr)\\
&=
w^{|(A\oplus M)\cup e|}
\cdot
\frac12\Bigl(\frac1{d(u)}+\frac1{d(v)}\Bigr)\\
&\ge
\frac{1}{\Delta}\,w^{|(A\oplus M)\cup e|}
\ \ge\
\frac{1}{2\Delta}\,w^{|(A\oplus M)\cup e|},
\end{align*}
where we used that $d(u),d(v)\le \Delta$.

\medskip
\noindent\textbf{Case 2: $A\in\mathcal C_2$ and $A\triangle e\in\mathcal C_0$.}
\begin{align*}
\psi(A)\lambda_w(A)P_x(A,A\triangle e)
&=
\bigl(2\,w^{|A\oplus M|}\bigr)\cdot
\frac14\,w^{\mathbf 1(e\notin A\oplus M)}
\Bigl(\frac1{d(u)}+\frac1{d(v)}\Bigr)\\
&=
\frac12\,w^{|(A\oplus M)\cup e|}
\Bigl(\frac1{d(u)}+\frac1{d(v)}\Bigr)\\
&\ge
\frac{1}{\Delta}\,w^{|(A\oplus M)\cup e|}
\ \ge\
\frac{1}{2\Delta }\,w^{|(A\oplus M)\cup e|}.
\end{align*}

\medskip
\noindent\textbf{Case 3: $A\in\mathcal C_2$ and $A\triangle e\in\mathcal C_2$.}
\[
\psi(A)\lambda_w(A)P_w(A,A\triangle e)
=
2\,w^{|A\oplus M|}\cdot
\frac{1}{4d(u)}
\min\!\Bigl(1,\ \frac{d(u)}{d(v)}\,w^{\mathbf 1(e\notin A\oplus M)-\mathbf 1(e\in A\oplus M)}\Bigr).
\]

If the minimum is attained by $1$, then
\[
\psi(A)\lambda_w(A)P_w(A,A\triangle e)
=
\frac{1}{2d(u)}\,w^{|A\oplus M|}
\ \ge\
\frac{1}{2\Delta}\,w^{|(A\oplus M)\cup e|},
\]
where we used that $|A\oplus M|\le |(A\oplus M)\cup e|$ and that $w\in(0,1)$.

If the minimum is attained by the second term, then
\[
\psi(A)\lambda_w(A)P_w(A,A\triangle e)
=
\frac{1}{2d(v)}\,w^{|(A\oplus M)\cup e|-\mathbf 1(e\in A\oplus M)}
\ \ge\
\frac{1}{2\Delta}\,w^{|(A\oplus M)\cup e|},
\]
where we again used $w\in(0,1)$ together with the inequality
\[
|(A\oplus M)\cup e|-\mathbf 1(e\in A\oplus M)\le |(A\oplus M)\cup e|.
\]

\medskip
\noindent\textbf{Putting it together}\par
Combining the three cases, we obtain the uniform bound
\[
\psi(A)\lambda_w(A)P_w(A,A\triangle e)
\ \ge\
\frac{1}{2\Delta}\,w^{|(A\oplus M)\cup e|}.
\]

It remains to relate $|(A\oplus M)\cup e|$ to $|(A\cup e)\oplus M|$.
A direct case analysis (depending on whether $e\in A$ and whether $e\in M$) yields
\[
|(A\oplus M)\cup e|\le |(A\cup e)\oplus M|+1.
\]
Since $w\in(0,1)$, this implies
\[
w^{|(A\oplus M)\cup e|}
\ \ge\
w^{|(A\cup e)\oplus M|+1}
=
w\cdot w^{|(A\cup e)\oplus M|}.
\]
Therefore for all cases we can conclude that,
\[
\psi(A)\lambda_w(A)P_w(A,A\triangle e)
\ \ge\
\frac{1}{2\Delta\,w^{-1}}\;w^{|(A\cup e)\oplus M|},
\]
which is the desired bound.

\end{proof}

In addition to the above lemma, we will use the following few set-theoretic facts.

\begin{claim}\label{claim_one}
    Let $I, F, A$ be finite sets such that $I \cap F \subseteq A \subseteq I \cup F$. Then for any finite set $M$, and sets $I’ = I \oplus M$, etc., we also have $I’ \cap F’ \subseteq A’ \subseteq I’ \cup F’$
\end{claim}

\begin{proof}
We prove the claim pointwise. Fix an element $x$. let $i,f,m,a \in \{0,1\}$ denote membership of $x$ in the sets $I, F, M, A$, respectively.
The assumption is equivalent to
\begin{align}
    (i, f) = (0,0) ~ &\Rightarrow ~ a = 0 \\
    (i, f) = (1,1) ~ &\Rightarrow ~ a = 1
\end{align}    
Now taking symmetric difference with $m$ modifies these numbers by
$i’ = i \oplus m$, $f’ = f \oplus m$, and $a’= a \oplus m$.
which swaps the allowed cases $(i,f,a) = (0,0,0)$ and $(1,1,1)$ when $m=1$ and does nothing when $m=0$.
\end{proof}

\begin{claim}\label{claim_two}
    Let $I, F, B$ be finite sets satisfying $I\cap F \subseteq B \subseteq I \cup F$, and let $U = I \oplus F \oplus B$. Then
    \begin{equation}
        \abs{I} + \abs{F} = \abs{U} + \abs{B}
    \end{equation}
\end{claim}

\begin{proof}
    By assumption it follows that $U \cap B = I \cap F$ and $U \cup B = I \cup F$. Then, by the inclusion-exclusion principle
    \begin{equation}
        \abs{I} + \abs{F} = \abs{I \cup F} + \abs{I \cap F} = \abs{U \cup B} + \abs{U \cap B} = \abs{U} + \abs{B}.
    \end{equation}
\end{proof}

We now prove our main result.

\begin{proof}[Proof of Theorem \ref{thm:fast_mixing_long}]
$ $\newline %

\par\noindent\textbf{The canonical paths}\\[0.5em]
We begin by specifying our choice of canonical paths. We choose the same paths as in \cite{Collevecchio_2016}.

Fix a lexicographical ordering of the vertices of $G$, and denote this by $[n]$. This labeling induces a lexicographical total order on the edges, which in turn induces a lexicographical total order on the set of subgraphs of $G$. For each cycle, we can further fix an orientation, by demanding that from the lowest labeled vertex, we move to the lowest labeled among its neighbors.

In order to move from any state of the worm process $I \in \mathcal W$ to any $F \in \mathcal C_0$, it suffices to update, exactly once, each edge in the symmetric difference $I \oplus F$. Choosing a set of paths doing such a `minimal' update will guarantee that $\ell_{\rm max} \leq \abs{E} \leq \Delta \cdot \abs{V}$ (remember that $\Delta$ is the maximum vertex degree in $G$).

We can decompose the set of edges that need to be processed into components, which we can process one at a time. To this end, note that $\partial (I \oplus F) = \partial I$ (since $F \in \mathcal C_0$), and hence $I\oplus F$ has either none or two odd vertices. 
Further note that if it has two odd vertices (i.e. if $I \in \mathcal C_2$), then the two odd vertices lie in the same connected component of $I \oplus F$, and call these vertices $u$ and $v$. Now, let $B_0$ be the shortest path between $u$ and $v$ in $I \oplus F$. Then, $I \oplus F \setminus B_0 \in \mathcal C_0$, and this set can be decomposed into a disjoint union of cycles in $G$. Decompose $I \oplus F \setminus B_0$, order the cycles using the vertex ordering, and denote the cycles by $B_1, B_2, \dots B_k$. In summary, we obtain a decomposition
\begin{equation}
    I \oplus F = \bigcup_{i = 0}^k B_i
\end{equation}
where $B_0$ is a (possibly empty) path between two vertices, and the $B_i$ for $i > 0$ are cycles in $G$.

Now, define the canonical path $\gamma_{I, F}$ as follows. The initial state is $I$. 
First, process the path $B_0$. Note that if $I \in \mathcal C_0$, $B_0$ is empty which is equivalent to skipping this first step. If $B_0 \neq \emptyset$ begin either at $u$ or $v$ (whichever has the lowest label) and unwind the path $B_0$ (by adding the `next' edge of the path). This inverts the occupation of all edges in $B_0$ relative to $I$, that is the state after this unwinding is $I \oplus B_0$. 
Next, unwind the cycles $B_i$ in order, in each cycle beginning at the lowest labeled edge and proceeding according to its orientation. 

It is easy to see that the set $\Gamma = \{\gamma_{I, F}\,;~ I, F \in \mathcal W \times \mathcal C_0 \}$ is indeed a valid set of canonical paths.

\vspace{1em}
\par\noindent\textbf{The injective map}\\[0.5em]
We proceed by upper-bounding the congestion for the set of canonical paths specified above. Consider the set of paths from $\mathcal C_k$ to $\mathcal C_0$ that use a given transition $t$
\begin{equation}\label{eq:set_of_paths}
    \mathscr P_{t, k} := \{(I, F) \in \mathcal C_k \times \mathcal C_0 \,;~ t \in \gamma_{I, F} \in \Gamma \}
\end{equation}
as well as the total set of paths using a given transition, $\mathscr P_t := \mathscr P_{t, 0} \cup \mathscr P_{t, 2}$.

We follow standard arguments, and introduce an injective map 
\begin{equation}\label{eq:eta_def}
    \eta_t: \mathscr P_t \to \mathcal C_0 \cup \mathcal C_2 \cup \mathcal C_4 \,;
        ~ \eta_{t = (A, A')}(I, F) := I \oplus F \oplus (A \cup A')
\end{equation}

To establish injectivity, we show that $\eta_t(I,F)$ allows the reconstruction of the pair $(I,F)$. Note that $A' = A \oplus e$ for some $e \in E$, and as a result $\eta_t(I,F) = I \oplus F \oplus (A \cup e)$. It follows that we can recover $I \oplus F$ directly via $I \oplus F = \eta_t(I,F) \oplus (A \cup e)$. From $I \oplus F$, we can fully reconstruct the transition path $\gamma_{I,F}$ using the canonical path construction previously outlined. Since each edge in $I \oplus F$ is flipped exactly once along $\gamma_{I,F}$, the transition $t = (A, A \oplus e)$ uniquely identifies a position along $\gamma_{I,F}$, namely the step at which edge $e$ is flipped. We can then unwind the remaining edges in $I \oplus F$ as specified by $\gamma_{I,F}$ to recover $F$. Finally, we recover $I$ via $I = \eta_t(I,F) \oplus (A \cup e) \oplus F$. Since $(I,F)$ is uniquely determined by $\eta_t(I,F)$, we conclude that $\eta_t$ is injective.

We will now show the following inequality, which will allow us to bound the congestion
\begin{equation}\label{eq:eta_bound}
    \frac{1}{\pi(\mathcal C_0)} \frac{\pi(I) \pi(F)}{\pi(A) P(A, A')} \leq
        \frac{2 \Delta w^{-1}}{\lambda(\mathcal C_0)} \psi(I) \lambda(\eta_t(I, F))
\end{equation}
Note that this inequality is slightly modified compared to \cite{Collevecchio_2016}, since we will be using Lemma \autoref{lem:lemma24modified}, our modification of Lemma 2.4 of \cite{Collevecchio_2016}.

Let $t = (A, A\oplus e)$ for some $A \in \mathcal W$ and some $e \in E$. Then, 
\begin{subequations}
\begin{align}
    \frac{1}{\pi(\mathcal C_0)} \frac{\pi(I) \pi(F)}{\pi(A) P(A, A')} 
        &= \frac{1}{\abs{V} \lambda(\mathcal C_0)} \frac{\Psi(I) \lambda(I) \abs{V} \lambda(F)}{\Psi(A) \lambda(A) P(A, A')} \\
        &\leq \frac{1}{\lambda(\mathcal C_0)} \Psi(I) \lambda(I) \lambda(F) \frac{2 \Delta w^{-1}}{w^{(A\cup e)\oplus M}} \\
        & = \frac{2 \Delta w^{-1}}{\lambda(\mathcal C_0)} \Psi(I) w^{\abs{I \oplus M} + \abs{F \oplus M} - \abs{(A\cup e) \oplus M}} \\
        & = \frac{2 \Delta w^{-1}}{\lambda(\mathcal C_0)} \Psi(I) w^{\abs{\eta_t(I, F)\oplus M}} \\
        & = \frac{2 \Delta w^{-1}}{\lambda(\mathcal C_0)} \Psi(I) \lambda(\eta_t(I, F))
\end{align}
\end{subequations}
In the first line, we have simply substituted the definition of the PS measure [\autoref{eq:PS_measure_app}].
In the second line we have used lemma \autoref{lem:lemma24modified}.
In the third line, we have substituted the definition of the measure $\lambda$ in \autoref{eq:lambda_measure_app}. 
In the fourth line, we have used a set-theoretic identity between $I$, $F$, $A \cup e$, and $M$ which we will explain below. 
The final line again is just using the definition of $\lambda$.

Let us derive the fact that
\begin{equation}
    \abs{I\oplus M} + \abs{F\oplus M} - \abs{(A \cup e) \oplus M} = \abs{\eta_t(I, F) \oplus M}
\end{equation}
which we used above. 
Note that by definition, the path $\gamma_{I, F}$ modifies edges in $I \oplus F$. Therefore, no edges in $I \cap F$ are updated when traversing $\gamma_{I, F}$. Conversely, the only edges that are ever updated when traversing $\gamma_{I, F}$ are in $I \cup F$. It follows that $A \cup e$ satisfies the constraint $I \cap F \subseteq A \cup e \subseteq I \cup F$. By claim \autoref{claim_one}, this means that also $I' \cap F' \subseteq (A \cup e)' \subseteq I' \cup F'$, where the $'$ denotes taking the symmetric difference with $M$. Using claim \autoref{claim_two} then yields
\begin{subequations}
\begin{align}
    \abs{I \oplus M} + \abs{F \oplus M} - \abs{(A \cup e) \oplus M}
        &= \abs{(I \oplus M) \oplus (F \oplus M) (\oplus (A \cup e) \oplus M)}\\
        &= \abs{F \oplus M \oplus (A \cup e) \oplus M}\\
        &= \abs{\eta_t(I, F) \oplus M}
\end{align}
\end{subequations}
which establishes \autoref{eq:eta_bound}.

\vspace{1em}
\par\noindent\textbf{The congestion bound}\\[0.5em]
We now use the injective map $\eta_t$, together with \autoref{eq:eta_bound} and our assumption of property \autoref{prop:bounded_sus}, to derive a congestion bound.

Let $t = (A, A\oplus e)$ be a maximally congested transition. Then
\begin{align}
    \varphi(\Gamma) &\leq \sum_{(I,F) \in \mathscr P_t} \frac{2 \Delta w^{-1}}{\lambda (\mathcal C_0)} \Psi(I) \lambda(\eta_t(I, F))\\
        &= \frac{2 \Delta w^{-1}}{\lambda(\mathcal C_0)} 
            \left[ \sum_{(I, F) \in \mathscr P_{t, 0}} \abs{V}\lambda(\eta_t(I, F)) +
                \sum_{(I, F) \in \mathscr P_{t, 2}} 2 \lambda(\eta_t(I, F))
            \right] \\
        & = \frac{2 \Delta w^{-1}}{\lambda(\mathcal C_0)}  \left[
            \abs{V} \lambda(\eta_t(\mathscr P_{t, 0})) 
            + 2 \lambda(\eta_t(\mathscr P_{t, 2}))
        \right],
\end{align}
where in the second line, we used that $\eta_t$ is an injection.

Now, as noted above $\eta_t(\mathscr P_{t, 0}) \subseteq \mathcal W$ and $\eta_t(\mathscr P_{t, 2}) \subseteq \mathcal C_0 \cup \mathcal C_2 \cup \mathcal C_4$, so by a union bound
\begin{align}
    \varphi(\Gamma) 
    &\leq \frac{2 \Delta w^{-1}}{\lambda(\mathcal C_0)} \left[
            \abs{V} \lambda\left( \mathcal C_0 \cup \mathcal C_2 \right) +
            2 \lambda\left( \mathcal C_0 \cup \mathcal C_2 \cup \mathcal C_4 \right)
            \right] \\
    &\leq 2 \Delta w^{-1} \left[ 
            (\abs{V}+2) + 
            (\abs{V}+2)\frac{\lambda(\mathcal C_2)}{\lambda(\mathcal C_0)} + 
            2 \frac{\lambda(\mathcal C_4)}{\lambda(\mathcal C_0)} 
        \right] \\
    &\leq 2 \Delta w^{-1} \left[(1 + \chi_2)(2+\abs{V}) + 2\chi_4 \right]
\end{align}

where the last line follows by property \autoref{prop:bounded_sus}.

Together with the upper bound on the length of any path, the mixing-time bound then follows from Theorem \autoref{thm:canonical_paths}.

\end{proof}

\end{widetext}

\end{document}